\documentclass[12pt]{article} 
\usepackage[sectionbib]{natbib}
\usepackage{array,epsfig,fancyheadings,rotating}
\usepackage[]{hyperref}  
\usepackage{sectsty, secdot}
\sectionfont{\fontsize{12}{14pt plus.8pt minus .6pt}\selectfont}
\renewcommand{\theequation}{\thesection\arabic{equation}}
\subsectionfont{\fontsize{12}{14pt plus.8pt minus .6pt}\selectfont}

\usepackage{amsmath}
\usepackage{amssymb}
\usepackage{amsfonts}
\usepackage{multirow}
\usepackage{amsthm}
\usepackage{xcolor}
\usepackage{amsmath,amsthm}
\usepackage{graphicx,psfrag,epsf,amssymb,xcolor,mathtools,float,subfigure,tabularx,longtable}
\usepackage{enumerate}
\usepackage{natbib}
\usepackage{url}

\newtheorem{theorem}{Theorem}

\newtheorem{assumption}{Assumption}

\theoremstyle{definition}
\newtheorem{definition}{Definition}

\newcommand{\zh}[1]{\textcolor{black}{#1}}
\allowdisplaybreaks
\begin{document}


\renewcommand{\baselinestretch}{2}

\markright{ \hbox{\footnotesize\rm Statistica Sinica
}\hfill\\[-13pt]
\hbox{\footnotesize\rm
}\hfill }

\markboth{\hfill{\footnotesize\rm FIRSTNAME1 LASTNAME1 AND FIRSTNAME2 LASTNAME2} \hfill}
{\hfill {\footnotesize\rm FILL IN A SHORT RUNNING TITLE} \hfill}

\renewcommand{\thefootnote}{}
$\ $\par


\centerline{\large\bf Multilayer Network Regression 
with Eigenvector}
\centerline{\large\bf Centrality and Community Structure}
\vspace{.25cm}
\centerline{Zhuoye Han} 
	\vspace{.1cm} 
	\centerline{\it School of Mathematical Sciences, Fudan University}
	\centerline{Tiandong Wang} 
	\vspace{.1cm} 
	\centerline{\it Shanghai Center for Mathematical Sciences, Fudan University}
	\centerline{\it Shanghai Academy of Artificial Intelligence for Science}
	\vspace{.3cm} 
	\centerline{Zhiliang Ying} 
	\vspace{.1cm} 
	\centerline{\it Department of Statistics, Columbia University}
	\vspace{.3cm} \fontsize{9}{11.5pt plus.8pt minus.6pt}\selectfont




\begin{quotation}
\noindent {\it Abstract:}
In the analysis of complex networks, centrality measures and community structures play pivotal roles. For multilayer networks, a critical challenge lies in effectively integrating information across diverse layers while accounting for the dependence structures both within and between layers. We propose an innovative two-stage regression model for multilayer networks, combining eigenvector centrality and network community structure within fourth-order tensor-like multilayer networks. We develop new community-based centrality measures, integrated into a regression framework. To address the inherent noise in network data, we conduct separate analyses of centrality measures with and without measurement errors and establish consistency for the least squares estimates in the regression model. The proposed methodology is applied to the world input-output dataset, investigating how input-output network data among different countries and industries influence the gross output of each industry.

\vspace{9pt}
\par
\end{quotation}\par

\def\thefigure{\arabic{figure}}
\def\thetable{\arabic{table}}

\renewcommand{\theequation}{\thesection.\arabic{equation}}

\fontsize{12}{14pt plus.8pt minus .6pt}\selectfont
\section{Introduction}
\label{sec:intro}

The inherent heterogeneity of real-world systems has led to multilayer networks as a powerful framework for understanding complex network mechanisms. While single-layer network analysis has yielded significant insights across various domains - from social models examining network effects on advertisement \citep{banerjee2019using} to economic networks studies investigating market structure \citep{allen2019ownership} and currency risk premium \citep{cai2021network, richmond2019trade} - these approaches typically analyze multilayer data in isolation, neglecting potential interlayer dependencies. For instance, conventional single-layer analyses of Twitter data treat friendship, reply, and retweet networks separately, disregarding their inherent interdependencies. In contrast, a multilayer network approach conceptualizes these aspects as interconnected layers, enabling a more comprehensive understanding of both intra- and inter-layer dependence structures.

In this paper, we develop new community-based centrality measures for multilayer networks, integrated into a two-stage regression framework. Centrality measures, which quantify node importance, serve as fundamental explanatory variables in single-layer network analysis \citep{cai2021network, cai2022linear}. Among various proposed measures \citep{jackson2008social, kolaczyk2014statistical}, eigenvector centrality \citep{rowlinson_1996} has emerged as particularly useful. This concept has been extended to multilayer networks, with our work building upon the tensor-based multilayer network structure introduced by \cite{de2013mathematical} and its corresponding eigenvector-like centrality definition presented by \cite{de2013centrality}. Community structure, another fundamental aspect of complex networks, characterizes node similarity through densely connected subgroups. In single-layer networks, community detection has proven invaluable for interpreting complex relational dynamics, with significant methodological advancements \citep{pothen1997graph,newman2004finding,palla2005uncovering,fortunato2010community}. However, existing research has predominantly focused on detection algorithms, with limited integration of community structure with node centrality measures for macroscopic network characterization.

Recent methodological developments in single-layer network analysis provide important foundations for our work. \cite{le2022linear} examine linear regression on multiple network eigenvectors, incorporating measurement error and developing inference methods for hypothesis testing. \cite{cai2022linear} establish theoretical properties for linear regression with multiple centrality measures in sparse single-layer networks, including consistency and distributional theory for the least squares estimators. Furthermore, \cite{cai2021network} propose the SuperCENT framework, which systematically investigates relationships between monoplex network centralities and response variables of interest.

While these single-layer approaches provide valuable insights, they prove inadequate for multilayer network analysis, as they typically examine layers in isolation, disregarding interlayer connections. Our work addresses this limitation by developing a tailored analytical framework that explicitly considers layer interactions. By imposing appropriate constraints on community structure, we obtain estimators with superior asymptotic properties for general multilayer networks.

This paper makes three primary contributions: (1) developing a novel method for extracting community-based node centrality features from tensor-form multilayer networks; (2) establishing relationships between these features and macroscopic node characteristics; and (3) deriving asymptotic properties of parameter estimators in our proposed regression model, accommodating both error-free and measurement-error scenarios. 

The rest of this paper is organized as follows. Section \ref{sec:method} proposes a multilayer regression framework that incorporates community-based centrality scores into tensorial multilayer networks. Section \ref{sec:theory} presents the theoretical properties of the least squares estimators under mild conditions, measurement-error, and error-free scenarios. Section \ref{sec:simu} evaluates the finite-sample performance of our methodology through simulation studies, while Section~\ref{sec:realdata} demonstrates its practical utility through a real-world application. Proofs of all theorems are provided in the supplementary materials.

\par

\section{Methodology}
\label{sec:method}
\subsection{The multilayer networks framework}
\label{subsec:method1}
In this section, we introduce the framework of multilayer networks. We focus on networks where different layers share a common vertex set, as well as their tensor - based representation. This representation captures both intra - and inter - layer connections, which are essential for comprehensive network analysis. Traditional definition of multilayer networks excluding interlayer connections is shown below \citep{tudisco2018node}.
\begin{definition}
	\label{def:multiplex}
 An undirected multiplex network $\mathcal{G}$ with $L$ layers is a collection of $L$ graphs:
	\begin{equation}\label{eq:def_multiplex}
		\mathcal{G}=\left\{G^{(\ell)}\equiv A^{(\ell)}\right\}_{\ell \in \cal L}.
	\end{equation}
Each $G^{(\ell)}$ represents the graph of the $\ell$-th layer network, where $\mathcal{L}=\{ 1, \ldots, L\}$ denotes the set of layers sharing the same node set $\mathcal{V}=\{1, 2, \ldots, N\}$. For each $\ell \in \cal L$, $G^{(\ell)}$ is represented by a symmetric adjacency matrix $A^{(\ell)}=\left(A_{i j}^{(\ell)}\right) \in \mathbb{R}^{N \times N}$ with nonnegative entries. The entry $A_{i j}^{(\ell)} = \omega_{ij}^{\ell}$ if an edge exists between node $i$ and node $j$ in layer $\ell$ with weight $\omega_{ij}^{\ell}$, otherwise $A_{i j}^{(\ell)} = 0$. The network is considered unweighted if all edges have identical weights.
\end{definition}
While Definition~\ref{def:multiplex} provides a matrix-based representation, it fails to capture interlayer connections crucial for multilayer network analysis. These connections have significant practical implications across various domains. In economic input-output networks, nodes represent economic sectors, layers correspond to countries, and interlayer connections capture international trade relationships. In social networks, nodes represent users, layers represent platforms (e.g., Facebook, Instagram, Twitter), and interlayer connections reflect cross-platform interactions. Similarly, in transportation networks, nodes represent hubs, layers represent transportation modes (e.g., buses, subways, trains), and interlayer connections quantify transfer intensities between modes. To adequately model these interactions, we employ a tensor representation (Definition~\ref{def:tensor}) that captures both intra- and inter-layer connections \cite{de2013mathematical}.
\begin{definition}\label{def:tensor}
	Let $\mathcal{B}=\left(\mathcal{B}_{i \ell j \ell^{\prime}}\right) \in \mathbb{R}^{N \times L \times N \times L}$ denote the fourth-order adjacency tensor of an undirected weighted multilayer network, where each element is defined as
	\begin{align}
		\mathcal{B}_{i \ell j \ell^{\prime}}=\left\{\begin{array}{l}
			\omega_{i \ell j \ell^{\prime}}, \text { if } (v_i^\ell, v_j^{\ell^{\prime}}) \in E \\
			0, \text { otherwise, }
		\end{array}\right.\label{eq:def_TensorB}
	\end{align}
	with $ i, j \in \cal V, \ell, \ell^{\prime} \in \cal L$, $v_i^\ell$ representing node $i$ in layer $\ell$, $E$ denoting the edge set, and $\omega_{i \ell j \ell^{\prime}}$ representing the edge weight between node $i$ in layer $\ell$ and node $j$ in layer $\ell^{\prime}$. 
\end{definition}
For multiplex networks with layers connected solely through shared nodes, we have:
\begin{align}
	\omega_{i \ell j \ell^{\prime}}=\left\{\begin{array}{l}
		1, \text { if } i = j, \ell \not= \ell^{\prime}\\
		0, \text { otherwise. }
	\end{array}\right.\label{eq:def_TensorB_1}
\end{align}
Following \cite{de2013mathematical}, we construct the supra-adjacency matrix associated with the multilayer adjacency tensor $\mathcal{B}$ as an $N L \times N L$ block matrix:
\begin{align}
	B_0=\left[\begin{array}{cccc}
		A^{(1)} & D^{(1,2)} & \cdots & D^{(1, L)} \\
		D^{(2,1)} & A^{(2)} & \ddots & \vdots \\
		\vdots & \ddots & \ddots & D^{(L-1, L)} \\
		D^{(L, 1)} & \cdots & D^{(L, L-1)} & A^{(L)}
	\end{array}\right]
	\label{eq:supra},
\end{align}
where $A^{(\ell)}$ are weighted symmetric adjacency matrices for each layer, and $D^{(\ell, \ell^{\prime})}$ represent interlayer adjacency matrices. For multiplex networks, the supra-adjacency matrix simplifies to:
\begin{align}
	B_0=\left[\begin{array}{cccc}
		A^{(1)} & I & \cdots & I \\
		I & A^{(2)} & \ddots & \vdots \\
		\vdots & \ddots & \ddots & I \\
		I & \cdots & I & A^{(L)}
	\end{array}\right].
	\label{eq:multiplex}
\end{align}
\subsection{Eigenvector centrality with community structure}
\label{sec: 2.2}
In this section, we introduce the construction of eigenvector-like centrality and community-based eigenvector-like centrality.

\subsubsection{Construction of eigenvector-like centrality for multilayer networks}
Building upon \cite{bonacich1972factoring}'s definition of eigenvector centrality for single-layer networks, we extend this concept to multilayer networks using the tensor-based framework introduced by \cite{de2013centrality,de2015ranking}. For a single-layer network $A = (a_{i j})$, the centrality $c(i)$ of node $i$ satisfies:
$$
\lambda c(i)=\sum_{j=1}^n a_{i j} c(j) \quad \forall i = 1, \cdots, N.
$$
This leads to the eigenvector equation:
\begin{equation}\label{eq:eigen}
	A c=\lambda c,
\end{equation}
where $c$ is the eigenvector corresponding to the largest eigenvalue of $A$. As $A$ is the adjacency matrix of an undirected (connected) graph, $A$ is positive semi-definite and due to the Perron-Frobenius theorem \citep{keener1993perron}, there exists an eigenvector of the maximal eigenvalue with only nonnegative (positive) entries.

For tensor-based multilayer networks, we define the eigenvector centrality through the matrix $V = (V_{i \alpha}) \in \mathbb{R}^{N\times L}$ satisfying:
\begin{equation}\label{eq1}
	\sum_{j=1}^n \sum_{\beta=1}^L \mathcal{B}_{i \alpha j \beta} V_{j \beta}=\lambda_1 V_{i \alpha},
\end{equation}
for $ i \in \cal V$ and $\alpha \in \cal L$. Here \(\cal V\) denotes the set of vertices and \(\cal L\) denotes the layer set. Then the eigenvector centrality of node $i$ is defined as the $i$-th row of the matrix $V$. Here in multilayer networks, the centrality of a node is determined by summing the centralities of its neighbors within the same layer and across different layers, weighted by the strength of both intralayer and interlayer connections. Using the supra-adjacency matrix from \eqref{eq:supra}, this becomes:
\begin{equation}\label{eq2}
	B_0\operatorname{vec}(V)=\lambda_1 \operatorname{vec}(V).
\end{equation}
Here $\lambda_1$ in \eqref{eq1} and \eqref{eq2} is defined as the spectral radius of positive semi-definite matrix $B_0$. The associated eigenvector is $\operatorname{vec}(V)$, where $\operatorname{vec}(\cdot)$ denotes the standard vectorization operator. By the Perron-Frobenius Theorem, the leading eigenvector is the unique eigenvector that can be chosen so that every entry is non-negative and $\|V\|_F = 1$, motivating its use as a centrality measure. In the following regression step, we consider covariate $C$ such that 
\begin{equation}
	C = a_N V ,   \label{model0}
\end{equation}
where \(a_N\) is a scaling parameter related to \(N\). While eigenvector centrality is often normalized to length 1 \citep{banerjee2013diffusion, cruz2017politician}, some studies treat \(a_N\) as a free parameter \citep{cai2021network, cai2022linear}. Following \cite{cai2022linear}, we leave \(a_N\) as a free parameter to thoroughly examine its impact on regression properties.  

The definition of such an eigenvector-like centrality measure has unique significance in real networks. For instance, in economic input-output networks, eigenvector-like centrality captures the importance of sectors by considering both domestic and international input-output relationships, thereby revealing global economic dynamics. In multilayer social networks, eigenvector-like centrality measures an individual's influence across multiple platforms by accounting for interactions both within and between platforms. Similarly, in multilayer transportation networks, eigenvector-like centrality assesses the significance of transportation hubs by evaluating their intra-modal and inter-modal connectivities.

Moreover, To account for real-world measurement noise, we model the observed network $B$ as:
\begin{align}
	B = B_0 + E_0,
	\label{eq:noisemodel}
\end{align}
where $E_0$ represents symmetric random noise and $B_0$ denotes the real undirected network structure. Also, we assume $E_0$ as a symmetric random matrix since in what follows, we only consider undirected weighted graphs with undirected interlayer interactions.
\subsubsection{Community structure and community-based centrality}
\label{sec: com_struc}
Assuming a shared community structure across all layers, we define community assignments \(c_i \in \{1, \dots, R\}\) for each node \(i \in \mathcal{V}\), with \(R\) communities of sizes \(N_r\) (\(\sum_{r=1}^R N_r = N\)). Following \cite{white2005spectral}, we construct an \(N \times R\) community matrix \(S\), where \(S_{ij} = 1\) if vertex \(i\) belongs to community \(j\) and \(0\) otherwise. The columns of \(S\) are mutually orthogonal, each row sums to 1, and \(S\) satisfies \(\operatorname{tr}(S^\top S) = N\) and \(S^\top S = \operatorname{diag}\{N_1, N_2, \dots, N_R\}\).  

With community information in hand, to construct community-based centrality, we use a row-wise Kronecker product called \emph{face-splitting product} \citep{slyusar1999family} in the construction of community-based centrality, which is also known as the transposed Khatri-Rao product.  Specifically, given matrices $\mathbf{A} = \left[\mathbf{a}_1^\top  \quad \mathbf{a}_2^\top \quad \cdots \quad \mathbf{a}_K^\top \right]^\top \in \mathbb{R}^{K \times I}$ and $\mathbf{B} = \left[\mathbf{b}_1^\top  \quad \mathbf{b}_2^\top \quad \cdots \quad \mathbf{b}_K^\top \right]^\top \in \mathbb{R}^{K \times J}$, their transposed Khatri-Rao product is defined as
$$
\mathbf{A} \bullet \mathbf{B}=\left[\begin{array}{c}
	\mathbf{a}_1 \otimes \mathbf{b}_1 \\
	\mathbf{a}_2 \otimes \mathbf{b}_2 \\
	\cdots \\
	\mathbf{a}_K \otimes \mathbf{b}_K
\end{array}\right] .
$$
Then we construct the community-based eigenvector centrality $U$ as 
\begin{align}
	U := S(H \bullet S^\top  C ),
	\label{def:U}
\end{align}
where $U = (u_{ij}) \in \mathbb{R}^{N\times L}$ with $u_{ij}$ being the mean of centrality in cluster $c_i$ of layer $j$ and $H = \left[\frac{1}{N_1}, \cdots, \frac{1}{N_R}\right]^\top$ with $N_r$ being the size of community $r$.

To address rank deficiency when $L > R$, we define:
\begin{align}
	Z := \frac{1}{L}U\operatorname{1}_{L},
	\label{def:Z}
\end{align}
which represents the row averages of $U$. The entrywise representation of \eqref{def:Z} is \(Z_i=\frac{1}{L} \frac{1}{N_{c_i}} \sum_{l \leq L} \sum_{j \mid c_j=c_i} a_N V_{j l}, i = 1, \cdots, N\). In the following regression setup, we use $Z \in \mathbb{R}^{N\times 1}$ to denote the community-based centrality of each node rather than $U$. Also, nodes belonging to the same community share the same community-based centrality in $Z$.

\subsection{Multilayer network regression}
With the definitions in Section \ref{sec: 2.2}, we propose a Centrality-based Multilayer Network Regression (C-MNetR) model that incorporates eigenvector-like centrality measures. Additionally, incorporating extra community information, we introduce another regression model, referred to as Centrality- and Community-based Multilayer Network Regression (CC-MNetR).

Let $X$ denote the covariate matrix of size $N\times P$, where $P$ is the number of covariates. The Centrality-based Multilayer Network Regression {(C-MNetR)} model is specified by
\begin{align}\label{model1}
	y&=X \beta_X+C \beta_C+\varepsilon,\\
	\text{where}\qquad
	&B_0\operatorname{vec}(V) =\lambda_1 \operatorname{vec}(V), \quad C = a_N V.\nonumber
\end{align}
Here, $\beta_X$ and $\beta_C$ are the regression parameters that can be estimated by the ordinary least squares (OLS) method. When there is no measurement error, $C$ can be uniquely constructed with observations $\{B_0, X, y\}$. The OLS estimator of $\beta = (\beta_X^\top, \beta_C^\top)^\top$ takes for
\begin{align}\label{eq:est1_ols}
	\hat{\beta}^{(ols)} = (W_1^\top W_1)^{-1}W_1^\top y
\end{align}
where $W_1 = (X, C)$. 

When there is measurement error, $C$ is not directly observed but can be estimated from the observations $\{B, X, y\}$. Specifically, we first compute the top eigenvector of $B$, denoted as $\operatorname{vec}(\hat{V})$, and subsequently derive $\hat{C} = a_N \hat{V}$. With $\hat{W}_1 = (X,\hat{C})$, the OLS estimator takes for 
\begin{equation}
	\hat{\beta} = (\hat{W}_1^\top \hat{W}_1)^{-1}\hat{W}_1^\top y.
	\label{eq:est1}
\end{equation}

The basic idea of C-MNetR is to regress directly on the centrality of nodes across different layers. However, in Section \ref{sec:theory}, we find that the asymptotic properties of the OLS estimators $\hat{\beta}_C^{(ols)}$ depend on the order of $a_N$ when measurement error is absent. When measurement error is present, the performance of C-MNetR deteriorates further. A common practice in centrality regression problems \citep{le2022linear} is to set $a_N = \sqrt{N}$, which ensures the consistency of $\hat{\beta}_X$ but leads to inconsistency in $\hat{\beta}_C$. Conversely, increasing the order of $a_N$ to ensure the consistency of $\hat{\beta}_C$ reintroduces bias in $\hat{\beta}_X$. This trade-off highlights a fundamental limitation of the C-MNetR framework. 

To avoid this shortcoming, we propose the Centrality- and Community-based Multilayer Network Regression ({CC-MNetR}) model, which is based on community-based centrality. CC-MNetR effectively incorporates information about the network community structure and demonstrates improved statistical properties of the corresponding OLS estimators after imposing some mild restrictions on the community structure, regardless of the presence of measurement errors. We define the CC-MNetR model as follows:
\begin{align}
	y&=X \beta_X + Z \beta_Z + \varepsilon,\label{model2}\\
	\text{where} \qquad
	&B_0\operatorname{vec}(V) =\lambda_1 \operatorname{vec}(V), \quad C = a_N V, \nonumber\\
	&U = S(H \bullet S^\top  C ), \quad Z = \frac{1}{L}U\operatorname{1}_{L}.\nonumber
\end{align}
For CC-MNetR, in the absence of measurement error, the OLS estimator of coefficient $\beta = (\beta_X^\top, \beta_Z^\top)^\top$ is
\begin{align}\label{eq:est2_ols}
	\tilde{\beta}^{(ols)} = (W_2^\top W_2)^{-1}W_2^\top y
\end{align}
where $W_2 = (X, Z)$ with observations $\{B_0, S, H, X, y\}$.

When there is measurement error, $C$ can be estimated from $\{B, S, H, X, y\}$ similarly. Specifically, we compute the top eigenvector of $B$, \(\operatorname{vec}(\hat{V})\), and derive:
\begin{align}\label{eq: 2.18}
	\hat{C} = a_N \hat{V}, \quad \hat{U} = S(H \bullet S^\top \hat{C}), \quad \hat{Z} = \frac{1}{L}\hat{U} \operatorname{1}_{L}.   
\end{align}
Let \(\tilde{W}_2 = (X, \hat{Z})\) and
\begin{equation}
	\tilde{\beta} = (\tilde{W}_2^\top \tilde{W}_2)^{-1}\tilde{W}_2^\top y.
	\label{eq:est2}
\end{equation}
We show in Section~\ref{sec:theory} that CC-MNetR demonstrates superior statistical properties by effectively incorporating community structure information, as detailed in Section \ref{sec:theory}.
\section{Theoretical results}\label{sec:theory}
\subsection{Model assumptions}
Let $\lambda_1 \geq \lambda_2 \geq \cdots \geq \lambda_N$ be eigenvalues of the noiseless network $B_0$. Define projection matrices $P_X := X(X^\top X)^{-1}X^\top $ and let $\sigma_{\min}((I_N - P_X)V)$ be the smallest singular value of $(I_N - P_X)V$. For a vector $\nu \equiv (\nu_j)_{j = 1}^N$, its $\ell_1$- and $\ell_2$-norms are defined as $\| \nu \|_1 := \sum_{j = 1}^N |\nu_j|$ and $\| \nu\|_2 := \left(\sum_{j = 1}^N \nu_j^2\right)^{\frac{1}{2}}$, respectively. Additionally, the matrix operator norm of matrix \(E_0\) is defined as \(\| E_0 \|_2 := \max_{\| u \|_2 \leq 1} \| E_0 u \|_2\). Let $C_i$ denotes the centrality of node $i$ in all $L$ layers with $C = \left[C_1^\top , \cdots, C_N^\top\right]^\top$. Below are some crucial model assumptions to establish important theoretical properties of the least square estimators.
\begin{assumption}
	\label{ass: noise}
	\zh{\textbf{(Noise Structure)} The regression noises, $\varepsilon_i$ $i\ge 1$, are i.i.d random variables with $\mathbb{E}\left[\varepsilon_i |X, C, E_0\right]=0$, $\text{Var}(\varepsilon_i|X, C, E_0) = \sigma_y^2$, and $\mathbb{E}[\varepsilon_i^4|X, C, E_0] \leq k_0 < \infty$. The network noise $E_0$ satisfies $\mathbb{E}\left[\|E_0\|_2^2\right] = O(N)$. }
\end{assumption}
We give examples satisfying Assumption~\ref{ass: noise}. Let $E_0$ be a block diagonal matrix with blocks $E_{01}, E_{02}, \cdots$, $ E_{0L}$ on the diagonal, where $E_{0i}$, $i = 1, \cdots, L$, are symmetric $N \times N$ matrices with upper-diagonal entries being i.i.d  $\mathcal{N}(0, \sigma_b^2)$. Then from Lemmas 3.7, 3.8, and 3.11 in \cite{van2017structured}, such structure ensures $\mathbb{E}\left[\|E_0\|_2^2\right] = O(N)$. If instead $E_0$ is a symmetric matrix with i.i.d. standard Gaussian entries on and above the diagonal, then Lemma 3.11 in \cite{van2017structured} gives $\mathbb{E}\left[\|E_0\|_2^2\right] = O(NL)$ for a fixed integer $L$.
\begin{assumption}
	\label{ass: design}
	\textbf{(Design Matrix)} 
		The design matrix $X \in \mathbb{R}^{N\times P}$ satisfies $N > P+L$ for fixed $P$ and $L$. For any $i,j,k,l$ (with $i \not= k$ or $j \not= l$), given $C$ and $E_0$, $X_{ij}$ and $X_{kl}$ are independent. For fixed j, $\{X_{ij}\}$ are i.i.d. with $\mathbb{E}[X_{ij} | C, E_0] = 0$, $\mathbb{E}[X_{ij}^2 | C, E_0] < \infty$. 
\end{assumption}
Note that from Assumption \ref{ass: design}, $\frac{1}{N}X^\top X \stackrel{P}{\longrightarrow} V_X$ where $V_X$ is a nonsingular diagonal matrix.

\begin{assumption}
    \textbf{(Identifiability)} There exists $l_N > 0$ such that 
                \begin{equation}
			\sigma_{\min}((I_N - P_X)V) \ge l_N >0.
			\label{con_3.1}
		\end{equation}
Also, $\hat{W}_1 = (X, \hat{C})$ is column full rank.
\label{ass: iden}
\end{assumption}

Condition~\eqref{con_3.1} in Assumption~\ref{ass: iden} implies $W_1 = (X, C)$ is column full rank (proof details are given in the supplement). Note that multilayer networks without interlayer relationships violate Assumption \ref{ass: iden}, as the top eigenvector of a block diagonal matrix \(B = \text{diag}\{A_1, \dots, A_L\}\) can be extended by the top eigenvector of some block \(A_\ell\), resulting in eigenvector centrality \(C\) with columns of zeros, making \(C\) not column full rank. However, multiplexes, as described in \eqref{eq:multiplex}, provide a simple case satisfying Assumption \ref{ass: iden}. 

Next, we impose constraints on the community structure, which is a key factor that allows us to obtain Theorem~\ref{thm:asy3}.
\begin{assumption}
	\label{ass: centrality}
	\textbf{(Centrality and Community Structure)} 
		The community sizes $\{N_i\}_{i = 1}^R$ are independent from $X$, and with probability 1, $\min_i \frac{N_i}{N} > \epsilon$ for some $\epsilon > 0$. Furthermore, \(\underset{1\le i \le N}{\min} \|C_i\|_1^2 \asymp \frac{a_N^2}{N}\).
\end{assumption}
Assumption~\ref{ass: centrality} restricts the minimum sum of centrality values for each node across all layers, excluding nodes with low centrality in all layers. This ensures that the analysis focuses on nodes with significant influence or importance in the network structure, regardless of their specific layer.
\begin{assumption}
	\label{ass: spectral gap}
	\textbf{(Spectral Gap)} The spectral gap, $\lambda_1-\lambda_2$, needs to be large enough so that $\frac{a_N}{\lambda_1-\lambda_2} \rightarrow 0$, as $N \rightarrow \infty$.
\end{assumption}
To ensure Assumption \ref{ass: spectral gap} holds, the network matrix \(B_0\) must exhibit a sufficiently large spectral gap \(\lambda_1 - \lambda_2\), reflecting significant differences in connection strengths among the blocks of its partitioned structure. Since different layers share the same community structure, this requires substantial variation in connection strengths across layers.

In regression analysis, Assumptions \ref{ass: noise}-\ref{ass: centrality} suffice for consistency in the absence of measurement errors, but Assumption \ref{ass: spectral gap} becomes necessary when measurement errors are introduced.

\subsection{Centrality-based regression without measurement error}
We first consider the case when the true network $B_0$, is observed. Under the C-MNetR framework, Theorem \ref{thm:asy1} shows the property of the OLS estimator under mild conditions.
\begin{theorem} \label{thm:asy1}
    Recall the estimator $\hat{\beta}^{(ols)} := (\hat{\beta}_X^{(ols)^\top }$, $\hat{\beta}_C^{(ols)^\top })^\top$ defined in \eqref{eq:est1_ols}. Under Assumptions \ref{ass: noise}, \ref{ass: design} and \ref{ass: iden}, we have:
	\begin{enumerate}[(i)]
		\item As $N\to\infty$,
		$\sqrt{N}(\hat{\beta}_X^{(ols)} - \beta_X)  \stackrel{d}{\longrightarrow} \mathcal{N}(0, \sigma_y^2 V_X^{-1})$.
		\item If $a_N l_N \to \infty$ as $N \to \infty$, then $\hat{\beta}_C^{(ols)} \stackrel{L_2}{\longrightarrow} \beta_C$ and $\hat{\beta}_C^{(ols)} \stackrel{P}{\longrightarrow} \beta_C$. 
	\end{enumerate}
\end{theorem}
When \(X = 0\) and \(L = 1\), our model reduces to the single-layer centrality regression in \cite{cai2022linear}, where \(P_X = 0\) and \eqref{con_3.1} simplifies to \(\sigma_{\min}(V) \ge l_N\). For \(V\) normalized with \(\sigma_{\min}(V) = 1\), setting \(l_N = 1\) recovers their condition \(a_N \to \infty\). When \(X\) is present, the scaling \(a_N l_N \to \infty\) in our framework generalizes this requirement to high-dimensional settings where \(l_N\) may decay with \(N\).

For CC-MNetR, we consider the community-based centrality $Z$ without measurement error. Theorem \ref{thm:asyZ} gives the properties of $\tilde{\beta}^{(ols)}$.
\begin{theorem}\label{thm:asyZ}
	Recall the estimator $\tilde{\beta}^{(ols)} = (\tilde{\beta}_X^{(ols)^\top }$, $\tilde{\beta}_Z^{(ols)^\top })^\top $ defined in \eqref{eq:est2_ols}, Under Assumptions \ref{ass: noise}, \ref{ass: design} and \ref{ass: iden}, we have the following:
	\begin{enumerate}[(i)]
		\item 
		As $N\to\infty$,
		$\sqrt{N}(\tilde{\beta}_X^{(ols)} - \beta_X)  \stackrel{d}{\longrightarrow} \mathcal{N}(0,\sigma_y^2 V_X^{-1})$.
		\item Let $a_N = \sqrt{N}$ and Assumption \ref{ass: centrality} holds, then $\tilde{\beta}_Z^{(ols)} \stackrel{L_2}{\longrightarrow} \beta_Z$ and consequently $\tilde{\beta}_Z^{(ols)} \stackrel{P}{\longrightarrow} \beta_Z$. 
	\end{enumerate}    
\end{theorem}
A common practice to ensure that elements in \(C\) and \(X\) are of the same order is to set \(a_N = \sqrt{N}\) \citep{le2022linear}. However, this does not guarantee the consistency of \(\hat{\beta}_C^{(ols)}\). Numerical evidence for the performance of \(\hat{\beta}_C^{(ols)}\) under various choices of \(a_N\) is provided in Section \ref{sec:simu}. 

\subsection{Centrality-based regression with measurement error}
In practice, real networks often contain noise, necessitating the consideration of the influence of \(E_0\). The Davis-Kahan theorem \citep{davis1970rotation} characterizes the \(\ell_2\)-difference of the top eigenvector. Theorem \ref{thm:asy2} gives the consistency of the OLS estimators with measurement error for C-MNetR.

\begin{theorem}\label{thm:asy2}
	Let $\hat{\beta} := (\hat{\beta}_X^\top ,\hat{\beta}_C^\top )^\top$ be defined as in \eqref{eq:est1}. When $a_N = \sqrt{N}$ and Assumptions \ref{ass: noise}, \ref{ass: design}, \ref{ass: iden} and \ref{ass: spectral gap} hold, we have $\hat{\beta}_X \stackrel{P}{\longrightarrow} \beta_X$.
\end{theorem}
Unfortunately, even when Assumption \ref{ass: spectral gap} holds and \(a_N = \sqrt{N}\), bias may persist in \(\hat{\beta}_C\). The lack of consistency in \(\hat{\beta}_C\) is primarily due to the node-based dependence in network data, as later confirmed by our numerical experiments in Table~\ref{tab:tabthree}.

However, CC-MNetR, which relies on the centrality- and community-based variable $Z$ provides a solution. Theorem~\ref{thm:asy3} guarantees the consistency of the OLS estimators in the presence of measurement errors.

\begin{theorem}\label{thm:asy3}
	Let $\tilde{\beta} \equiv (\tilde{\beta}_X^\top ,\tilde{\beta}_Z^\top )^\top$ be defined as in \eqref{eq:est2}. When $a_N = \sqrt{N}$ and Assumptions \ref{ass: noise}-\ref{ass: spectral gap} hold, $\tilde{\beta}_X \stackrel{P}{\longrightarrow} \beta_X$ and $\tilde{\beta}_Z \stackrel{P}{\longrightarrow} \beta_Z$.
\end{theorem}

In addition, Theorem~\ref{prop1_revised} below guarantees that, \(\tilde{\beta}_Z^{(ols)}\) and $\tilde{\beta}_Z$ are asymptotically normal, facilitating formal inference regardless of whether measurement errors are present.

\begin{theorem}
	Suppose Assumptions \ref{ass: noise}-\ref{ass: centrality} hold.  
	\begin{enumerate}[(i)]
		\item \textbf{(Noiseless Network)} If $a_N = \sqrt{N}$, then $\sqrt{\frac{Z^\top Z}{\sigma_y^2}} \left( \tilde{\beta}_Z^{(ols)} - \beta_Z \right) \stackrel{d}{\longrightarrow} \mathcal{N}(0, 1).$
		
		\item \textbf{(Noisy Network)} If \(a_N \to \infty\) and \(a_N^2 /\delta \to 0\) as \(N \to \infty\), then $\sqrt{\frac{Z^\top Z}{\sigma_y^2}} \left( \tilde{\beta}_Z - \beta_Z \right) \stackrel{d}{\longrightarrow} \mathcal{N}(0, 1).$
	\end{enumerate}
	\label{prop1_revised}
\end{theorem}
The detailed proofs of Theorem~\ref{thm:asy1}-\ref{prop1_revised} are collected in the supplement.
\section{Simulation Experiments}  
\label{sec:simu}  
In this section, we use simulated data to examine the performance of the proposed C-MNetR and CC-MNetR models for multiplexes, both with and without measurement errors.  

\subsection{Data Generation}  
\label{subsec:simu}  
For multiplexes, since connections between layers are given by the identity matrix, we only need to generate the weighted adjacency matrix for each layer. We consider networks with no self-loops, i.e., for layer \(\ell\), all diagonal elements of \(A^{(\ell)}\) are 0. Here, we set \(L = 2\) and \(P = 2\), and generate the binary adjacency matrix \(T^{(\ell)}\) using the stochastic block model (SBM) introduced in \cite{holland1983stochastic}.  

We assume networks in different layers share the same community structure and generate single-layer networks as follows. Suppose there are three communities (\(R = 3\)), with each community containing approximately \(N/3\) nodes. For layer \(\ell\), we define \(P^{(\ell)} \equiv (p^{(\ell)}_{l,k})\), where \(p^{(\ell)}_{l,k}\) denotes the probability of having an edge between community \(l\) and community \(k\) in layer \(\ell\). In this example, we assume the connection probability matrices \(\mathcal{P} = \{P^{(1)}, \dots, P^{(L)}\}\) are identical and  
\[
P^{(\ell)} = \left[\begin{array}{ccc}
	0.8 & 0.1 & 0.1  \\
	0.1 & 0.8 & 0.1  \\
	0.1 & 0.1 & 0.8  
\end{array}\right], \quad \ell = 1, 2.  
\]  
This implies entries in the binary adjacency matrix \(T^{(\ell)}\) satisfy  
\[
T^{(\ell)}_{ij} \mid p^{(\ell)}_{c_i,c_j} \sim \operatorname{Bin}(1, p^{(\ell)}_{c_i,c_j}).  
\]  
For all non-zero entries in \(T^{(\ell)}\), we further assume the weights of these edges are i.i.d. \(U(1, 2)\) random variables, thus giving the weighted matrices \(A^{(\ell)}\) and henceforth \(B_0\). We then generate the symmetric matrix \(E_0\) by assuming its upper triangular entries are i.i.d. \(\mathcal{N}(0, \sigma_b^2)\).  

We assume all entries of the design matrix \(X\) are i.i.d. \(\mathcal{N}(0, 1)\), and centrality measures \(C\) and \(Z\) are calculated according to \eqref{eq: 2.18}. Additionally, we set \(\beta_X = (1, 2)^{\top}\), \(\beta_C = (1, 2)^{\top}\), \(\beta_Z = 2\), and suppose \(\varepsilon \sim \mathcal{N}\left(0, \sigma_y^2 I_N\right)\) with \(\sigma_y = 1\).

\subsection{Performance without measurement error}
In this section, we examine the performance of the OLS estimators using simulated data. We choose $N = 100,200,500,1000$, set $\sigma_b = 0.25$ and generate $n = 1000$ replications for each choice of $N$.

Figure 4 in Section S3.1 in the supplement gives the QQ-plot of 1000 estimates of $\hat{\beta}^{(ols)}$ under $N = 500$ and $a_N = N^{0.8}$ which illustrates the normality of $\hat{\beta}^{(ols)}$.
Numerical results in Table~\ref{tab:tabone} confirm the consistency of $\hat{\beta}_{X}^{(ols)}$ when $a_N = N^{0.5}, N^{0.8}, N$. 
However, for $\hat{\beta}_{C}^{(ols)}$, consistency does not hold when $a_N = \sqrt{N}$. \zh{This is because, when $a_N = \sqrt{N}$, the multiplex network configuration we used results in $V$ satisfying $l_N = O(N^{-1/2})$, as shown in Figure 6 in Supplementary material S3.2. Thus $a_N l_N = O(1)$, Condition \eqref{con_3.1} is not fulfilled and the consistency of $\hat{\beta}_C^{(ols)}$ cannot be guaranteed.} As the order of $a_N$ increases, the consistency of $\hat{\beta}_{C}^{(ols)}$ improves, which agrees with our conclusions in Theorem~\ref{thm:asy1}. 

For the CC-MNetR model, Figure~5 in S3.1 in the supplement gives the QQ-plot of 1000 estimates of $\tilde{\beta}^{(ols)}$ when $N = 500$ and $a_N = \sqrt{N}$, which confirms the asymptotic normality of $\tilde{\beta}^{(ols)}$. Table~\ref{tab:tabtwo} further shows the consistency of $\tilde{\beta}^{(ols)}$ for $a_N = \sqrt{N}$.

\begin{table}[H]
	\centering
	\begin{tabular}{c|ccccc}
		\hline
		\(a_N\) & \(N\) & \(\hat{\beta}_{X_1}^{(ols)}\) & \(\hat{\beta}_{X_2}^{(ols)}\) & \(\hat{\beta}_{C_1}^{(ols)}\) & \(\hat{\beta}_{C_2}^{(ols)}\) \\
		\hline
		\multirow{4}{*}{\(N^{0.5}\)} & 100 & 1.0033 (0.10) & 2.0063 (0.10) & 1.0088 (1.15) & 1.9886 (0.76) \\
		& 200 & 0.9983 (0.07) & 1.9988 (0.07) & 1.0099 (1.00) & 1.9998 (0.95) \\
		& 500 & 0.9975 (0.05) & 2.0003 (0.05) & 1.0384 (1.04) & 1.9637 (0.93) \\
		& 1000 & 0.9994 (0.03) & 2.0010 (0.03) & 0.9993 (0.78) & 1.9982 (1.29) \\
		\hline
		\multirow{4}{*}{\(N^{0.8}\)} & 100 & 1.0010 (0.11) & 2.0004 (0.10) & 1.0150 (0.29) & 1.9904 (0.19) \\
		& 200 & 0.9973 (0.07) & 1.9978 (0.07) & 0.9843 (0.22) & 2.0143 (0.21) \\
		& 500 & 0.9983 (0.04) & 2.0000 (0.04) & 1.0017 (0.16) & 1.9984 (0.14) \\
		& 1000 & 1.0000 (0.03) & 2.0000 (0.03) & 0.9980 (0.09) & 2.0031 (0.16) \\
		\hline
		\multirow{4}{*}{\(N\)} & 100 & 0.9979 (0.11) & 1.9990 (0.10) & 0.9943 (0.11) & 2.0039 (0.07) \\
		& 200 & 1.0013 (0.07) & 2.0016 (0.07) & 1.0013 (0.07) & 1.9987 (0.07) \\
		& 500 & 1.0009 (0.05) & 2.0000 (0.05) & 0.9990 (0.04) & 2.0009 (0.04) \\
		& 1000 & 1.0005 (0.03) & 2.0021 (0.03) & 0.9994 (0.02) & 2.0008 (0.03) \\
		\hline
	\end{tabular}
	\caption{The average (standard deviation) of 1000 estimates of \(\hat{\beta}^{(ols)}\) where the true value \(\beta = (1, 2, 1, 2)^\top\).}
	\label{tab:tabone}
\end{table}

\begin{table}[h]
	\centering
	\begin{tabular}{c|ccc}
		\hline N & $\tilde{\beta}_{X_1}^{(ols)}$ & $\tilde{\beta}_{X_2}^{(ols)}$ & $\tilde{\beta}_{Z}^{(ols)}$ \\
		\hline 100 & 1.0032(0.10) & 2.0062(0.10) & 1.9927(0.14) \\
		\hline 200 & 0.9983(0.07) & 1.9989(0.07) &  2.0094(0.10) \\
		\hline 500 & 0.9974(0.05) & 2.0002(0.05) & 1.9982(0.06) \\
		\hline 1000 & 0.9994(0.03) & 2.0010(0.03) &  1.9978(0.04) \\
		\hline
	\end{tabular}
	\caption{The average (standard deviation) of 1000 estimates of $\tilde{\beta}^{(ols)}$ when $\beta = (1,2,2)^\top$ and $a_N = \sqrt{N}$.}
	\label{tab:tabtwo}
\end{table}

\subsection{Performance with Measurement Error}  
\label{sec: 4.2}  
When measurement error is present, Theorems~\ref{thm:asy2} and \ref{thm:asy3} indicate that Assumption \ref{ass: spectral gap} is necessary for the consistency of \(\hat{\beta}\) and \(\tilde{\beta}\). To ensure \(B_0\) has a sufficiently large spectral gap, we generate weighted adjacency matrices \(A^{(\ell)}\) with significant differences: for \(\ell = 1, 2\), entries are i.i.d. \(U(1, 2)\) and \(\text{Exp}(1)\), respectively, then rescaled to the range \([1, 2]\),ensuring that all entries are on a comparable scale. For example, denote the original elements of $A^{(\ell)}$ as $a_{ij}$, then the rescaled elements $a'_{ij}$ are calculated as \( a'_{ij} = 1 + \frac{(a_{ij} - \min(A^{(\ell)}))}{(\max(A^{(\ell)}) - \min(A^{(\ell)}))}. \) Other settings follow Section \ref{subsec:simu}, with \(N = 100, 200, 500, 1000\) and \(\sigma_b = 0.25\).  

\begin{table}[H]
	\centering
	\begin{tabular}{c|c|ccccc}
		\hline
		\(a_N\) & \(N\) & \(\frac{a_N}{\delta}\) & \(\hat{\beta}_{X_1}\) & \(\hat{\beta}_{X_2}\) & \(\hat{\beta}_{C_1}\) & \(\hat{\beta}_{C_2}\) \\
        \hline \multirow{4}{*}{\(N^{0.3}\)} &$100$ & 0.31 & 1.0015(0.11) & 1.9956(0.10) & 0.9129(1.91) & 2.9689(24.90) \\
		 &$200$ & 0.18 & 1.0016(0.07) &  1.9982(0.07) & 0.7047(2.49) & 9.8955(68.32)\\
	& $500$ & 0.09 &  1.0001(0.05) & 1.9984(0.05) & 0.6196(3.07) & 27.97(210.89)\\
	   & $1000$ & 0.06 & 1.0002(0.03) & 1.9978(0.03) &0.9108(3.20) & 15.0318(460.67)\\
       \hline
		\multirow{4}{*}{\(N^{0.5}\)} & 100 & 0.77 & 0.9952 (0.10) & 1.9990 (0.10) & 0.8794 (0.73) & 3.5993 (9.61) \\
		& 200 & 0.52 & 0.9968 (0.07) & 2.0009 (0.07) & 0.8481 (0.84) & 6.1348 (22.32) \\
		& 500 & 0.33 & 0.9978 (0.05) & 2.0000 (0.05) & 0.8879 (0.82) & 9.9986 (56.75) \\
		& 1000 & 0.22 & 0.9989 (0.03) & 1.9992 (0.03) & 0.9367 (0.84) & 11.0566 (121.18) \\
		\hline
		\multirow{4}{*}{\(N^{1.5}\)} & 100 & 77 & 1.0505 (0.72) & 1.9527 (0.51) & 0.9164 (0.04) & 3.0875 (0.52) \\
		& 200 & 104 & 0.8463 (0.50) & 2.0655 (0.55) & 0.8776 (0.03) & 5.3456 (0.86) \\
		& 500 & 163 & 0.7587 (0.53) & 1.9127 (0.50) & 0.8773 (0.02) & 10.4453 (1.27) \\
		& 1000 & 224 & 1.0104 (0.47) & 1.9073 (0.49) & 0.8821 (0.01) & 19.0037 (1.79) \\
		\hline
	\end{tabular}
	\caption{The average (standard deviation) of 1000 estimates of \(\hat{\beta}\) where the true value \(\beta = (1, 2, 1, 2)^\top\).}
	\label{tab:tabthree}
\end{table}

For \(\hat{\beta}\), Table \ref{tab:tabthree} shows \(\hat{\beta}_X\) is consistent under Assumption \ref{ass: spectral gap}, but \(\hat{\beta}_C\) exhibits persistent bias when \(a_N = \sqrt{N}\), even as \(N\) increases. This bias highlights the limitation of C-MNetR with measurement error, motivating the introduction of the CC-MNetR model.

Increasing \(a_N\) does not simultaneously ensure consistency for \(\hat{\beta}_X\) and \(\hat{\beta}_C\). The consistency of \(\hat{\beta}_X\) requires \(a_N\) to align with the spectral gap, while \(\hat{\beta}_C\) requires the order of $a_N$ be greater than $1/l_N$. For example, even with \(a_N = N^{1.5}\), \(\hat{\beta}_X\) remains biased, and \(\hat{\beta}_C\) is inconsistent, as shown in Table~\ref{tab:tabthree}.

\begin{table}[h]
	\centering
	\begin{tabular}{c|c|c|ccc}
		\hline $a_N$ &$N$ &$\frac{a_N}{\delta}$ & $\tilde{\beta}_{X_1}$ & $\tilde{\beta}_{X_2}$ & $\tilde{\beta}_{Z}$\\
		\hline 
        \multirow{4}{*}{\(N^{0.3}\)} &$100$ & 0.31 & 1.0017(0.34) & 1.9958(0.01) & 1.9724(0.43)\\
		 &$200$ & 0.18 & 1.0023(0.07) &  1.9981(0.07) & 1.9882(0.38)\\
	& $500$ & 0.09 & 1.0005(0.05) & 1.9985(0.05) & 1.9963(0.34)\\
	  & $1000$ & 0.06 & 1.0003(0.03) & 1.9979(0.03) &2.0017(0.25) \\
		\hline \multirow{4}{*}{\(N^{0.5}\)} &$100$ & 0.77 & 0.9952(0.10) & 1.9995(0.10) & 2.0073(0.18)\\
		 &$200$ & 0.52 & 0.9977(0.07) &  2.0007(0.07) & 2.0020(0.13)\\
	& $500$ & 0.33 & 0.9982(0.05) & 2.0002(0.05) & 2.0082(0.09)\\
	   & $1000$ & 0.22 & 0.9989(0.03) & 1.9993(0.03) &1.9991(0.06) \\
       \hline
	\end{tabular}
	\caption{The average (standard deviation) of 1000 estimates of $\tilde{\beta}$ where the true value $\beta = (1,2,2)^\top$.}
	\label{tab:tab4}
\end{table}

On the other hand, the CC-MNetR model demonstrates superior performance. Table \ref{tab:tab4} confirms the consistency of \(\tilde{\beta}\), showing more tractable asymptotic behavior under measurement errors. Overall, the CC-MNetR model offers a reliable approach for analyzing multilayer network data with community information. The MSE line graph for these results is provided in Section S3.3 in the supplement.

\section{Applications to the World Input-Output Database}
\label{sec:realdata}
We now apply CC-MNetR to the World Input-Output Database (WIOD), which covers 28 EU countries and 15 other major countries or regions from 2000 to 2014 \citep{timmer2015illustrated}. Data for 56 sectors are classified according to the International Standard Industrial Classification Revision 4 (ISIC Rev.4). To construct the multilayer network, each node represents an individual industry sector, and each layer corresponds to the input-output data for a specific country, resulting in $N = 56$ and $L = 43$. Detailed information for each sector is listed in Table in Table 2 in Section S5 of the supplement. 

Each year's input-output table forms a supra-adjacency matrix \(B_a\). We focus on \(B = B_a + B_a^\top\), representing total flow between sectors. \(B\) is normalized to ensure comparability, with elements scaled to \([0, 2]\), in accordance with how $B$ is partitioned in \eqref{eq:supra}. We set $a_N = \sqrt{NL}$ and the 56 industries are naturally divided into 20 communities. The community information is provided in Table 2 in Section S5 of the supplement. For our analysis, we use the input-output tables from 2007 and 2014 as examples.

Different community-based centrality measures 
$Z$ for the years 2007 and 2014 are summarized in Figure 8 of the supplement. The figure shows that the communities corresponding to \emph{Construction}, \emph{Manufacturing}, and \emph{Mining and Quarrying} exhibit the highest centrality scores in the network. This observation aligns with the understanding that these sectors are often considered the backbone of an economy due to their significant contributions to GDP and employment \citep{lean2001empirical, szirmai2013manufacturing, FUGIEL2017159}. In contrast, the communities representing \emph{Activities of Extraterritorial Organizations and Bodies} and \emph{Activities of Households as Employers}; \emph{Undifferentiated Goods- and Services-Producing Activities of Households for Own Use} show the lowest centrality scores within the network.

For our analysis, we selected the years 2007 and 2014 to investigate any significant changes before and after the global financial crisis. The centrality score of the \emph{Construction} sector declined from 2007 to 2014, potentially reflecting the impact of the real estate market crash during the financial crisis. Conversely, sectors such as \emph{Wholesale and Retail Trade; Repair of Motor Vehicles and Motorcycles}, \emph{Human Health and Social Work Activities}, and \emph{Administrative and Support Service Activities} experienced a substantial increase in their centrality scores. This suggests an economic recovery and a shift in the economic structure during the post-crisis period, with increased prominence of sectors like healthcare and services.

After analyzing the changes in centrality scores $Z$, we now apply the CC-MNetR model to the WIOD data using the Socio-Economic Accounts (SEA) dataset for covariates \(X\) and response \(y\). The SEA provides industry-level data for each country on various macro-characteristics, including employment, capital stocks, gross output, and value-added at current and constant prices, measured in millions of local currency. The industry classification aligns with the world input-output tables. Since the SEA data is organized at the sector level for each country, we calculate the average sector data across different countries before conducting the regression analysis.

The description of the 10 variables in the SEA dataset is summarized in Table 1 in Section S4 of the supplement. We choose \emph{Gross Output} (GO) as the response variable 
$y$, while the remaining 9 variables are considered as potential covariates $X$. However, since \emph{Intermediate Input} (II) is part of the total flow matrix $B$ and depends on the community-based centrality score 
$Z$, we exclude II from our analysis.

A significant issue in model specification is multicollinearity among covariates. Specifically, strong linear relationships exist among the variables EMP (Number of Persons Engaged), LAB (Labour Compensation), EMPE (Number of Employees), and H\_EMPE (Total Hours Worked by Employees), due to their related economic interpretations. To address this, we retain only EMP and remove LAB, EMPE, and H\_EMPE, mitigating multicollinearity concerns.
\begin{table}[H]
	\centering
	\begin{tabular}{c|c c c c c c}
		\textbf{Variable}  & VA & CAP & COMP & EMP & K \\
		\hline
		\textbf{VIF} & 38.80 & 16.34 &  7.03 & 3.44 & 1.41  
	\end{tabular}
	\caption{VIFs of 5 potential variables}
	\label{tab: vif}
\end{table}
Figure 9 in Section S4.2 of the supplement shows the pairwise scatter plot matrix of the 6 variables, highlighting persistent multicollinearity. To further refine the model, we use the Variance Inflation Factor (VIF) to quantify the severity of the multicollinearity issue in the regression analysis. The VIF for a specific predictor $X_i$ is computed as
\[\text{VIF}(X_i) = \frac{1}{1 - R_i^2},\]
where $R_i^2$ is the coefficient of determination from a regression of $X_i$ on all the other covariates in the model. VIF measures how much the variance of an estimated regression coefficient is increased due to multicollinearity. Typically, a VIF value exceeding 5 is considered indicative of high multicollinearity \citep{o2007caution}. 

After calculating the VIF for each covariate, we remove those with a VIF greater than 5. The VIFs for the 5 potential covariates based on the data from 2014 are listed in Table~\ref{tab: vif}. Among the three variables with VIFs larger than 5 (VA, CAP, and COMP), we retain only VA. This refinement results in a final model including only VA (Value Added), EMP (Number of Persons Engaged), and K (Nominal Capital Stock) as covariates. Additionally, we normalize GO, VA, EMP, and K to have a mean of 0 and a variance of 1.

We compare two regression models to evaluate the impact of the community-based centrality score \(Z\). The R-squared statistic of the first model \(y = \beta X\) is 0.8152, indicating that approximately 81.52\% of the variation in GO (Gross Output) is explained by VA, EMP, and K. Including \(Z\) in the CC-MNetR model increases the R-squared statistic to 0.87, demonstrating the importance of \(Z\) in the regression analysis. Additionally, the F-test results in Table~3 of the supplement show that \(Z\) is associated with a high F-statistic and a small p-value (p \(< 0.001\)) for both 2007 and 2014, reinforcing its statistical significance in explaining GO.

From Table~3 of the supplement, the estimated regression models for 2007 and 2014 are:
\begin{align*}
	\text{GO}_{2007} &= 0.70 Z + 0.89 \text{VA} - 0.06 \text{EMP} - 0.005 \text{K} - 0.38, \\
	\text{GO}_{2014} &= 0.61 Z + 0.95 \text{VA} - 0.14 \text{EMP} - 0.02 \text{K} - 0.34.
\end{align*}
For 2007, the coefficient of VA is \(0.89\), indicating that a unit increase in VA raises GO by \(0.89\) units. This highlights VA's importance as a key driver of production efficiency and economic growth, as it represents gross output minus intermediate inputs. The negative coefficient for EMP (\(-0.06\)) suggests diminishing marginal productivity, where additional labor inputs beyond an optimal point contribute less to GO due to overcrowding or insufficient complementary capital. The positive coefficient for \(Z\) (\(0.70\)) indicates that sectors with higher centrality scores produce more output, reflecting the importance of network connections in economic activity.

Comparing the models for 2007 and 2014, we find that the positive coefficient for $Z$ suggests that as the centrality score of a node increases, so does the Gross Output ($y$). This aligns with the understanding that sectors with extensive connections to other sectors (i.e., those that are more central in the network) tend to produce more output. Additionally, we observe a decrease in the estimated coefficient for EMP and an increase in its statistical significance. This suggests that the impact of EMP on GO has strengthened over time.


\section{Concluding remarks}
\label{sec:conc}
We propose a multilayer network regression framework where either individual or community-based centrality scores can be used as predictors. Both theoretical analyses and numerical evidence demonstrate that the CC-MNetR model effectively captures the impact of underlying communities on network behavior, offering a robust approach for network analysis.

The real data analysis presented in Section \ref{sec:realdata} aligns with established economic principles, highlighting the significant role of centrality measures. Specifically, the community-based centrality measure $Z$ has shown strong significance within the regression model. The positive coefficient for $Z$ indicates a substantial linear relationship between a sector's centrality score and its total flow within the input-output framework. Furthermore, the stable estimated values for $Z$ in both years suggest that the significant impact of network centrality measures on economic flows remains consistent over time.

\section*{Supplementary Materials}

The supplementary materials includes detailed proofs of theoretical results, simulation studies comparing model performance, and real-world applications using the World Input-Output Database (WIOD). Additionally, industry classification details from WIOD and variable details of Socio-Economic Accounts (SEA) dataset are provided for reference. This supplementary material supports the main findings and methodology presented in the paper.
\par
\section*{Acknowledgements}

T. Wang gratefully acknowledges 
Science and Technology Commission of Shanghai Municipality Grant 23JC1400700 and National Natural Science Foundation of China Grant 12301660.
	
	Both authors also thank Shanghai Institute for Mathematics and Interdisciplinary Sciences (SIMIS) for their financial support. This research was partly funded by SIMIS under grant number SIMIS-ID-2024-WE. T. Wang is grateful for the resources and facilities provided by SIMIS, which were essential for the completion of this work.
\par


\bibhang=1.7pc
\bibsep=2pt
\fontsize{9}{14pt plus.8pt minus .6pt}\selectfont
\renewcommand\bibname{\large \bf References}
\expandafter\ifx\csname
natexlab\endcsname\relax\def\natexlab#1{#1}\fi
\expandafter\ifx\csname url\endcsname\relax
  \def\url#1{\texttt{#1}}\fi
\expandafter\ifx\csname urlprefix\endcsname\relax\def\urlprefix{URL}\fi

  \bibliographystyle{chicago}      
  \bibliography{Netreg}   

\vskip .65cm  
\noindent  
Zhuoye Han, Fudan University  
\vskip 2pt  
\noindent  
E-mail: zhuoyehan21@m.fudan.edu.cn  
\vskip 2pt  

\noindent  
Tiandong Wang, Fudan University and Shanghai Academy of Artificial Intelligence for Science 
\vskip 2pt  
\noindent  
E-mail: td\_wang@fudan.edu.cn  
\vskip 2pt  

\noindent  
Zhiliang Ying, Columbia University  
\vskip 2pt  
\noindent  
E-mail: zy9@columbia.edu  
\vskip 2pt

\end{document}



\renewcommand{\baselinestretch}{2}

\markright{ \hbox{\footnotesize\rm Statistica Sinica: Supplement
}\hfill\\[-13pt]
\hbox{\footnotesize\rm
}\hfill }

\markboth{\hfill{\footnotesize\rm Z. HAN, T. WANG AND Z. YING} \hfill}
{\hfill {\footnotesize\rm MULTILAYER NETWORK REGRESSION} \hfill}

\renewcommand{\thefootnote}{}
$\ $\par \fontsize{12}{14pt plus.8pt minus .6pt}\selectfont


 \centerline{\large\bf MULTILAYER NETWORK REGRESSION}
\vspace{2pt}
 \centerline{\large\bf WITH EIGENVECTOR CENTRALITY}
\vspace{2pt}
 \centerline{\large\bf AND COMMUNITY STRUCTURE}
\vspace{.25cm}
 \author{Zhuoye Han, Tiandong Wang and Zhiliang Ying}
\vspace{.4cm}
 \centerline{\it Fudan University and Columbia University}
\vspace{.55cm}
 \centerline{\bf Supplementary Material}
\vspace{.55cm}
\fontsize{9}{11.5pt plus.8pt minus .6pt}\selectfont
\noindent
This document provides supplementary material for the paper on \emph{Multilayer network regression with eigenvector centrality and community structure}. It includes detailed proofs of theoretical results, simulation studies comparing model performance, and real-world applications using the World Input-Output Database (WIOD). Additionally, industry classification details from WIOD and variable details of Socio-Economic Accounts (SEA) dataset are provided for reference. This supplementary material supports the main findings and methodology presented in the paper.
\par

\setcounter{section}{0}
\setcounter{equation}{0}
\def\theequation{S\arabic{section}.\arabic{equation}}
\def\thesection{S\arabic{section}}

\fontsize{12}{14pt plus.8pt minus .6pt}\selectfont
\section{Proof of theoretical results in Section~3}
In this section, we present the proofs. To begin with, we introduce the following notation for projection matrices used throughout the analysis: $P_X:= X(X^\top X)^{-1}X^\top $, $P_C:= C(C^\top C)^{-1}C^\top $, $P_{\hat{C}}:= \hat{C}(\hat{C}^\top \hat{C})^{-1}\hat{C}^\top $ and $P_{\hat{Z}}:= \hat{Z}(\hat{Z}^\top \hat{Z})^{-1}\hat{Z}^\top $. These projection matrices play a central role in characterizing the properties of the estimators and their asymptotic behavior.
\subsection{Proof of Column Full-Rank Implication} \label{apx:proof-rank}

\textbf{Claim:} If $\sigma_{\min}\left((I_N - P_X)V\right) \ge l_N >0$, then $W_1 = (X, C)$ is column full-rank.

\begin{proof}
    Assume $W_1 = [X \quad C]$ is \textit{not} column full-rank. Then there exists a non-zero vector $\theta = [\theta_X^\top \quad \theta_C^\top]^\top \neq 0$ such that:
    $$
    X \theta_X + C \theta_C = 0.
    $$
    If $\theta_C = 0$, then $X \theta_X = 0$. Since $X$ is column full-rank (by $N > P+L$), this implies $\theta_X = 0$, contradicting $\theta \neq 0$. Thus, $\theta_C \neq 0$. 

    Rearranging the equation:
    $$
    C \theta_C = -X \theta_X.
    $$
    Projecting both sides onto the orthogonal complement of $X$:
    $$
    (I_N - P_X)C \theta_C = (I_N - P_X)(-X \theta_X) = 0,
    $$
    where we used $(I_N - P_X)X = 0$. Substituting $C = a_N V$:
    $$
    (I_N - P_X)V \theta_C = \frac{1}{a_N} (I_N - P_X)C \theta_C = 0.
    $$
    This implies:
    $$
    \sigma_{\min}\left((I_N - P_X)V\right) \le \| (I_N - P_X)V \theta_C \|_2 / \| \theta_C \|_2 = 0,
    $$
    which contradicts $\sigma_{\min}\left((I_N - P_X)V\right) \ge l_N > 0$. Therefore, $W_1$ must be column full-rank.
\end{proof}

\subsection{Three important Lemmas}
\begin{lemma}\label{lem:davis_Kahan}\citep{davis1970rotation}Recall the network model in (2.10). Let $\delta: =\lambda_1-\lambda_2$ be the spectral gap between the largest and second largest eigenvalues of $B_0$. 
	Suppose $\widetilde{u}_1$ and $u_1$ are the top eigenvectors of $B$ and $B_0$, respectively. Then we have 
	$$
	\left\|\widetilde{u}_1-u_1\right\|_2=O\left(\frac{\|E_0\|_2}{\delta}\right),
	$$
	where $\|E_0\|_2=\max _{\|u\|_2 \leq 1}\|E_0 u\|_2$ denotes the matrix operator norm.
\end{lemma}
Lemma~\ref{lem:davis_Kahan} requires \(\delta \gg \|E_0\|_2\) for \(\widetilde{u}_1\) to converge to \(u_1\). In our framework, this result directly translates to the estimation error bound between the noisy and true centrality measures. Specifically, we derive the following explicit rate for \(\hat{C}\):
\begin{lemma}\label{lem:rate_C}
	Under Assumptions 1-3, we have
	\begin{align}\label{eq:rate_C}
		\mathbb{E}\left[\|\hat{C} - C\|_F^2\right] = O\left(\frac{a_N^2 N}{\delta^2}\right).
	\end{align}
\end{lemma}
\begin{proof}
    Under Assumption 2 and the setting of Lemma~\ref{lem:davis_Kahan}, we have 
\begin{align*}
	\|\hat{C} - C\|_F^2 &= \|\text{vec}(\hat{C}) - \text{vec}(C)\|_2^2\\
	&=a_N^2 \left\|\widetilde{u}_1-u_1\right\|_2^2,
\end{align*}
and from Lemma~\ref{lem:davis_Kahan} we see that $\left\|\widetilde{u}_1-u_1\right\|_2=O\left(\frac{\|E_0\|_2}{\delta}\right)$. Therefore, $\left\|\widetilde{u}_1-u_1\right\|_2^2 \le c  \frac{\|E_0\|_2^2}{\delta^2} \text{ a.s.} $ for some positive constant $c$. Therefore, we have
\begin{align*}
	\mathbb{E}\left[\|\hat{C} - C\|_F^2\right] \le c a_N^2  \frac{\mathbb{E}\left[\|E_0\|_2^2\right]}{\delta^2}.
\end{align*}
Under Assumption 1 where $ \mathbb{E}\left[\|E_0\|_2^2\right] = O(N)$, for fixed $L$, we obtain
\begin{align*}
	\mathbb{E}\left[\|\hat{C} - C\|_F^2\right] = O(\frac{a_N^2 N}{\delta^2}).
\end{align*}
\end{proof}

In what follows, Lemma \ref{lemma:trace} is a powerful tool, which we now explain. It is used in the proofs of Theorem 3 and 4.
\begin{lemma}\label{lemma:trace}
	Suppose $A$ and $B$ are positive semi-definite matrices with the same size $n \times n$. Then we have $\text{tr}(AB) \le \text{tr}(A)\text{tr}(B)$. 
\end{lemma}
\begin{proof}
	From Cauchy-Schwarz inequality, we have 
	\begin{align*}
		\text{tr}(AB) \le \|A\|_F\|B\|_F = \sqrt{\text{tr}(A^2)} \sqrt{\text{tr}(B^2)}.
	\end{align*}
	Denote the eigenvalues of $A$ as $\nu_i \ge 0, i = 1, \cdots n$, the eigenvalues of $B$ as $\mu_i \ge 0, i = 1, \cdots, n$. Then $\sqrt{\text{tr}(A^2)} = \sqrt{\sum\nu_i^2} \le \sum \nu_i = \text{tr}(A)$, and similarly we have $\sqrt{\text{tr}(B^2)} = \sqrt{\sum\mu_i^2} \le \sum \mu_i = \text{tr}(B)$. Finally, we have $\text{tr}(AB) \le \text{tr}(A)\text{tr}(B)$ and the proof is complete.
\end{proof}

\subsection{Proof of Theorem 1}

(i) Let $W_1 = (X, C)$ and $\beta = (\beta_X^\top , \beta_C^\top )^\top $, then the OLS estimator is 
$$\hat{\beta}^{(ols)} = \underset{\beta_X, \beta_C}{\arg \min }\left\|y-X \beta_X-C \beta_C\right\|_2^2.$$
Define also that $\mathbb{L} := \left\|y-X \beta_X-C \beta_C\right\|_2^2$, then setting the partial derivatives of all the parameters as zero leads to
\begin{align*}
	& \frac{\partial \mathbb{L}}{\partial \beta_X}= -\frac{2}{N}X^\top  (y - X \beta_X-C \beta_C) = 0,\\
	& \frac{\partial \mathbb{L}}{\partial \beta_C}= -\frac{2}{N}C^\top  (y - X \beta_X-C \beta_C) = 0,
\end{align*}
which gives
\begin{equation*}
	\begin{array}{l}
		X^\top X\hat{\beta}_X^{(ols)} = X^\top (y - C\hat{\beta}_C^{(ols)}), \\
		C^\top C\hat{\beta}_C^{(ols)} = C^\top (y - X\hat{\beta}_X^{(ols)}).
	\end{array}
\end{equation*}
This further implies
\begin{equation*}
	\begin{array}{ll}
		\hat{\beta}_X^{(ols)} &= (X^\top (I_N - P_{C})X)^{-1}X^\top (I_N - P_C)y \\
		&= (X^\top (I_N - P_{C})X)^{-1}X^\top (I_N - P_C)(X \beta_X+C\beta_C+\varepsilon) \\
		&= \beta_X + (X^\top (I_N - P_{C})X)^{-1}X^\top (I_N - P_C)(C\beta_C+\varepsilon), 
	\end{array}
\end{equation*}
and
\begin{equation*}
	\begin{array}{ll}
		\hat{\beta}_C^{(ols)} &= (C^\top (I_N - P_{X})C)^{-1}C^\top (I_N - P_X)y \\
		&= (C^\top (I_N - P_{X})C)^{-1}C^\top (I_N - P_X)(X \beta_X+C\beta_C+\varepsilon)\\
		&= \beta_C + (C^\top (I_N - P_{X})C)^{-1}C^\top (I_N - P_X)(X \beta_X + \varepsilon).
	\end{array}
\end{equation*}
Note that the projection matrices $P_C$ and $P_X$ satisfy $(I_N - P_C)C = 0$ and $(I_N - P_X)X = 0$, we then have 
\begin{equation}\label{eq: raw}
	\begin{array}{ll}
		\hat{\beta}_X^{(ols)} - \beta_X = (X^\top (I_N - P_{C})X)^{-1}X^\top (I_N - P_C)\varepsilon, \\
		\hat{\beta}_C^{(ols)} - \beta_C = (C^\top (I_N - P_{X})C)^{-1}C^\top (I_N - P_X)\varepsilon.
	\end{array}
\end{equation}
Also, since $W_1$ is column full rank, from the inverse formula for the partitioned matrix, we see that
\begin{align*}
	(W_1^\top W_1)^{-1} &= \left(\begin{array}{cc}X^\top X & X^\top C \\ C^\top X & C^\top C\end{array}\right)^{-1}\\
	&= \left(\begin{array}{cc}(X^\top (I_N - P_{C})X)^{-1}& \boldsymbol{*}_1 \\ \boldsymbol{*}_2 & (C^\top (I_N - P_{X})C)^{-1}\end{array}\right),
\end{align*}
where
\begin{equation*}
	\begin{array}{ll}
		\boldsymbol{*}_1 &= -(X^\top (I_N - P_{C})X)^{-1}X^\top  C(C^\top  C)^{-1} \\
		&= -(X^\top X)^{-1}X^\top C(C^\top (I_N - P_{X})C)^{-1} ,
	\end{array}
\end{equation*}
\begin{equation*}
	\begin{array}{ll}
		\boldsymbol{*}_2 &= -(C^\top (I_N - P_{X})C)^{-1} C^\top X(X^\top X)^{-1} \\
		&= -(C^\top C)^{-1}C^\top X(X^\top (I_N - P_{C})X)^{-1},
	\end{array}
\end{equation*}
and 
\begin{align*}
	(X^\top (I_N - P_{C})X)^{-1} =& (X^\top X)^{-1} + (X^\top X)^{-1}X^\top C(C^\top (I_N - P_X)C)^{-1}C^\top X(X^\top X)^{-1},\\ 
	(C^\top (I_N - P_{X})C)^{-1} =& (C^\top C)^{-1} + (C^\top C)^{-1}C^\top X(X^\top (I_N - P_C)X)^{-1}X^\top C(C^\top C)^{-1}.
\end{align*}
Here, both $X^\top (I_N - P_{C})X$ and $C^\top (I_N - P_{X})C$ are symmetric and positive definite. 

Next, we consider the asymptotic behavior of
\begin{align*}
	\hat{\beta}_X^{(ols)} - \beta_X &= (X^\top (I_N - P_{C})X)^{-1}X^\top (I_N - P_C)\varepsilon \\
	&= (\frac{1}{N}X^\top (I_N - P_{C})X)^{-1}\frac{1}{N}X^\top (I_N - P_C)\varepsilon.
\end{align*}
We start by showing
\begin{align}\label{eq: proj_1}
	\frac{1}{N} X^\top P_C X \stackrel{L_1}{\longrightarrow} 0.
\end{align}
Note that the projection matrix $P_C$ is idempotent and $\text{rank}(P_C) = \text{tr}(P_C) = L$. Therefore, for a fixed $C$, there exists an orthogonal matrix $J = (J_{ij})_{i,j = 1}^N$ such that 
\begin{align}\label{eq: eq_ortho}
	J P_C J^\top = \left[\begin{array}{cc}
		I_L & \boldsymbol{0} \\
		\boldsymbol{0} & \boldsymbol{0}
	\end{array}\right]_{N\times N}.
\end{align}
Here we consider each individual element of $\frac{1}{N}X^\top P_{C} X$. Denote $X = [X_1, \cdots, X_P]$, and we have
\begin{align*}
	\frac{1}{N}X^\top P_{C} X &= \left(\frac{1}{N}X_i^\top P_{C} X_j\right)_{i,j = 1}^P.
\end{align*}
Denote the conditional expectation on $C$ as $\mathbb{E}^C[\boldsymbol{\cdot}]:= \mathbb{E}[\boldsymbol{\cdot} |C]$. By \eqref{eq: eq_ortho}, we have
\begin{align}
	\mathbb{E}^C\left[\left|\frac{1}{N}X_i^\top P_{C} X_j\right|  \right] &= \mathbb{E}^C\left[\left|\frac{1}{N}X_i^\top J^\top \left[\begin{array}{cc}
		I_L & \boldsymbol{0} \nonumber\\
		\boldsymbol{0} & \boldsymbol{0}
	\end{array}\right] J X_j\right|\right]\nonumber\\
	&=\mathbb{E}^C\left[\left|\frac{1}{N}\sum_{k = 1}^{L}\left(\sum_{l = 1}^N J_{kl}X_{li}\right)  \left(\sum_{l = 1}^N J_{kl}X_{lj}\right)\right|\right]\nonumber\\
	&\le \frac{1}{N}\sum_{k = 1}^{L}\mathbb{E}^C\left[\left|\left(\sum_{l = 1}^N J_{kl}X_{li}\right)  \left(\sum_{l = 1}^N J_{kl}X_{lj}\right)\right|\right]\nonumber\\
	&\le \frac{1}{N}\sum_{k = 1}^{L}\left(\mathbb{E}^C\left|\sum_{l = 1}^N J_{kl}X_{li}\right|^2\right)^{\frac{1}{2}}  \left(\mathbb{E}^C\left|\sum_{l = 1}^N J_{kl}X_{lj}\right|^2\right)^{\frac{1}{2}},
	\label{eq:CS}
\end{align}
where the last inequality follows from the Cauchy-Schwartz inequality.

Since $J$ is orthogonal, we have $\sum_{l = 1}^L J_{kl}^2 = 1$, for $k \in \{1, \cdots, N\}$. Since $\mathbb{E}[X_{ij}|C, E_0] = 0$ and $\mathbb{E}[X_{ij}^2|C, E_0] < \infty$ for $1\le i \le N, 1 \le j \le P$,  we obtain
\begin{align*}
	\mathbb{E}^C\left|\sum_{l = 1}^N J_{kl}X_{li}\right|^2 &= \mathbb{E}^C\left(\sum_{l = 1}^N J_{kl}^2 X_{li}^2 + \sum_{l \not= l'}J_{kl}X_{li} \cdot J_{kl'}X_{l'i}\right)\\
	&= \sum_{l = 1}^N J_{kl}^2 \mathbb{E}^C X_{li}^2 + \sum_{l \not= l'}J_{kl} J_{kl'}\mathbb{E}^C\left[X_{li} X_{l'i}\right]\\
	&= \mathbb{E}^C X_{1i}^2 + \sum_{l \not= l'}J_{kl} J_{kl'}\cdot 0\\
	&= \mathbb{E}^C X_{1i}^2 < \infty ,
\end{align*}
as $X_{li}$ and $ X_{l'i}$ are independent for $l\neq l'$.  Therefore, by \eqref{eq:CS}, we have as $N \longrightarrow \infty$,
\begin{align}
	\mathbb{E}^C\left[\left|\frac{1}{N}X_i^\top P_{C} X_j\right|  \right] 
	&\le \frac{1}{N} \sum_{k = 1}^L \sqrt{\mathbb{E}^C X_{1i}^2\mathbb{E}^C X_{1j}^2} = \frac{L}{N}\sqrt{\mathbb{E}^C X_{1i}^2\mathbb{E}^C X_{1j}^2} \longrightarrow 0, \quad \forall i,j \in \{1, \cdots, P\}, \label{eq: norm1}
\end{align}
which further gives
\begin{align*}
	\mathbb{E}\left[\left|\frac{1}{N}X_i^\top P_{C} X_j\right|  \right]  
	&= \mathbb{E}\left[\mathbb{E}^C\left[\left|\frac{1}{N}X_i^\top P_{C} X_j\right|  \right] \right] \\
	&\le \mathbb{E}\left[ \frac{L}{N}\sqrt{\mathbb{E}^C X_{1i}^2\mathbb{E}^C X_{1j}^2}\right]\\
	&= \frac{L}{N}\sqrt{\mathbb{E}^C X_{1i}^2\mathbb{E}^C X_{1j}^2} \longrightarrow 0, \quad \forall i,j \in \{1, \cdots, P\},
\end{align*}
and proves \eqref{eq: proj_1}.
Also, \eqref{eq: proj_1} implies
\begin{align}\label{eq: proj_2}
	\frac{1}{N}X^\top P_C X \stackrel{P}{\longrightarrow} 0.
\end{align}

By the law of large numbers, we have 
\begin{align}\label{eq: prop_X}
	\frac{1}{N}X^\top X \stackrel{P}{\longrightarrow} V_X,
\end{align}
where $V_X$ is a deterministic and nonsingular diagonal matrix. From \eqref{eq: proj_2} and \eqref{eq: prop_X} we also obtain
\begin{align}
	\frac{1}{N}X^\top(I - P_C) X \stackrel{P}{\longrightarrow} V_X. \label{eq: no_inverse}
\end{align}
Applying the continuous mapping theorem to \eqref{eq: no_inverse} we conclude
\begin{align}
	(\frac{1}{N}X^\top(I - P_C) X)^{-1} \stackrel{P}{\longrightarrow} V_X^{-1}.\label{eq: part1_1}
\end{align}
Now we consider the asymptotic normality of 
\begin{align*}
	\sqrt{N}(\hat{\beta}_X^{(ols)} - \beta_X) = (\frac{1}{N}X^\top (I_N - P_{C})X)^{-1}\frac{1}{\sqrt{N}}X^\top (I_N - P_C)\varepsilon.
\end{align*}
By \eqref{eq: proj_1}, we have
\begin{align*}
	\mathbb{E}\left[\left\| \frac{1}{\sqrt{N}}X^\top P_C\varepsilon \right\|_2^2\right] = \frac{\sigma_y^2}{N}\text{tr}\left(\mathbb{E}\left[X^\top P_C X\right]\right) \to 0,
\end{align*}
which $\frac{1}{\sqrt{N}}X^\top P_C\varepsilon \stackrel{P}{\longrightarrow} 0$. Also, for $\frac{1}{\sqrt{N}}X^\top \varepsilon$, the central limit theorem gives that
\begin{align*}
	\frac{1}{\sqrt{N}}X^\top\varepsilon \stackrel{d}{\longrightarrow} \mathcal{N}(0, \sigma_y^2 V_X).
\end{align*}
Hence, we arrive at the asymptotic normality result: 
\begin{align*}
	\sqrt{N}(\hat{\beta}_X^{(ols)} - \beta_X) \stackrel{d}{\longrightarrow} \mathcal{N}(0, \sigma_y^2 V_X^{-1}).
\end{align*}

(ii) To prove the consistency of $\hat{\beta}_C^{(ols)} - \beta_C$, we first point out that 
\begin{align*}
	\hat{\beta}_C^{(ols)} - \beta_C = (C^\top(I_N - P_X) C)^{-1}C^\top (I_N - P_X)\varepsilon,
\end{align*}
and examine
the $\ell_2$-norm of $\hat{\beta}_C^{(ols)} - \beta_C$ as follows:
\begin{align*}
	\mathbb{E}\left[\left\|\hat{\beta}_C^{(ols)} - \beta_C\right\|_2^2\right] &= \mathbb{E}\left[\varepsilon^\top(I_N - P_X) C(C^\top(I_N - P_X) C)^{-2}C^\top (I_N - P_X)\varepsilon\right]\\
	&= \mathbb{E}\left[\text{tr}\left(\varepsilon^\top(I_N - P_X) C(C^\top(I_N - P_X) C)^{-2}C^\top (I_N - P_X)\varepsilon\right)\right]\\
	&= \text{tr}\left(\mathbb{E}\left[\varepsilon\varepsilon^\top(I_N - P_X) C(C^\top(I_N - P_X) C)^{-2}C^\top (I_N - P_X)\right]\right)\\
	&= \sigma_y^2\text{tr}\left(\mathbb{E}\left[(I_N - P_X) C(C^\top(I_N - P_X) C)^{-2}C^\top (I_N - P_X)\right]\right)\\
	&= \sigma_y^2\mathbb{E}\left[\text{tr}\left((I_N - P_X) C(C^\top(I_N - P_X) C)^{-2}C^\top (I_N - P_X)\right)\right]\\
	&=\sigma_y^2\mathbb{E}\left[\text{tr}\left(C^\top(I_N - P_X) C (C^\top(I_N - P_X) C)^{-2} \right)\right]\\
	&= \sigma_y^2\mathbb{E}\left[\text{tr}\left((C^\top(I_N - P_X) C)^{-1} \right)\right].
\end{align*}
Note that $C^\top(I_N - P_X) C$ is positive definite, and we denote its eigenvalues as $\mu_1 \ge \cdots \ge \mu_L>0$. Then we have
\begin{align*}
	\text{tr}\left((C^\top(I_N - P_X) C)^{-1} \right) = \sum_{i = 1}^L \frac{1}{\mu_i}\le \frac{L}{\mu_L} = \frac{L}{a_N^2 \sigma_{\min}^2((I_N - P_X) V)} 
	\le \frac{L}{a_N^2 l_N^2}.
\end{align*}
Thus, as $N \to \infty$,  
\begin{align*}
	\mathbb{E}\left[\left\|\hat{\beta}_C^{(ols)} - \beta_C\right\|_2^2\right] &= \sigma_y^2\mathbb{E}\left[\text{tr}\left((C^\top(I_N - P_X) C)^{-1} \right)\right]\\
	&\le \sigma_y^2\mathbb{E}\left[\frac{L}{a_N^2 l_N^2}\right] = \frac{\sigma_y^2 L}{a_N^2 l_N^2} \to 0,
\end{align*}
thereby verifying the consistency of $\hat{\beta}_C^{(ols)}$.
\hfill $\square$

\subsection{Proof of Theorem 2:} 
\label{sec: proof2}
Similar to the calculation of \eqref{eq: raw}, we have
\begin{align}
	\tilde{\beta}_X^{(ols)} - \beta_X = (X^\top (I_N - P_{Z})X)^{-1}X^\top (I_N - P_Z)\varepsilon, \nonumber \\
	\tilde{\beta}_Z^{(ols)} - \beta_Z = (Z^\top (I_N - P_{X})Z)^{-1}Z^\top (I_N - P_X)\varepsilon. \label{eq: beta_zols}
\end{align}
Applying a similar proof strategy to $\tilde{\beta}_X^{(ols)}$ gives its consistency and asymptotic normality. Thus, we only need to consider $\tilde{\beta}_Z^{(ols)}$, and divide the proof into three steps:
\begin{enumerate}
	\item Show that $\frac{1}{N} Z^\top P_XZ \stackrel{P}{\longrightarrow} 0$.
	\item For $a_N = \sqrt{N}$, show that there exists a constant $m > 0$ such that $\frac{1}{N}\|Z\|_2^2 \ge m $ a.s..
	\item Show that $\frac{1}{N}Z^\top(I_N - P_X) \varepsilon \stackrel{P}{\longrightarrow} 0$.
\end{enumerate}
With the above three steps, we conclude that as $N \to \infty$
\begin{align*}
	&\left|\tilde{\beta}_Z^{(ols)} - \beta_Z\right| \\
	=& \left|({Z}^\top (I - P_X){Z})^{-1}{Z}^\top  (I_N - P_{X}) \varepsilon\right|\\
	=& \left(\frac{1}{N}{Z}^\top (I - P_X){Z}\right)^{-1} \left|\frac{1}{N}{Z}^\top  (I_N - P_{X}) \varepsilon\right|\\
	\le& \left(m - \frac{1}{N}{Z}^\top P_X {Z}\right)^{-1}\left|\frac{1}{N} {Z}^\top  (I_N - P_{X}) \varepsilon\right|\\
	&\stackrel{P}{\longrightarrow} m^{-1} \cdot 0 = 0,
\end{align*}
showing the consistency of $\tilde{\beta}_Z^{(ols)}$.

\textbf{Step 1: }We start the proof by showing $\frac{1}{N} Z^\top P_XZ \stackrel{P}{\longrightarrow} 0$. Note that 
\begin{align*}
	\frac{1}{N} Z^\top P_X Z = \frac{1}{N}Z^\top X \left(\frac{1}{N}X^\top X\right)^{-1} \frac{1}{N} X^\top Z,
\end{align*}
where 
\begin{align*}
	\frac{1}{N} X^\top Z = \left[\begin{array}{c}
		\frac{1}{N} \sum_{i = 1}^N X_{i1} Z_i   \\
		\vdots \\
		\frac{1}{N} \sum_{i = 1}^N X_{iP} Z_i
	\end{array}\right].
\end{align*}
Recall from the definition of $Z$ (2.12) that $Z$ represents the estimated community-based centrality of nodes, and nodes within the same community share the same value for the centrality measure ${Z}$. Here we rewrite $Z$ corresponding to $R$ communities by $\{{Z}^{(1)}, \cdots, {Z}^{(R)}\}$ i.e. $Z = [{Z}^{(c_1)}, \cdots, {Z}^{(c_N)}]^\top$ where ${Z}^{(c_i)}$ denotes the centrality of community $c_i$ with $c_i \in \{1,\cdots, R\}$ being the community label of node $i$. From the definitions of $Z$ and $U$, we observe that 
\begin{align}
	L\sum_{k = 1}^N Z_k &= \sum_{r = 1}^R N_r {Z}^{(r)} = \sum_{i,j}C_{ij}, \label{eq: thm3_col_sum1}
\end{align}
which further gives
\begin{align}
	N_r {Z}^{(r)} \le \sum_{r = 1}^R N_r {Z}^{(r)} = \frac{1}{L}\sum_{i,j}{C}_{ij} \le \frac{1}{L}\sqrt{NL\sum_{i,j}{C}_{ij}^2} = \frac{\sqrt{N a_N^2}}{\sqrt{L}} = \frac{N}{\sqrt{L}}. \label{eq: thm3_Ztilde}
\end{align}
From Assumption 4, we have $\min_i \frac{N_i}{N} > \epsilon$, thus
\begin{align}
	\label{eq: thm3_30}
	{Z}^{(r)} \le \frac{1}{\sqrt{L}} \frac{N}{N_r} \le \frac{1}{\sqrt{L}} \max_r \frac{N}{N_r} < \frac{1}{\epsilon\sqrt{L}}.
\end{align}
Now we prove that $\frac{1}{N} X^\top Z \stackrel{P}{\longrightarrow} 0$. From \eqref{eq: thm3_30}, entrywisely we have
\begin{align*}
	\left|\frac{1}{N} \sum_{i = 1}^N X_{ij}{Z}_i\right|  &= \left|\frac{1}{N} \sum_{r = 1}^R {Z}^{(r)} \sum_{i = 1}^{N_r} X_{1i}\right|\\
	&\le \sum_{r = 1}^R {Z}^{(r)} \left|\frac{1}{N}\sum_{i = 1}^{N_r} X_{ij}\right| \\
	&\le \sum_{r = 1}^R \frac{1}{\epsilon \sqrt{L}} \left|\frac{1}{N}\sum_{i = 1}^{N_r} X_{ij}\right| \stackrel{P}{\longrightarrow}0
\end{align*}
if $\frac{1}{N}\sum_{i = 1}^{N_r} X_{ij} \stackrel{P}{\longrightarrow}0$. So now we only need to prove $\frac{1}{N}\sum_{i = 1}^{N_r} X_{ij} \stackrel{P}{\longrightarrow}0$. Consider the second moment of $\frac{1}{N}\sum_{i = 1}^{N_r} X_{ij}$: 
\begin{align*}
	\mathbb{E}\left[\left(\frac{1}{N}\sum_{i = 1}^{N_r} X_{ij} \right)^2\right] &= \mathbb{E}\left[\frac{1}{N^2}\mathbb{E}^{N_r}\left[\left(\sum_{i = 1}^{N_r} X_{ij} \right)^2\right]\right]\\
	&= \frac{1}{N^2}\mathbb{E}\left[N_r\mathbb{E}X_{ij}^2\right]\\
	&= \frac{1}{N}\mathbb{E}\left[\frac{N_r}{N}\mathbb{E}X_{ij}^2\right] \le \frac{1}{N} \mathbb{E}X_{ij}^2 \to 0
\end{align*}
where $\frac{N_r}{N} \le 1$. Thus $\frac{1}{N}\sum_{i = 1}^{N_r} X_{ij} \stackrel{L_2}{\longrightarrow}0$ and $\frac{1}{N}\sum_{i = 1}^{N_r} X_{ij} \stackrel{P}{\longrightarrow}0$, which implies
\begin{align}
	\frac{1}{N} X^\top Z \stackrel{P}{\longrightarrow} 0 \label{eq: thm3_X_Z}
\end{align}
and
\begin{align}
	\frac{1}{N} {Z}^\top P_X Z \stackrel{P}{\longrightarrow} 0 \cdot V_X^{-1} \cdot 0 = 0. \label{eq: thm3_31}
\end{align}

\textbf{Step 2: }Now we consider $\frac{1}{N}Z^\top Z$. 
Since $C_{ij} > 0$, with Cauchy-Schwarz inequality, we have 
\begin{align*}
	\frac{1}{N}\|Z\|_2^2 = \frac{1}{N^2}\sum_{k = 1}^N Z_k^2 \sum_{k = 1}^N 1 \ge \frac{1}{N^2}\left(\sum_{k = 1}^N Z_k\right)^2.
\end{align*}
Combining with equation \eqref{eq: thm3_col_sum1} and Assumption 4, we see that 
\begin{align*}
	\frac{1}{N}\|Z\|_2^2 \ge \frac{1}{N^2}\left(\frac{1}{L}\sum_{i,j}C_{ij}\right)^2 \ge \frac{1}{N^2}\frac{1}{L^2} (N\underset{1 \le i \le N}{\min}\sum_{j = 1}^L C_{ij})^2 = \frac{1}{L^2} \underset{1 \le i \le N}{\min}\|C_i\|_1^2 \asymp \frac{a_N^2}{N} = 1.
\end{align*}
Hence, with $a_N = \sqrt{N}$, there exists a constant $m > 0$ such that 
\begin{align}
	\frac{1}{N}\|Z\|_2^2 \ge m , \quad a.s..\label{eq: thm3_Z^2_low}
\end{align}

From the results of \textbf{Step 1} and \textbf{Step 2}, by \eqref{eq: thm3_31} and \eqref{eq: thm3_Z^2_low}, we see that for $N$ sufficiently large,
\begin{align*}
	\frac{1}{N}Z^\top (I_N - P_X){Z} &= \frac{1}{N}{Z}^\top {Z} - \frac{1}{N}{Z}^\top P_X {Z} \\
	&\ge m - \frac{1}{N}{Z}^\top P_X {Z} > 0,
\end{align*}
which gives
\begin{align}
	\left(\frac{1}{N}{Z}^\top (I_N - P_X){Z}\right)^{-1} \le \left(m - \frac{1}{N}{Z}^\top P_X{Z}\right)^{-1}. \label{eq: thm3_33}
\end{align}

\textbf{Step 3: }Now we consider the behavior of 
\begin{align}
	\frac{1}{N}{Z}^\top  (I_N - P_{X}) \varepsilon = \frac{1}{N}{Z}^\top \varepsilon - \frac{1}{N}{Z}^\top P_{X} \varepsilon.\label{eq: thm3_34}
\end{align}
For the first part of RHS of \eqref{eq: thm3_34}, $\frac{1}{N} Z^{\top} \varepsilon \stackrel{P}{\longrightarrow} 0$ follows from the arguments showing $\frac{1}{N} X^{\top} Z \stackrel{P}{\longrightarrow} 0$. Moreover, we have $\frac{1}{N} X^{\top} Z \stackrel{L_2}{\longrightarrow} 0$: 

\begin{align}
	\mathbb{E}\left[\left\| \frac{1}{N}{Z}^\top \varepsilon\right\|_2^2\right]&= \frac{1}{N^2}\mathbb{E}\left[\varepsilon^\top {Z} {Z}^\top \varepsilon\right]\nonumber\\
	&= \frac{\sigma_y^2}{N^2}\mathbb{E}\left[\|{Z} \|_2^2\right].
\end{align}
And then we need to calculate $\mathbb{E}\left[\|Z\|_2^2 \right]$. The randomness of $Z$ arises from both eigenvector centrality and community structure. Therefore, we need to consider $\|Z\|_2^2$ from a different perspective here. From the definition of ${Z}$ in (2.12), we have 
\begin{align*}
	\|Z\|_2^2  &=  \frac{1}{L^2}\|U \operatorname{1}_{L} \|_2^2 \le \frac{1}{L} \|U\|_F^2 .
\end{align*}
By the definition of $U$ in (2.11), we have 
\begin{align*}
	\|U\|_F^2 &= \left\| S\bigl(H \bullet S^\top  C\bigr)\right\|_F^2 \\   
	&= \left\| \text{vec}\left(S\bigl(H \bullet S^\top C\bigr)\right)\right\|_2^2, \\
	\intertext{and since $\text{vec}(A(\omega \bullet B)) = (B^\top \odot A) \omega$ \citep{slyusar1999family},  then} 
	&= \left\|\left(C^\top S \odot S\right)H \right\|_2^2 \le \left\|\left(C^\top S \odot S\right)\right\|_F^2 \left\|H \right\|_2^2.
\end{align*}

It follows from the special structure of $S= [S_1, \cdots, S_R]$ that $\|T \odot S\|_F^2 = \sum_{i = 1}^R N_i \|T_i \|_2^2$ ,  for an $L\times R$ matrix $T = [T_1, \cdots, T_R] $, where $N_i$ denotes the size of community $i$. Then we obtain 
\begin{align*}
	\left\|C^\top S \odot S\right\|_F^2 &= \sum_{i = 1}^R N_i \left\|C^\top S_i\right\|_2^2\\
	& \le \sum_{i = 1}^R N_i \left\|C^\top\right\|_F^2 \left\| S_i\right\|_2^2\\
	& = \sum_{i = 1}^R N_i \left\|C^\top\right\|_F^2 N_i = \sum_{i = 1}^R N_i^2 \left\|C\right\|_F^2.
\end{align*}
Hence, with Assumption 4, we have the following upper bound for $\|Z\|_2^2$:
\begin{align}
	\|Z\|_2^2  & \le \frac{1}{L}\|U\|_F^2 \le \frac{1}{L} \left\|\left(C^\top S \odot S\right)\right\|_F^2 \left\|H \right\|_2^2 \nonumber\\
	& \le \frac{1}{L}\sum_{i = 1}^R N_i^2 \left\|C\right\|_F^2 \sum_{i = 1}^R \frac{1}{N_i^2} \nonumber \\
	&= \frac{1}{L}  \left\|C\right\|_F^2 \sum_{i = 1}^R N_i^2 \sum_{i = 1}^R \frac{1}{N_i^2} = \frac{1}{L}  \left\|C\right\|_F^2 \sum_{i = 1}^R \frac{N_i^2}{N^2} \sum_{i = 1}^R \frac{N^2}{N_i^2} \le \frac{R^2}{L \epsilon^2} \left\|C\right\|_F^2
	\label{eq:Z_upper}
\end{align}
Then taking expectations on both sides of \eqref{eq:Z_upper} gives
\begin{align}
	\mathbb{E}\left[\|Z\|_2^2 \right] & \le \frac{R^2}{L \epsilon^2} \mathbb{E}\left[ \left\| C\right\|_F^2 \right] =  O(a_N^2),\label{eq: thm3_35}
\end{align}
and with \eqref{eq: thm3_35}, we have as $N \to \infty$
\begin{align}
	\mathbb{E}\left[\left\| \frac{1}{N}{Z}^\top \varepsilon\right\|_2^2\right]= \frac{\sigma_y^2}{N^2}\mathbb{E}\left[\|{Z} \|_2^2\right] \le \frac{\sigma_y^2}{N^2} O(a_N^2)\to 0,\label{eq: thm3_36}
\end{align}
i.e. $\frac{1}{N}{Z}^\top \varepsilon \stackrel{L_2}{\longrightarrow} 0$.\\
For the second part of RHS of \eqref{eq: thm3_34}, using \eqref{eq: thm3_X_Z} and the law of large numbers, we have
\begin{align}
	\frac{1}{N}{Z}^\top P_{X} \varepsilon = \frac{1}{N}{Z}^\top X \left(\frac{1}{N}X^\top X\right)^{-1} \frac{1}{N} X^\top \varepsilon \stackrel{P}{\longrightarrow} 0 \cdot V_X^{-1} \cdot 0 = 0,\label{eq: thm3_37}
\end{align}
so that $\frac{1}{N}{Z}^\top P_X \varepsilon \stackrel{P}{\longrightarrow} 0$. 

Therefore, combining \eqref{eq: thm3_36}, and \eqref{eq: thm3_37}, we conclude that $\frac{1}{N}Z^\top(I_N - P_X) \varepsilon \stackrel{P}{\longrightarrow} 0$ and completes the proof of Step 3.
\hfill $\square$
\subsection{Proof of Theorem 3}
Now we consider the situation where measurement errors exist.
Let $\hat{W}_1 = (X, \hat{C})$, then the two-stage estimator $\hat{\beta}$ becomes
\begin{equation*}
	\begin{array}{ll}
		\hat{\beta} &=(\hat{W}_1^\top \hat{W}_1)^{-1}\hat{W}_1^\top y  \\
		&= \left(\begin{array}{cc}X^\top X & X^\top \hat{C} \\ \hat{C}^\top X & \hat{C}^\top \hat{C}\end{array}\right)^{-1} \left(\begin{array}{c}X^\top  \\ \hat{C}^\top \end{array}\right)y\\
		&= \left(\begin{array}{cc}(X^\top (I_N - P_{\hat{C}})X)^{-1}& \boldsymbol{*}_1 \\ \boldsymbol{*}_2 & (\hat{C}^\top (I_N - P_{X})\hat{C})^{-1}
		\end{array}\right)\left(\begin{array}{c}X^\top  \\ \hat{C}^\top \end{array}\right)y,
	\end{array}
\end{equation*}
where 
\begin{equation*}
	\begin{array}{ll}
		\boldsymbol{*}_1 &= -(X^\top (I_N - P_{\hat{C}})X)^{-1}X^\top \hat{C}(\hat{C}^\top \hat{C})^{-1} \\
		&= -(X^\top X)^{-1}X^\top \hat{C}(\hat{C}^\top (I_N - P_{X})\hat{C})^{-1} ,
	\end{array}
\end{equation*}
and 
\begin{equation*}
	\begin{array}{ll}
		\boldsymbol{*}_2 &= -(\hat{C}^\top (I_N - P_{X})\hat{C})^{-1} \hat{C}^\top X(X^\top X)^{-1} \\
		&= -(\hat{C}^\top \hat{C})^{-1}\hat{C}^\top X(X^\top (I_N - P_{\hat{C}})X)^{-1}.
	\end{array}
\end{equation*}
This gives
\begin{equation*}
	\hat{\beta} =\left(\begin{array}{c}
		\hat{\beta}_X\\ 
		\hat{\beta}_C
	\end{array}\right)
	=\left(\begin{array}{c}
		(X^\top (I_N - P_{\hat{C}})X)^{-1} X^\top  (I_N - P_{\hat{C}}) y\\ 
		(\hat{C}^\top (I_N - P_{X})\hat{C})^{-1}\hat{C}^\top (I_N - P_{X}) y
	\end{array}\right).
\end{equation*}
Since we assume $y = X\beta_{X} + C \beta_{C} + \varepsilon$, we have for $\delta_C := C- \hat{C}$,
\begin{equation}
	\hat{\beta} =\left(\begin{array}{c}
		\hat{\beta}_X\\ 
		\hat{\beta}_C
	\end{array}\right)
	=\beta + \left(\begin{array}{c}
		(X^\top (I_N - P_{\hat{C}})X)^{-1} X^\top  (I_N - P_{\hat{C}}) [\delta_C\beta_C + \varepsilon]\\ 
		(\hat{C}^\top (I_N - P_{X})\hat{C})^{-1}\hat{C}^\top (I_N - P_{X}) [\delta_C\beta_C + \varepsilon]
	\end{array}\right). \label{eq: 3.3expre}
\end{equation}

We first consider the consistency of $\hat{\beta}_X - \beta_X$, and observe from~\eqref{eq: 3.3expre} that
\begin{align*}
	\hat{\beta}_X - \beta_X &= (X^\top (I_N - P_{\hat{C}})X)^{-1} X^\top  (I_N - P_{\hat{C}}) [\delta_C\beta_C + \varepsilon]\\
	& = (\frac{1}{N}X^\top (I_N - P_{\hat{C}})X)^{-1} \left[\frac{1}{N}X^\top  (I_N - P_{\hat{C}}) \delta_C\beta_C + \frac{1}{N}X^\top  (I_N - P_{\hat{C}}) \varepsilon \right].
\end{align*}
Hence, as long as we justify the three convergence results below:
\begin{enumerate}
	\item $(\frac{1}{N}X^\top(I - P_{\hat{C}}) X)^{-1} \stackrel{P}{\longrightarrow} V_X^{-1}$,
	\item $\frac{1}{N}X^{\top}(I_N - P_{\hat{C}})\varepsilon \stackrel{P}{\longrightarrow} 0$,
	\item $\frac{1}{N}X^{\top}(I_N - P_{\hat{C}})\delta_C\beta_C \stackrel{P}{\longrightarrow} 0$,
\end{enumerate}
we are able to obtain the consistency of $\hat{\beta}_X$.
\medskip

\textbf{Step 1:} Using a similar proof strategy as for \eqref{eq: proj_1}  gives
\begin{align}
	\frac{1}{N} X^\top P_{\hat{C}} X \stackrel{L_1}{\longrightarrow} 0,\label{eq: thm2_proj_1}
\end{align}
which combined with the law of large numbers leads to
\begin{align}
	\frac{1}{N}X^\top(I - P_{\hat{C}}) X \stackrel{P}{\longrightarrow} V_X. \label{eq: thm2_no_inverse}
\end{align}
Applying the continuous mapping theorem to \eqref{eq: thm2_no_inverse}, we obtain
\begin{align}
	(\frac{1}{N}X^\top(I - P_{\hat{C}}) X)^{-1} \stackrel{P}{\longrightarrow} V_X^{-1}.\label{eq: part2_1}
\end{align}

\textbf{Step 2: }Now we show that $\frac{1}{N}X^\top(I- P_{\hat{C}})\varepsilon \stackrel{P}{\longrightarrow} 0$.
Consider the $\ell_2$-norm:
\begin{align*}
	\mathbb{E}\left[\left\|\frac{1}{N}X^\top P_{\hat{C}} \varepsilon\right\|_2^2 \right] &= \frac{1}{N^2}\mathbb{E}\left[\varepsilon^\top P_{\hat{C}} X X^\top P_{\hat{C}} \varepsilon \right]\\
	&= \frac{1}{N^2}\mathbb{E}\left[\text{tr}(\varepsilon^\top P_{\hat{C}} X X^\top P_{\hat{C}} \varepsilon) \right]\\
	&= \frac{1}{N^2}\mathbb{E}\left[\text{tr}(\varepsilon \varepsilon^\top P_{\hat{C}} X X^\top P_{\hat{C}} ) \right]\\
	&= \frac{1}{N^2}\sigma_y^2\mathbb{E}\left[\text{tr}( P_{\hat{C}} X X^\top P_{\hat{C}} ) \right]\\
	&= \frac{\sigma_y^2}{N^2}\mathbb{E}\left[\text{tr}(X^\top P_{\hat{C}} X) \right]\\
	&= \frac{\sigma_y^2}{N}\text{tr}\left(\mathbb{E}\left[\frac{1}{N}X^\top P_{\hat{C}} X \right]\right)\\
	&\le \frac{\sigma_y^2}{N}\text{tr}\left(\mathbb{E}\left[\left|\frac{1}{N}X^\top P_{\hat{C}} X\right| \right]\right) \to 0,
\end{align*}
where the convergence is given by \eqref{eq: thm2_proj_1} . Therefore, $\frac{1}{N}X^\top P_{\hat{C}} \varepsilon \stackrel{L_2}{\longrightarrow} 0$ and 
\begin{align}
	\frac{1}{N}X^\top P_{\hat{C}} \varepsilon \stackrel{P}{\longrightarrow} 0.\label{eq: part2_2_2}
\end{align}
Then combining the law of large numbers with \eqref{eq: part2_2_2} gives
\begin{align}\label{eq: thm2_part2}
	\frac{1}{N}X^\top (I_N - P_{\hat{C}})\varepsilon  \stackrel{P}{\longrightarrow} 0.
\end{align}

\textbf{Step 3: }For $\frac{1}{N}X^{\top}(I_N - P_{\hat{C}})\delta_C\beta_C$, we again consider its $\ell_2$-norm:
\begin{align*}
	\mathbb{E}\left[\left\|\frac{1}{N}X^{\top}(I_N - P_{\hat{C}})\delta_C\beta_C\right\|_2^2 \right] &= \frac{1}{N^2}\mathbb{E}\left[\beta_C^\top \delta_C^\top(I_n - P_{\hat{C}}) X X^\top (I_N - P_{\hat{C}}) \delta_C\beta_C \right]\\
	&= \frac{1}{N^2}\mathbb{E}\left[\text{tr}(\beta_C^\top \delta_C^\top(I_n - P_{\hat{C}}) X X^\top (I_N - P_{\hat{C}}) \delta_C\beta_C) \right]\\
	&= \frac{1}{N^2}\mathbb{E}\left[\text{tr}(\beta_C\beta_C^\top \delta_C^\top(I_n - P_{\hat{C}}) X X^\top (I_N - P_{\hat{C}}) \delta_C) \right]\\
	\intertext{and applying Lemma \ref{lemma:trace} gives}
	&\le \frac{1}{N^2}\mathbb{E}\left[\text{tr}(\beta_C\beta_C^\top)\text{tr}( \delta_C^\top(I_n - P_{\hat{C}}) X X^\top (I_N - P_{\hat{C}}) \delta_C) \right]\\
	&= \frac{1}{N^2}\text{tr}(\beta_C\beta_C^\top)\mathbb{E}\left[\text{tr}(\delta_C^\top(I_n - P_{\hat{C}}) X X^\top (I_N - P_{\hat{C}}) \delta_C) \right]\\
	&= \frac{1}{N^2}\text{tr}(\beta_C\beta_C^\top)\mathbb{E}\left[\text{tr}( \delta_C\delta_C^\top(I_n - P_{\hat{C}}) X X^\top (I_N - P_{\hat{C}})) \right];\\
	\intertext{since $I_N - P_{\hat{C}}$ is idempotent, applying Lemma \ref{lemma:trace} again leads to}
	&\le \frac{1}{N^2}\text{tr}(\beta_C\beta_C^\top)\mathbb{E}\left[\text{tr}( \delta_C\delta_C^\top)\text{tr}((I_n - P_{\hat{C}}) X X^\top (I_N - P_{\hat{C}})) \right]\\
	&= \frac{1}{N^2}\text{tr}(\beta_C\beta_C^\top)\mathbb{E}\left[\|\delta_C\|_F^2 \text{tr}(X^\top(I_n - P_{\hat{C}}) X ) \right]\\
	&= \text{tr}(\beta_C\beta_C^\top)\left(\frac{1}{N^2}\mathbb{E}\left[\|\delta_C\|_F^2 \text{tr}(X^\top X ) \right] - \frac{1}{N^2}\mathbb{E}\left[\|\delta_C\|_F^2 \text{tr}(X^\top P_{\hat{C}} X ) \right] \right).
\end{align*}
Next, since $\mathbb{E}^{C,E_0}[\frac{1}{N}X_i^\top X_i] < \infty$,  for $i = 1, \cdots, P$, we then have
\begin{align*}
	\frac{1}{N^2}\mathbb{E}\left[\|\delta_C\|_F^2 \text{tr}(X^\top X ) \right] &= \frac{1}{N}\mathbb{E}\left[\|\delta_C\|_F^2 \mathbb{E}^{C,E_0}\left[\text{tr}(\frac{1}{N}X^\top X ) \right]\right]\\
	&= \frac{1}{N}\mathbb{E}\left[\|\delta_C\|_F^2 \text{tr}\left(\mathbb{E}^{C,E_0}\left[\frac{1}{N}X^\top X \right]\right)\right] \\
	&= \frac{1}{N}\mathbb{E}\left[\|\delta_C\|_F^2 \sum_{i = 1}^P\mathbb{E}^{C,E_0}\left[\frac{1}{N}X_i^\top X_i\right]\right].
\end{align*}
Also, we see from Lemma~\ref{lem:rate_C} that $\mathbb{E}\left[\|\delta_C\|_F^2\right] = O(\frac{a_N^2 N}{\delta^2})$, so 
\begin{align}\label{eq: thm2_part2_1}
	\frac{1}{N^2}\mathbb{E}\left[\|\delta_C\|_F^2 \text{tr}(X^\top X ) \right] = \frac{1}{N}O(\frac{a_N^2 N}{\delta^2}) = O(\frac{a_N^2}{\delta^2}).
\end{align}
Similar to \eqref{eq: norm1}, we have for $i,j = 1, \cdots, P$, 
$$
\mathbb{E}^{\hat{C}}\left[\left|\frac{1}{N}X_i^\top P_{\hat{C}} X_j\right| \right] =\mathbb{E}^{C, E_0}\left[\left|\frac{1}{N}X_i^\top P_{\hat{C}} X_j\right| \right]
\le \frac{L}{N}\sqrt{\mathbb{E}^{C, E_0}X_{1i}^2\mathbb{E}^{C, E_0}X_{1j}^2} = O(\frac{1}{N}),
$$ 
then 
\begin{align*}
	\frac{1}{N^2}\mathbb{E}\left[\|\delta_C\|_F^2 \text{tr}(X^\top P_{\hat{C}} X ) \right] &=  \frac{1}{N^2}\mathbb{E}\left[\|\delta_C\|_F^2 \mathbb{E}^{C,E_0}\left[\text{tr}(X^\top P_{\hat{C}} X ) \right] \right]\\
	&= \mathbb{E}\left[\frac{1}{N}\|\delta_C\|_F^2 \mathbb{E}^{C,E_0}\left[\text{tr}(\frac{1}{N}X^\top P_{\hat{C}} X ) \right] \right]\\
	&= \mathbb{E}\left[\frac{1}{N}\|\delta_C\|_F^2 \mathbb{E}^{C, E_0}\left[\text{tr}(\frac{1}{N}X^\top P_{\hat{C}} X ) \right]\right]\\
	&= \mathbb{E}\left[\frac{1}{N}\|\delta_C\|_F^2 O\left(\frac{1}{N}\right) \right]\\
	&= \frac{1}{N}\mathbb{E}\left[\|\delta_C\|_F^2  \right]O\left(\frac{1}{N}\right).
\end{align*}  
Additionally, we obtain from Lemma~\ref{lem:rate_C} that
\begin{align}\label{eq: thm2_part2_2}
	\frac{1}{N^2}\mathbb{E}\left[\|\delta_C\|_F^2 \text{tr}(X^\top P_{\hat{C}} X ) \right] = \frac{1}{N}\mathbb{E}\left[\|\delta_C\|_F^2  \right]O\left(\frac{1}{N}\right) = \frac{1}{N}O(\frac{a_N^2 N}{\delta^2})O\left(\frac{1}{N}\right) = O(\frac{1}{\delta^2}).
\end{align}
Combining \eqref{eq: thm2_part2_1} and \eqref{eq: thm2_part2_2}, we have
\begin{align*}
	\frac{1}{N^2}\mathbb{E}\left[\|\delta_C\|_F^2 \text{tr}(X^\top X ) \right] - \frac{1}{N^2}\mathbb{E}\left[\|\delta_C\|_F^2 \text{tr}(X^\top P_{\hat{C}} X ) \right]  = O(\frac{a_N^2}{\delta^2}) - O(\frac{1}{\delta^2}) = O(\frac{a_N^2}{\delta^2}).
\end{align*}
Provided Assumption 4 holds, then as $N \to \infty$,
\begin{align*}
	\mathbb{E}\left[\left\|\frac{1}{N}X^{\top}(I_N - P_{\hat{C}})\delta_C\beta_C\right\|_2^2 \right] &\le \text{tr}(\beta_C\beta_C^\top)\left(\frac{1}{N^2}\mathbb{E}\left[\|\delta_C\|_F^2 \text{tr}(X^\top X ) \right] - \frac{1}{N^2}\mathbb{E}\left[\|\delta_C\|_F^2 \text{tr}(X^\top P_{\hat{C}} X ) \right] \right)\\
	&= \text{tr}(\beta_C\beta_C^\top) O(\frac{a_N^2}{\delta^2}) \to 0,
\end{align*}
which implies $\frac{1}{N}X^{\top}(I_N - P_{\hat{C}})\delta_C\beta_C \stackrel{L_2}{\longrightarrow} 0$, so that$\frac{1}{N}X^{\top}(I_N - P_{\hat{C}})\delta_C\beta_C \stackrel{P}{\longrightarrow} 0$. 
This completes the proof of Step 3.

\hfill $\square$

\subsection{Proof of Theorem 4}
\label{sec:A.6}
Similar to the calculation procedure of \eqref{eq: 3.3expre} in Theorem 3, we denote $\delta_Z = Z - \hat{Z}$, and 
\begin{equation*}
	\tilde{\beta} =\left(\begin{array}{c}
		\tilde{\beta}_{X}\\ 
		\tilde{\beta}_Z
	\end{array}\right)
	=\left(\begin{array}{c}
		(X^\top (I_N - P_{\hat{Z}})X)^{-1} X^\top  (I_N - P_{\hat{Z}}) y\\ 
		(\hat{Z}^\top (I_N - P_{X})\hat{Z})^{-1}\hat{Z}^\top (I_N - P_{X}) y
	\end{array}\right).
\end{equation*}
From the regression model $y = X\beta_{X} + Z\beta_Z +\varepsilon$, we have 
\begin{equation*}
	\tilde{\beta} =\left(\begin{array}{c}
		\tilde{\beta}_X\\ 
		\tilde{\beta}_Z
	\end{array}\right)
	=\beta + \left(\begin{array}{c}(X^\top (I_N - P_{\hat{Z}})X)^{-1}X^\top (I_N - P_{\hat{Z}}) [\delta_Z\beta_Z + \varepsilon]\\ 
		(\hat{Z}^\top(I_N - P_{X}) \hat{Z})^{-1}\hat{Z}^\top  (I_N - P_{X})[\delta_Z\beta_Z + \varepsilon]
	\end{array}\right).
\end{equation*}

First, we show the consistency of $\tilde{\beta}_X$. Since 
\begin{align*}
	\tilde{\beta}_X - \beta_X = (X^\top (I_N - P_{\hat{Z}})X)^{-1}X^\top (I_N - P_{\hat{Z}}) [\delta_Z\beta_Z + \varepsilon],
\end{align*}
then similar to the proof of \eqref{eq: part2_1} and \eqref{eq: thm2_part2}, we have 
\begin{align}\label{eq: thm3_part1}
	(\frac{1}{N}X^\top (I_N - P_{\hat{Z}})X)^{-1}\stackrel{P}{\longrightarrow} V_X^{-1},
\end{align}
and 
\begin{align}\label{eq: thm3_part2}
	\frac{1}{N}X^\top (I_N - P_{\hat{Z}})  \varepsilon \stackrel{P}{\longrightarrow} 0.
\end{align}
Thus, it suffices to prove 
\begin{align}\label{eq: thm3_part3}
	\frac{1}{N}X^\top (I_N - P_{\hat{Z}}) \delta_Z\beta_Z \stackrel{P}{\longrightarrow} 0.
\end{align}
Here we prove \eqref{eq: thm3_part3} by computing the $\ell_2$-norm of $\frac{1}{N}X^\top (I_N - P_{\hat{Z}}) \delta_Z\beta_Z$:
\begin{align}
	\mathbb{E}\left[\left\|\frac{1}{N}X^\top (I_N - P_{\hat{Z}}) \delta_Z\beta_Z\right\|_2^2\right] &=\frac{1}{N^2}\mathbb{E}\left[\beta_Z^\top \delta_Z^\top (I_N - P_{\hat{Z}})XX^\top (I_N - P_{\hat{Z}})\delta_Z \beta_Z\right] \nonumber\\
	=& \beta_Z^2\frac{1}{N^2} \mathbb{E}\left[\text{tr}(\delta_Z^\top (I_N - P_{\hat{Z}})XX^\top (I_N - P_{\hat{Z}})\delta_Z) \right]\nonumber\\
	=& \beta_Z^2\frac{1}{N^2} \mathbb{E}\left[\text{tr}(\delta_Z \delta_Z^\top (I_N - P_{\hat{Z}})XX^\top (I_N - P_{\hat{Z}}))\right],\nonumber\\
	\intertext{and by Lemma \ref{lemma:trace}, we have the upper bound}
	\le & \beta_Z^2\frac{1}{N^2} \mathbb{E}\left[\text{tr}(\delta_Z \delta_Z^\top) \text{tr}((I_N - P_{\hat{Z}})XX^\top (I_N - P_{\hat{Z}}))\right] \nonumber\\
	=&  \beta_Z^2 \frac{1}{N^2}\mathbb{E}\left[\text{tr}(\delta_Z \delta_Z^\top) \text{tr}((I_N - P_{\hat{Z}})XX^\top )\right]\nonumber\\
	=& \beta_Z^2\frac{1}{N^2} \mathbb{E}\left[\|\delta_Z\|_F^2 \text{tr}(X^\top(I_N - P_{\hat{Z}})X )\right]\nonumber\\
	=& \beta_Z^2\frac{1}{N^2} (\mathbb{E}\left[\|\delta_Z\|_F^2 \text{tr}(X^\top X)\right] - \mathbb{E}\left[\|\delta_Z\|_F^2 \text{tr}(X^\top P_{\hat{Z}}X)\right]).\label{eq: l2}
\end{align}
Now we first calculate $\mathbb{E}\left[\|\delta_Z\|_F^2 \right]$. From the definition of $\hat{Z}$, we see that 
\begin{align*}
	\|\delta_Z\|_F^2  &=  \frac{1}{L^2}\|(\hat{U} -U) \operatorname{1}_{L} \|_2^2 \le \frac{1}{L} \|\hat{U} -U\|_F^2 .
\end{align*}
Similar to \eqref{eq:Z_upper}, replacing $U$ and $C$ with $\hat{U} -U$, and $\hat{C} -C$ respectively yields
\begin{align}
	\|\delta_Z\|_F^2  & \le \frac{1}{L}\|\hat{U} -U\|_F^2 \nonumber\\
	& \le \frac{1}{L} \left\|\left((\hat{C} - C)^\top S \odot S\right)\right\|_F^2 \left\|H \right\|_2^2 \nonumber\\
	& \le \frac{1}{L}\sum_{i = 1}^R N_i^2 \left\|\hat{C} - C\right\|_F^2 \sum_{i = 1}^R \frac{1}{N_i^2} \nonumber\\
	& = \frac{1}{L}  \left\|\hat{C} - C\right\|_F^2 \sum_{i = 1}^R N_i^2 \sum_{i = 1}^R \frac{1}{N_i^2} \nonumber \\
	& = \frac{1}{L}  \left\|\hat{C} - C\right\|_F^2 \sum_{i = 1}^R \frac{N_i^2}{N^2} \sum_{i = 1}^R \frac{N^2}{N_i^2} \le \frac{R^2}{L \epsilon^2}  \left\|\hat{C} - C\right\|_F^2
	.\label{eq: thm3_deltaZ}
\end{align}
Then by Lemma~\ref{lem:rate_C}, taking expectations on both sides of \eqref{eq: thm3_deltaZ} gives
\begin{align}
	\mathbb{E}\left[\|\delta_Z\|_F^2 \right] & \le \frac{R^2}{L \epsilon^2} \mathbb{E}\left[ \left\|\hat{C} - C\right\|_F^2 \right]\nonumber\\
	&= O(\frac{a_N^2 N}{\delta^2}).\label{eq: thm3_part1_3}
\end{align}

Given the upper bound in \eqref{eq: thm3_part1_3}, we return to the $\ell_2$-norm of $\frac{1}{N}X^\top (I_N - P_{\hat{Z}}) \delta_Z\beta_Z$ as in \eqref{eq: l2}. Since $Z$ and $\hat{Z}$ are independent of $X$, we have
\begin{align*}
	\frac{1}{N^2} \mathbb{E}\left[\|\delta_Z\|_F^2 \text{tr}(X^\top X)\right] 
	&= \mathbb{E}\left[\frac{1}{N} \|\delta_Z\|_F^2 \text{tr}\left(\frac{1}{N} X^\top X\right)\right]\\
	&= \mathbb{E}\left[\frac{1}{N} \|\delta_Z\|_F^2 \mathbb{E}^{Z, \hat{Z}}\left[\text{tr}\left(\frac{1}{N} X^\top X\right)\right]\right]\\
	&= \mathbb{E}\left[\frac{1}{N} \|\delta_Z\|_F^2 \mathbb{E}^{C, \hat{C}}\left[\text{tr}\left(\frac{1}{N} X^\top X\right)\right]\right]
\end{align*}
Since $\mathbb{E}^{C, E_0}\left[\frac{1}{N} X_i^\top X_i\right] < \infty$, then
\[
\frac{1}{N^2} \mathbb{E}\left[\|\delta_Z\|_F^2 \text{tr}(X^\top X)\right]= O\left(\frac{a_N^2}{\delta^2}\right).
\]
In addition, we see that 
\begin{align*}
	\frac{1}{N^2} \mathbb{E}\left[\|\delta_Z\|_F^2 \text{tr}(X^\top P_{\hat{Z}} X)\right] &= \frac{1}{N} \mathbb{E}\left[\|\delta_Z\|_F^2 \mathbb{E}^{\hat{Z},Z}\left[\text{tr}\left(\frac{1}{N}X^\top P_{\hat{Z}} X\right)\right] \right]\\
	&= \frac{1}{N} \mathbb{E}\left[\|\delta_Z\|_F^2 \text{tr}\left(\mathbb{E}^{Z, \hat{Z}}\left[\frac{1}{N}X^\top P_{\hat{Z}} X\right]\right) \right]\\
	&= \frac{1}{N} \mathbb{E}\left[\|\delta_Z\|_F^2 \sum_{i = 1}^P\left(\mathbb{E}^{Z,\hat{Z}}\left[\frac{1}{N}X_i^\top P_{\hat{Z}} X_i\right]\right) \right]\\
	\intertext{and $\mathbb{E}^{Z, \hat{Z}}\left[\frac{1}{N}X_i^\top P_{\hat{Z}} X_j\right] = O(\frac{1}{N})$ gives}
	&= \frac{1}{N} \mathbb{E}\left[\|\delta_Z\|_F^2 O(\frac{1}{N}) \right]\\
	&= \frac{1}{N} O(\frac{a_N^2 N}{\delta^2}\frac{1}{N}) = O(\frac{a_N^2}{N \delta^2}).
\end{align*}
Hence, when Assumption 4 holds, 
\begin{align*}
	&\mathbb{E}\left[\left\|\frac{1}{N}X^\top (I_N - P_{\hat{Z}}) \delta_Z\beta_Z\right\|_2^2\right]\\
	\le& \beta_Z^2 \frac{1}{N^2}\left(\mathbb{E}\left[\|\delta_Z\|_F^2 \text{tr}(X^\top X)\right] - \mathbb{E}\left[\|\delta_Z\|_F^2 \text{tr}(X^\top P_{\hat{Z}} X)\right]\right)\\
	=& O(\frac{a_N^2}{\delta^2}) - O(\frac{a_N^2}{N\delta^2}) = O(\frac{a_N^2}{\delta^2})\to 0,
\end{align*}
which shows $\frac{1}{N}X^\top (I_N - P_{\hat{Z}}) \delta_Z\beta_Z \stackrel{L_2}{\longrightarrow} 0$ and then $\frac{1}{N}X^\top (I_N - P_{\hat{Z}}) \delta_Z\beta_Z \stackrel{P}{\longrightarrow} 0$, completing the proof of \eqref{eq: thm3_part3}.  Finally, combining \eqref{eq: thm3_part1}, \eqref{eq: thm3_part2} and \eqref{eq: thm3_part3}, we obtain
\begin{align*}
	\tilde{\beta}_X - \beta_X &= (X^\top (I_N - P_{\hat{Z}})X)^{-1}X^\top (I_N - P_{\hat{Z}}) [\delta_Z\beta_Z + \varepsilon]\\
	&= (\frac{1}{N}X^\top (I_N - P_{\hat{Z}})X)^{-1}\frac{1}{N}X^\top (I_N - P_{\hat{Z}}) [\delta_Z\beta_Z + \varepsilon]\\
	& \stackrel{P}{\longrightarrow} V_X^{-1}(0+0) = 0,
\end{align*}
showing the
consistency of $\tilde{\beta}_X$. 

\medskip

Now we consider the consistency of $\tilde{\beta}_Z$. Note that 
\begin{align*}
	\tilde{\beta}_Z - \beta_Z = (\hat{Z}^\top (I - P_X)\hat{Z})^{-1}\hat{Z}^\top  (I_N - P_{X})[\delta_{Z}\beta_{Z} + \varepsilon],
\end{align*}
and we divide the proof of the consistency of $\tilde{\beta}_Z - \beta_Z$ to 2 steps:
\begin{itemize}
	\item \textbf{Step 1:} Prove that with $N$ sufficiently large, we have 
	\begin{align*}
		\left(\frac{1}{N}\hat{Z}^\top (I_N - P_X)\hat{Z}\right)^{-1} &\le \left(m - \frac{2}{N}\frac{1}{L} a_N^2 \left\|\tilde{u}_1 - u_1\right\|_2 \sum_{i = 1}^R N_i^2 \sum_{i = 1}^R \frac{1}{N_i^2} - \frac{1}{N}\hat{Z}^\top P_X\hat{Z}\right)^{-1}\\
		&\stackrel{P}{\longrightarrow} m^{-1}
	\end{align*}
	where $m$ is a positive constant. 
	\item \textbf{Step 2:} Prove that $\frac{1}{N}\hat{Z}^\top  (I_N - P_{X})[\delta_{Z}\beta_{Z} + \varepsilon] \stackrel{P}{\longrightarrow} 0$.
\end{itemize}
Assembling the results of \textbf{Step 1} and \textbf{Step 2}, we obtain that as $N \to \infty$,
\begin{align*}
	&\left|\tilde{\beta}_Z - \beta_Z \right|\\
	=& \left|(\hat{Z}^\top (I - P_X)\hat{Z})^{-1}\hat{Z}^\top  (I_N - P_{X})[\delta_{Z}\beta_{Z} + \varepsilon]\right|\\
	=& \left(\frac{1}{N}\hat{Z}^\top (I - P_X)\hat{Z}\right)^{-1} \left|\frac{1}{N}\hat{Z}^\top  (I_N - P_{X})[\delta_{Z}\beta_{Z} + \varepsilon]\right|\\
	\le& \left(m - \frac{2}{N}\frac{1}{L} a_N^2 \left\|\tilde{u}_1 - u_1\right\|_2 \sum_{i = 1}^R N_i^2 \sum_{i = 1}^R \frac{1}{N_i^2} - \frac{1}{N}\hat{Z}^\top P_X\hat{Z}\right)^{-1}\left|\frac{1}{N}\hat{Z}^\top  (I_N - P_{X})[\delta_{Z}\beta_{Z} + \varepsilon]\right|\\
	&\stackrel{P}{\longrightarrow} m^{-1} \cdot 0 = 0,
\end{align*}
which gives the consistency of $\tilde{\beta}_Z$.

\textbf{Step 1: }Given the analogous properties of $Z$ and $\hat{Z}$, similar to \eqref{eq: thm3_X_Z} and \eqref{eq: thm3_31}, we have 
\begin{align}
	\frac{1}{N}\hat{Z}^\top X \stackrel{P}{\longrightarrow} 0. \label{eq: thm3_52}
\end{align}
and
\begin{align}
	\frac{1}{N}\hat{Z}^\top P_X \hat{Z} \stackrel{P}{\longrightarrow} 0. \label{eq: thm3_53}
\end{align}
Then we focus on $\frac{1}{N}\hat{Z}^\top \hat{Z}$. Since $\frac{1}{N}\delta_Z^\top \delta_Z \ge 0$, we see that 
\begin{align}
	\frac{1}{N}\hat{Z}^\top \hat{Z} = \frac{1}{N}Z^\top Z -\frac{2}{N} Z^\top \delta_Z +\frac{1}{N}\delta_Z^\top \delta_Z \ge \frac{1}{N}Z^\top Z -\frac{2}{N} Z^\top \delta_Z.  \label{eq: thm3_deltaZ_hat}
\end{align}
Then with the upper bound of $\|Z\|_2^2 $ and $\|\delta_Z\|_2^2$ derived in \eqref{eq:Z_upper} and \eqref{eq: thm3_deltaZ}, we have
\begin{align}
	\frac{2}{N} Z^\top \delta_Z \le \frac{2}{N} \|Z\|_2 \|\delta_Z\|_2 &\le \frac{2}{N} \frac{1}{L} \|C\|_F \left\|\hat{C} - C\right\|_F \sum_{i = 1}^R N_i^2 \sum_{i = 1}^R \frac{1}{N_i^2} \nonumber \\
	& \le \frac{2}{N} \frac{1}{L} \|C\|_F \left\|\hat{C} - C\right\|_F \frac{R^2}{\epsilon^2} \nonumber \\
	& = \frac{2}{N}\frac{1}{L} a_N^2 \left\|\tilde{u}_1 - u_1\right\|_2 \frac{R^2}{\epsilon^2}. \label{eq: thm3_cross_term_2}
\end{align}
Plugging \eqref{eq: thm3_Z^2_low} and \eqref{eq: thm3_cross_term_2} into \eqref{eq: thm3_deltaZ_hat}, we see that 
\begin{align}
	\frac{1}{N}\hat{Z}^\top \hat{Z} &\ge m - \frac{2}{N}\frac{1}{L} \|C\|_F \left\|\hat{C} - C\right\|_F \frac{R^2}{\epsilon^2} \nonumber\\
	&= m - \frac{2}{N}\frac{1}{L} a_N \left\|\hat{C} - C\right\|_F \frac{R^2}{\epsilon^2}\nonumber\\
	&= m - \frac{2}{N}\frac{1}{L} a_N^2 \left\|\tilde{u}_1 - u_1\right\|_2 \frac{R^2}{\epsilon^2},\label{eq: thm3_46}
\end{align}
where $\tilde{u}_1$ and $u_1$ are as defined in Lemma~\ref{lem:davis_Kahan}. To derive a positive lower bound for$\frac{1}{N}\hat{Z}^\top \hat{Z}$, i.e. show that the RHS of \eqref{eq: thm3_46} is greater than zero for $N$  sufficiently large, we need to show that 
$$\frac{2}{N}\frac{1}{L} a_N^2 \left\|\tilde{u}_1 - u_1\right\|_2 \frac{R^2}{\epsilon^2} \stackrel{P}{\longrightarrow} 0.$$
Since $\mathbb{E}\left[ \frac{2}{N}\frac{1}{L} a_N^2 \left\|\tilde{u}_1 - u_1\right\|_2 \frac{R^2}{\epsilon^2}\right] = O(\frac{a_N^2 \sqrt{N}}{N \delta}) = O(\frac{\sqrt{N}}{\delta}) \to 0$, we have $$ \frac{2}{N}\frac{1}{L} a_N^2 \left\|\tilde{u}_1 - u_1\right\|_2 \frac{R^2}{\epsilon^2} \stackrel{L_1}{\longrightarrow} 0$$ and 
\begin{align}
	\frac{2}{N}\frac{1}{L} a_N^2 \left\|\tilde{u}_1 - u_1\right\|_2 \frac{R^2}{\epsilon^2} \stackrel{P}{\longrightarrow} 0. \label{eq: thm3_74}
\end{align}
Then with \eqref{eq: thm3_53}, \eqref{eq: thm3_46} and \eqref{eq: thm3_74}, for $N$ large enough, we have 
\begin{align*}
	\frac{1}{N}\hat{Z}^\top (I_N - P_X)\hat{Z} &= \frac{1}{N}\hat{Z}^\top \hat{Z} - \frac{1}{N}\hat{Z}^\top P_X\hat{Z} \\
	&\ge m - \frac{2}{N}\frac{1}{L} a_N^2 \left\|\tilde{u}_1 - u_1\right\|_2 \sum_{i = 1}^R N_i^2 \sum_{i = 1}^R \frac{1}{N_i^2} - \frac{1}{N}\hat{Z}^\top P_X\hat{Z} > 0,
\end{align*}
which implies
\begin{align}
	\left(\frac{1}{N}\hat{Z}^\top (I_N - P_X)\hat{Z}\right)^{-1} &\le \left(m - \frac{2}{N}\frac{1}{L} a_N^2 \left\|\tilde{u}_1 - u_1\right\|_2 \sum_{i = 1}^R N_i^2 \sum_{i = 1}^R \frac{1}{N_i^2} - \frac{1}{N}\hat{Z}^\top P_X\hat{Z}\right)^{-1} \nonumber\\
	&\stackrel{P}{\longrightarrow} m^{-1}. \label{eq: thm3_81}
\end{align}

\textbf{Step 2: }Now we consider 
\begin{align}
	\frac{1}{N}\hat{Z}^\top  (I_N - P_{X})[\delta_{Z}\beta_{Z} + \varepsilon] = \frac{1}{N}\hat{Z}^\top \delta_{Z}\beta_{Z} - \frac{1}{N}\hat{Z}^\top P_{X}\delta_{Z}\beta_{Z} + \frac{1}{N}\hat{Z}^\top \varepsilon - \frac{1}{N}\hat{Z}^\top P_{X} \varepsilon.\label{eq: thm3_78}
\end{align}

For the first part of RHS of \eqref{eq: thm3_78}, the Cauchy-Schwarz inequality gives the upper bound that
\begin{align*}
	\mathbb{E}\left[\left| \frac{1}{N}\hat{Z}^\top \delta_{Z}\beta_{Z} \right|\right] &= \frac{\beta_{Z} }{N}\mathbb{E}\left[\left| \hat{Z}^\top \delta_{Z}\right|\right]\\
	&\le \frac{\beta_{Z} }{N}\left(\mathbb{E}\left[\left\| \hat{Z} \right\|_2^2\right]\right)^{\frac{1}{2}}\left(\mathbb{E}\left[\left\|\delta_{Z} \right\|_2^2\right]\right)^{\frac{1}{2}}.
\end{align*}
Similar to the calculation leading to \eqref{eq:Z_upper}, replacing $Z$ and $C$ with $\hat{Z}$ and $\hat{C}$ respectively, we have
\begin{align}
	\|\hat{Z}\|_2^2 \le \frac{1}{L} \|\hat{C}\|_F^2 \sum_{i = 1}^R N_i^2 \sum_{i = 1}^R \frac{1}{N_i^2} \le \frac{R^2}{L \epsilon^2} \|\hat{C}\|_F^2, \label{eq: thm3_upper_bound_Z_hat}
\end{align}
and
\begin{align}
	\mathbb{E}\left[\|\hat{Z}\|_2^2\right] \le \mathbb{E}\left[\frac{R^2}{L \epsilon^2} \|\hat{C}\|_F^2 \right] = O(a_N^2). \label{eq: thm3_80}
\end{align}
Then combining \eqref{eq: thm3_part1_3} and \eqref{eq: thm3_80}, provided Assumption 4 holds, we have as $N \to 0$, 
\begin{align}
	\mathbb{E}\left[\left| \frac{1}{N}\hat{Z}^\top \delta_{Z}\beta_{Z} \right|\right] &\le  \frac{\beta_{Z} }{N} O(a_N) O(\frac{a_N \sqrt{N}}{\delta}) = O(\frac{\sqrt{N}}{\delta}) \to 0, \label{eq: thm3_part2_1}
\end{align}
i.e. $\frac{1}{N}\hat{Z}^\top \delta_{Z}\beta_{Z} \stackrel{L_1}{\longrightarrow} 0$.\\
For the second part in the RHS of \eqref{eq: thm3_78}, also using the Cauchy-Schwarz inequality and Lemma \ref{lemma:trace}, we have the following upper bound:
\begin{align*}
	\mathbb{E}\left[\left| \frac{1}{N}\hat{Z}^\top P_{X}\delta_{Z}\beta_{Z} \right|\right] &= \frac{\beta_{Z} }{N}\mathbb{E}\left[\left| \hat{Z}^\top P_{X} \delta_{Z}\right|\right]\\
	&\le \frac{\beta_{Z} }{N}\left(\mathbb{E}\left[\left\| P_{X}\hat{Z} \right\|_2^2\right]\right)^{\frac{1}{2}}\left(\mathbb{E}\left[\left\|\delta_{Z} \right\|_2^2\right]\right)^{\frac{1}{2}}.
\end{align*}
From \eqref{eq: thm3_80}, we see that 
\begin{align*}
	\mathbb{E}\left[\left\| P_{X}\hat{Z} \right\|_2^2\right]&= \mathbb{E}\left[\hat{Z}^\top  P_{X}\hat{Z} \right]
	= \mathbb{E}\left[\text{tr}\left(\hat{Z}^\top  P_{X}\hat{Z}\right) \right]
	= \mathbb{E}\left[\text{tr}\left(\hat{Z}\hat{Z}^\top  P_{X}\right) \right]\\
	&\le \mathbb{E}\left[\text{tr}\left(\hat{Z}\hat{Z}^\top\right)\text{tr}\left(P_{X}\right) \right] = P\mathbb{E}\left[\|\hat{Z}\|_2^2\right] \le O(a_N^2),
\end{align*}
and as $N \to \infty$, 
\begin{align}
	\mathbb{E}\left[\left| \frac{1}{N}\hat{Z}^\top P_{X}\delta_{Z}\beta_{Z} \right|\right] &= \frac{\beta_{Z} }{N} O(a_N)O(\frac{a_N \sqrt{N}}{\delta}) = O\left(\frac{\sqrt{N}}{\delta}\right) \to 0.\label{eq: thm3_88}
\end{align}

For the third part in the RHS of \eqref{eq: thm3_78}, also from Lemma \ref{lemma:trace} and \eqref{eq: thm3_80} we have as $N \to \infty$,
\begin{align}
	\mathbb{E}\left[\left\| \frac{1}{N}\hat{Z}^\top \varepsilon\right\|_2^2\right]&= \frac{1}{N^2}\mathbb{E}\left[\varepsilon^\top \hat{Z} \hat{Z}^\top \varepsilon\right]\nonumber\\
	&= \frac{\sigma_y^2}{N^2}\mathbb{E}\left[\|\hat{Z} \|_2^2\right] =O(\frac{a_N^2}{N^2}) \to 0.\label{eq: thm3_90}
\end{align}
Finally,
for the fourth part in the RHS of \eqref{eq: thm3_78}, with \eqref{eq: thm3_52} and the law of large numbers, we have
\begin{align}
	\frac{1}{N}\hat{Z}^\top P_{X} \varepsilon = \frac{1}{N}\hat{Z}^\top X \left(\frac{1}{N}X^\top X\right)^{-1} \frac{1}{N} X^\top \varepsilon \stackrel{P}{\longrightarrow} 0 \cdot V_X^{-1} \cdot 0 = 0.\label{eq: thm3_91}
\end{align}
By synthesizing equations \eqref{eq: thm3_part2_1}, \eqref{eq: thm3_88}, \eqref{eq: thm3_90} and \eqref{eq: thm3_91}, we can demonstrate the convergence of \eqref{eq: thm3_78} which completes the proof of Step 2.

\hfill $\square$

\subsection{Proof of Theorem 5}

\textbf{1. Noiseless Network Case}

The OLS estimator is:  
\[
\tilde{\beta}_Z^{(ols)} - \beta_Z = \left(Z^\top (I_N - P_X) Z\right)^{-1} Z^\top (I_N - P_X) \varepsilon.
\]  
Normalizing by \(\sqrt{Z^\top Z / \sigma_y^2}\), we have
\begin{align*}
   \sqrt{\frac{Z^\top Z}{\sigma_y^2}} \left( \tilde{\beta}_Z^{(ols)} - \beta_Z \right) = &\underbrace{\frac{Z^\top Z}{N} \left( \frac{Z^\top (I_N - P_X) Z}{N} \right)^{-1}}_{(A)} \cdot \underbrace{\frac{Z^\top \varepsilon}{\sqrt{Z^\top Z \sigma_y^2}}}_{(B)} \\
    &\quad - \frac{Z^\top Z}{N} \left( \frac{Z^\top (I_N - P_X) Z}{N} \right)^{-1} \cdot \underbrace{\frac{Z^\top P_X \varepsilon}{\sqrt{Z^\top Z \sigma_y^2}}}_{(C)}.
\end{align*}
We only need to prove that (A) $\stackrel{P}{\longrightarrow} 1$, (B) $\stackrel{d}{\longrightarrow} \mathcal{N}(0,1)$ and (C) $\stackrel{d}{\longrightarrow} 0$.
For the term (A), with \eqref{eq: thm3_31}, we obtain 
\[
(A) = \frac{Z^\top (I_N - P_X) Z}{N} \left( \frac{Z^\top Z}{N} \right)^{-1} \stackrel{P}{\longrightarrow} 1.
\]

The term (B) is specified by
\[
(B) = \frac{Z^\top \varepsilon}{\sqrt{Z^\top Z \sigma_y^2}} = \sum_{i=1}^N \frac{Z_i \varepsilon_i}{\sqrt{Z^\top Z \sigma_y^2}}.
\]
Define \(W_{N,i} = \frac{Z_i \varepsilon_i}{\sqrt{Z^\top Z \sigma_y^2}}\), then:
\[
\mathbb{E}\left[ W_{N,i} \,|\, Z \right] = 0, \quad \sum_{i=1}^N \mathbb{E}\left[ W_{N,i}^2 \,|\, Z \right] = \frac{\sum_{i=1}^N Z_i^2 \sigma_y^2}{Z^\top Z \sigma_y^2} = 1.
\]

To apply the conditional central limit theorem \citep{yuan2014conditional}, we 
need to show
\[
\forall \epsilon > 0, \quad \sum_{i=1}^N \mathbb{E}\left[ W_{N,i}^2 \cdot 1_{\{|W_{N,i}| > \epsilon\}} \,|\, Z \right] \stackrel{a.s.}{\longrightarrow} 0.
\]

For any \(\epsilon_0 > 0\),  
\[
1_{\{|W_{N,i}| > \epsilon_0\}} \leq \frac{W_{N,i}^2}{\epsilon_0^2}.
\]
Thus,
\[
\mathbb{E}\left[ W_{N,i}^2 \cdot 1_{\{|W_{N,i}| > \epsilon_0\}} \,|\, Z \right] \leq \frac{\mathbb{E}\left[ W_{N,i}^4 \,|\, Z \right]}{\epsilon_0^2}.
\]
From Assumption 1, \(\mathbb{E}[\varepsilon_i^4] \leq k_0\), then
\[
\mathbb{E}\left[ W_{N,i}^4 \,|\, Z \right] = \frac{Z_i^4 \mathbb{E}[\varepsilon_i^4]}{(Z^\top Z \sigma_y^2)^2} \leq \frac{k_0 Z_i^4}{(Z^\top Z \sigma_y^2)^2}.
\]
Sum over \(i\):
\[
\sum_{i=1}^N \mathbb{E}\left[ W_{N,i}^4 \,|\, Z \right] \leq \frac{k_0 \sum_{i=1}^N Z_i^4}{(Z^\top Z \sigma_y^2)^2}.
\]
Since
\[
\sum_{i=1}^N Z_i^4 \leq \left(\max_{1 \le i \le N} Z_i^2\right) \sum_{i=1}^N Z_i^2 = \left(\max_{1 \le i \le N} Z_i^2\right) Z^\top Z,
\]
we have
\[
\frac{k_0 \sum_{i=1}^N Z_i^4}{(Z^\top Z \sigma_y^2)^2} \le  \frac{k_0 \left(\max_{1 \le i \le N} Z_i^2\right) }{Z^\top Z \sigma_y^4} = \frac{k_0}{\sigma_y^4} \cdot \frac{\left(\max_{1 \le i \le N} Z_i^2\right) }{Z^\top Z}.
\]
From Assumption 4, for each community \( r \), \( N_r \geq \epsilon N \). The condition \( Z^\top Z > N_r Z_i^2 \) implies:
\[
Z_i^2 < \frac{Z^\top Z}{N_r} \quad \forall i \in \text{community } r.
\]
Thus, there exists \(r_0\) such that \(\max_{1 \le i \le N} Z_i^2 < \frac{Z^\top Z}{N_{r_0}}\). Further,
\[
\frac{k_0 \sum_{i=1}^N Z_i^4}{(Z^\top Z \sigma_y^2)^2} \le \frac{k_0}{\sigma_y^4} \cdot \frac{\left(\max_{1 \le i \le N} Z_i^2\right) }{Z^\top Z} < \frac{k_0}{\sigma_y^4} \frac{1}{N_{r_0}} \le \frac{k_0}{\sigma_y^4} \frac{1}{\epsilon N} \stackrel{a.s.}{\longrightarrow} 0.
\]
Finally, we have 
\[
\sum_{i=1}^N \mathbb{E}\left[ W_{N,i}^2 \cdot 1_{\{|W_{N,i}| > \epsilon_0\}} \,|\, Z \right] \leq \frac{k_0}{\sigma_y^4} \frac{1}{\epsilon N} \stackrel{a.s.}{\longrightarrow} 0,
\]
which implies that the Lindeberg condition holds almost surely given \(Z\). By the conditional Lindeberg-Feller CLT, we obtain the convergence of the conditional distribution over \(Z\), i.e.,  
\[
\forall t \in \mathbb{R}, \quad \mathbb{P}\left( \frac{Z^\top \varepsilon}{\sqrt{Z^\top Z \sigma_y^2}} \leq t \, \bigg| \, Z \right) \stackrel{P}{\longrightarrow} \Phi(t).
\]  
Here, convergence in probability of the random probability measures means convergence in probability in the space $PM(\mathbb{R})$ of probability measures on $\mathbb{R}$ metrized by weak convergence. This implies that for almost every realization of \(Z\), the conditional distribution converges to the standard normal distribution as \(N \to \infty\). 

To extend this to the unconditional distribution, we integrate over \(Z\):  
\[
\mathbb{P}\left( \frac{Z^\top \varepsilon}{\sqrt{Z^\top Z \sigma_y^2}} \in A \right) = \mathbb{E} \left[ \mathbb{P}\left( \frac{Z^\top \varepsilon}{\sqrt{Z^\top Z \sigma_y^2}} \in A \, \bigg| \, Z \right) \right],
\]  
for any measurable set \(A\). Since the conditional probability \( \mathbb{P}(\cdot | Z) \) is bounded by \(1\) (i.e., \(0 \leq \mathbb{P}(\cdot | Z) \leq 1\) for all \(Z\)), the Dominated Convergence Theorem (DCT) justifies interchanging the limit and expectation:  
\[
\lim_{N \to \infty} \mathbb{P}\left( \frac{Z^\top \varepsilon}{\sqrt{Z^\top Z \sigma_y^2}} \in A \right) = \mathbb{E} \left[ \lim_{N \to \infty} \mathbb{P}\left( \frac{Z^\top \varepsilon}{\sqrt{Z^\top Z \sigma_y^2}} \in A \, \bigg| \, Z \right) \right] = \Phi(A).
\]  
Thus, the unconditional distribution converges weakly:  
\begin{align}
    \frac{Z^\top \varepsilon}{\sqrt{Z^\top Z \sigma_y^2}} \stackrel{d}{\longrightarrow} \mathcal{N}(0,1).
\label{eq: conv_weak}
\end{align}

For the term (C), we have 
\begin{align*}
\frac{Z^\top P_X \varepsilon}{\sqrt{Z^\top Z \sigma_y^2}} = \frac{Z^\top X (X^\top X)^{-1} X^\top \varepsilon}{\sqrt{Z^\top Z \sigma_y^2}} = \frac{Z^\top X}{\sqrt{Z^\top Z N}} \cdot \sqrt{\frac{N}{\sigma_y^2}} (X^\top X)^{-1} X^\top \varepsilon
\end{align*}
where combining \eqref{eq: thm3_Z^2_low} and \(\frac{Z^\top X}{N} \stackrel{P}{\longrightarrow} 0\), we have
\[
\frac{Z^\top X}{\sqrt{Z^\top Z N}} = \sqrt{\frac{N}{Z^\top Z}} \cdot \frac{Z^\top X}{N} \stackrel{P}{\longrightarrow} 0.
\]
Moreover, from the central limit theorem, we obtain
\[
\sqrt{\frac{N}{\sigma_y^2}} (X^\top X)^{-1} X^\top \varepsilon \stackrel{d}{\longrightarrow} \mathcal{N}(0, V_X^{-1}).
\]
By Slutsky’s Theorem, we finally obtain term (C) converge to 0 in distribution.

Combine the results of terms (A), (B) and (C) above, we get
\[
\sqrt{\frac{Z^\top Z}{\sigma_y^2}} \left( \tilde{\beta}_Z^{(ols)} - \beta_Z \right) \stackrel{d}{\longrightarrow} \mathcal{N}(0, 1).
\]

\textbf{2. Noisy Network Case}

Let \(\hat{Z} = Z - \delta_Z\). The OLS estimator is:  
\[
\tilde{\beta}_Z - \beta_Z = \left( \hat{Z}^\top (I_N - P_X) \hat{Z} \right)^{-1} \hat{Z}^\top (I_N - P_X)(\varepsilon + \delta_Z \beta_Z).
\]  
Normalizing by \(\sqrt{Z^\top Z/ \sigma_y^2}\), we have 
\begin{align*}
    \sqrt{\frac{Z^\top Z}{\sigma_y^2}} \left( \tilde{\beta}_Z - \beta_Z \right) &= \underbrace{\frac{Z^\top Z}{N} \left( \frac{\hat{Z}^\top (I_N - P_X) \hat{Z}}{N} \right)^{-1}}_{(D)} \cdot \underbrace{\frac{\hat{Z}^\top(I_N - P_X) \varepsilon}{\sqrt{Z^\top Z \sigma_y^2}}}_{(E)} \\
    &\quad + \frac{Z^\top Z}{N} \left( \frac{\hat{Z}^\top (I_N - P_X) \hat{Z}}{N} \right)^{-1} \cdot\underbrace{\frac{\hat{Z}^\top(I_N - P_X) \delta_Z \beta_Z}{\sqrt{Z^\top Z \sigma_y^2}}}_{(F)}.
\end{align*}
Then we only need to prove that (D) $\stackrel{P}{\longrightarrow} 1$, (E) $\stackrel{d}{\longrightarrow} \mathcal{N}(0, 1)$ and (F) $\stackrel{P}{\longrightarrow} 0$.

Notice that if \(a_N \to \infty\) and \(a_N^2 /\delta \to 0\) as \(N \to \infty\), Assumption 5 is satisfied.

For the term (D), we notice that 
\begin{align*}
	\frac{\hat{Z}^\top \hat{Z}}{N}(\frac{1}{N}\hat{Z}^\top (I - P_X)\hat{Z})^{-1} &=  \left(\frac{\hat{Z}^\top \hat{Z}}{N} - \frac{\hat{Z}^\top P_X \hat{Z}}{N}\right)(\frac{1}{N}\hat{Z}^\top (I - P_X)\hat{Z})^{-1} \\
	&\quad + \frac{\hat{Z}^\top P_X \hat{Z}}{N}(\frac{1}{N}\hat{Z}^\top (I - P_X)\hat{Z})^{-1}\\
	&= 1 + \frac{\hat{Z}^\top P_X \hat{Z}}{N}(\frac{1}{N}\hat{Z}^\top (I - P_X)\hat{Z})^{-1}.
\end{align*}
From \eqref{eq: thm3_53} and Assumption 5, we know $\frac{1}{N}\hat{Z}^\top P_X \hat{Z} \stackrel{P}{\longrightarrow} 0$. Also from \eqref{eq: thm3_81}, there exists $m > 0$ such that for $N$ large enough, $\left(\frac{1}{N}\hat{Z}^\top (I_N - P_X)\hat{Z}\right)^{-1} \le m^{-1}$ with probability 1. Thus we have $\frac{\hat{Z}^\top P_X \hat{Z}}{N}(\frac{1}{N}\hat{Z}^\top (I - P_X)\hat{Z})^{-1} \stackrel{P}{\longrightarrow} 0$ and $\frac{\hat{Z}^\top \hat{Z}}{N}(\frac{1}{N}\hat{Z}^\top (I - P_X)\hat{Z})^{-1} \stackrel{P}{\longrightarrow} 1$ is proved.

Then we show that $\frac{Z^\top Z}{\hat{Z}^\top \hat{Z}} \stackrel{P}{\longrightarrow} 1$. Notice that 
\begin{align*}
	\frac{\hat{Z}^\top \hat{Z}}{Z^\top Z} &= \frac{Z^\top Z - 2\delta_Z^\top Z + \|\delta_Z\|_2^2}{Z^\top Z}\\
	&= 1 + \frac{\|\delta_Z\|_2^2  - 2\delta_Z^\top Z }{\|Z\|_2^2}\\
	&= 1 + \frac{(\|\delta_Z\|_2^2  - 2\delta_Z^\top Z)/N }{\|Z\|_2^2/N}.
\end{align*}
From \eqref{eq: thm3_part1_3}, $\|\delta_Z\|_2^2/N \stackrel{P}{\longrightarrow} 0$. With \eqref{eq: thm3_46} and \eqref{eq: thm3_74}, $\delta_Z^\top Z/N \stackrel{P}{\longrightarrow} 0$. From \eqref{eq: thm3_Z^2_low}, we know that $\|Z\|_2^2/N$ is lower bounded by a positive constant $m$. Thus we have $\frac{(\|\delta_Z\|_2^2  - 2\delta_Z^\top Z)/N }{\|Z\|_2^2/N} \stackrel{P}{\longrightarrow} 0$ and $\frac{\hat{Z}^\top \hat{Z}}{Z^\top Z} \stackrel{P}{\longrightarrow} 1$. Therefore, the term (D) converges to $1$ in probability.

Then we turn to analyze term (E). Notice that 
\begin{align}
    \frac{\hat{Z}^\top(I_N - P_X) \varepsilon}{\sqrt{Z^\top Z \sigma_y^2}} &= \frac{(Z - \delta_Z)^\top(I_N - P_X) \varepsilon}{\sqrt{Z^\top Z \sigma_y^2}} \nonumber\\
    &= \frac{Z^\top(I_N - P_X) \varepsilon}{\sqrt{Z^\top Z \sigma_y^2}} - \frac{\delta_Z^\top(I_N - P_X) \varepsilon}{\sqrt{Z^\top Z \sigma_y^2}}\nonumber\\
    &= \frac{Z^\top \varepsilon}{\sqrt{Z^\top Z \sigma_y^2}} - \frac{Z^\top P_X \varepsilon}{\sqrt{Z^\top Z \sigma_y^2}} - \frac{\delta_Z^\top(I_N - P_X) \varepsilon}{\sqrt{Z^\top Z \sigma_y^2}}.\label{eq: new_prop3}
\end{align}
From the arguments of the convergence of terms (B) and (C) in the noiseless case, we know that
\[
\frac{Z^\top \varepsilon}{\sqrt{Z^\top Z \sigma_y^2}} - \frac{Z^\top P_X \varepsilon}{\sqrt{Z^\top Z \sigma_y^2}} \stackrel{d}{\longrightarrow} \mathcal{N}(0, 1).
\] 
Consider
\[
\frac{\delta_Z^\top(I_N - P_X) \varepsilon}{\sqrt{Z^\top Z \sigma_y^2}} = \frac{\delta_Z^\top \varepsilon}{\sqrt{Z^\top Z \sigma_y^2}} - \frac{\delta_Z^\top P_X \varepsilon}{\sqrt{Z^\top Z \sigma_y^2}}
\]
where 
\begin{align*}
	\mathbb{E}\left[\left\|\frac{1}{\sqrt{N}}\delta_Z^\top\varepsilon\right\|^2\right] \le \frac{\sigma_y^2}{N} \mathbb{E}\left[\|\delta_Z\|_F^2 \right] = O(\frac{a_N^2}{\delta^2}) \to 0
\end{align*}
and
\begin{align*}
	\mathbb{E}\left[\left\|\frac{1}{\sqrt{N}}\delta_Z^\top P_{X}\varepsilon\right\|^2\right] \le \frac{P \sigma_y^2}{N} \mathbb{E}\left[\|\delta_Z\|_F^2 \right] = O(\frac{a_N^2}{\delta^2}) \to 0,
\end{align*}
with \eqref{eq: thm3_Z^2_low}, we have 
\[
\frac{\delta_Z^\top(I_N - P_X) \varepsilon}{\sqrt{Z^\top Z \sigma_y^2}}  = \frac{\sqrt{N}}{\sqrt{Z^\top Z \sigma_y^2}}\cdot \frac{\delta_Z^\top(I_N - P_X) \varepsilon}{\sqrt{N}} \stackrel{P}{\longrightarrow} 0.
\]
Thus, the term (E) converges to $\mathcal{N}(0, 1)$ in distribution.

For the term (F), with Cauchy-Schwarz inequality, 
\[
\frac{\hat{Z}^\top(I_N - P_X) \delta_Z}{\sqrt{Z^\top Z \sigma_y^2}} = \frac{\sqrt{N}}{\sqrt{Z^\top Z \sigma_y^2}} \cdot \frac{\hat{Z}^\top(I_N - P_X) \delta_Z}{\sqrt{N}},
\]
where 
\[
\left|\frac{\hat{Z}^\top(I_N - P_X) \delta_Z}{\sqrt{N}}\right| \le \left|\frac{\hat{Z}^\top \delta_Z}{\sqrt{N}}\right| + \left|\frac{\hat{Z}^\top P_X \delta_Z}{\sqrt{N}}\right|.
\]
Consider the expectation below
\[
\mathbb{E}\left[ \left|\frac{\hat{Z}^\top \delta_Z}{\sqrt{N}}\right| \right] \le 
\frac{1}{\sqrt{N}} \mathbb{E}\left[ \|\hat{Z}\| \cdot \|\delta_Z\| \right] 
\le \frac{1}{\sqrt{N}} \sqrt{ \mathbb{E}[\|\hat{Z}\|^2] \cdot \mathbb{E}[\|\delta_Z\|^2] }.
\]
Combining \eqref{eq: thm3_part1_3} and \eqref{eq: thm3_80},
\[
\frac{1}{\sqrt{N}} \sqrt{ \mathbb{E}[\|\hat{Z}\|^2] \cdot \mathbb{E}[\|\delta_Z\|^2] } \le \frac{1}{\sqrt{N}} O(a_N) O(\frac{a_N \sqrt{N}}{\delta}) = O(\frac{a_N^2}{\delta}) \to 0.
\]
Also, 
\[
\left|\frac{\hat{Z}^\top P_X \delta_Z}{\sqrt{N}}\right| = \left|\frac{\hat{Z}^\top X}{\sqrt{N}}(\frac{1}{N}X^\top X)^{-1}\frac{1}{N}X^\top  \delta_Z\right|.
\]
Similar to the arguments of \eqref{eq: conv_weak}, we have \(\frac{\hat{Z}^\top X}{\sqrt{N}}\) converges to a normal distribution in distribution.  With \eqref{eq: thm3_part1_3}, 
\[\mathbb{E}\left[\left\|\frac{1}{N}X^\top  \delta_Z\right\|_2^2\right] \le \frac{1}{N^2}P \mathbb{E}\left[\left\|\delta_Z\right\|_2^2\right] = O(\frac{1}{N^2}\cdot \frac{a_N^2 N}{\delta^2}) = O(\frac{1}{N}\cdot \frac{a_N^2}{\delta^2}) \to 0
\]
Together with \eqref{eq: prop_X}, we have \(\frac{\hat{Z}^\top P_X \delta_Z}{\sqrt{N}} \stackrel{P}{\longrightarrow} 0\). Finally, we obtain the term (F) converges to 0 in probability.

Combine the results above that term (D) converges to 1 in probability, term (E) converge to the standard normal distribution in distribution, and term (F) converges to 0 in probability, we have
\[
\sqrt{\frac{Z^\top Z}{\sigma_y^2}} \left( \tilde{\beta}_Z - \beta_Z \right) \stackrel{d}{\longrightarrow} \mathcal{N}(0, 1). 
\]  
\hfill $\square$

\section{Simulation results of comparison between RCEF and C-MNetR/CC-MNetR:}
\label{sec: Comp}
In this section, we show the simulation results of both Regression with Community Fixed Effects (RCEF) and C-MNetR/CC-MNetR, which demonstrate that RCEF does not exhibit strong consistency properties under the same conditions when C-MNetR performs great, let alone compared to the much better-performing CC-MNetR.
\begin{itemize}
	\item \textbf{Compared to C-MNetR:} First, we have conducted a detailed comparison in the absence of measurement error between our method, C-MNetR:
	\begin{align*}
		y = X\beta_X + C\beta_C + \varepsilon
	\end{align*}
	and method RCFE: 
	\begin{align*}
		y = X\beta_X + C\beta_C + S\beta_S + \varepsilon
	\end{align*}
	where $S$ is the community label matrix. Results in both Table 4.1 (c) in the manuscript and Figure \ref{Figure: CMNetR} show the consistency of $\hat{\beta}_C^{(ols)}$ with $a_N = N$ in the absence of measurement error. But in Figure \ref{Figure: CFE} with the same condition, $\hat{\beta}_S$ lacks consistency. 
	\item \textbf{Compared to CC-MNetR:} Then, we have conducted a detailed comparison in the absence of measurement error between our method, CC-MNetR:
	\begin{align*}
		y = X\beta_X + Z\beta_Z + \varepsilon
	\end{align*}
	and method RCFE: 
	\begin{align*}
		y = X\beta_X + C\beta_C + S\beta_S + \varepsilon
	\end{align*}
	where $S$ is the community label matrix. Our findings indicate that directly regressing on community labels to obtain community fixed effects does not achieve consistency.
	\begin{figure}
		\centering
		\includegraphics[scale=0.4]{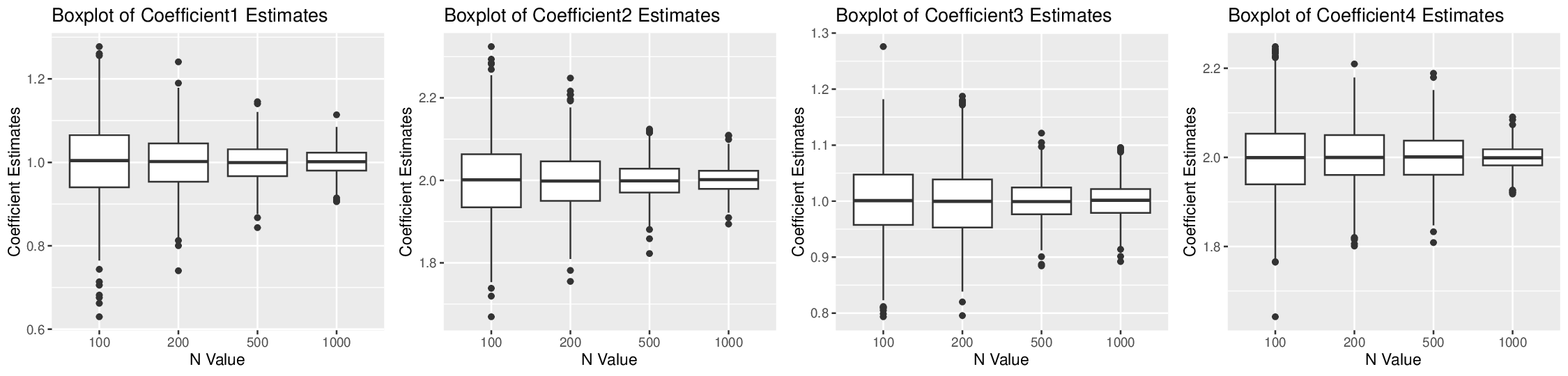}
		\caption{Boxplot of $n = 1000$ Estimates of 4 Coefficients $\beta = (\beta_{X_1}, \beta_{X_2}, \beta_{C_1}, \beta_{C_2})$ in \textbf{C-MNetR} when $a_N = N$ without measurement error. $\hat{\beta}_C$ shows consistency in this case. }
		\label{Figure: CMNetR}
	\end{figure}
	\begin{figure}
		\centering
		\includegraphics[scale=0.4]{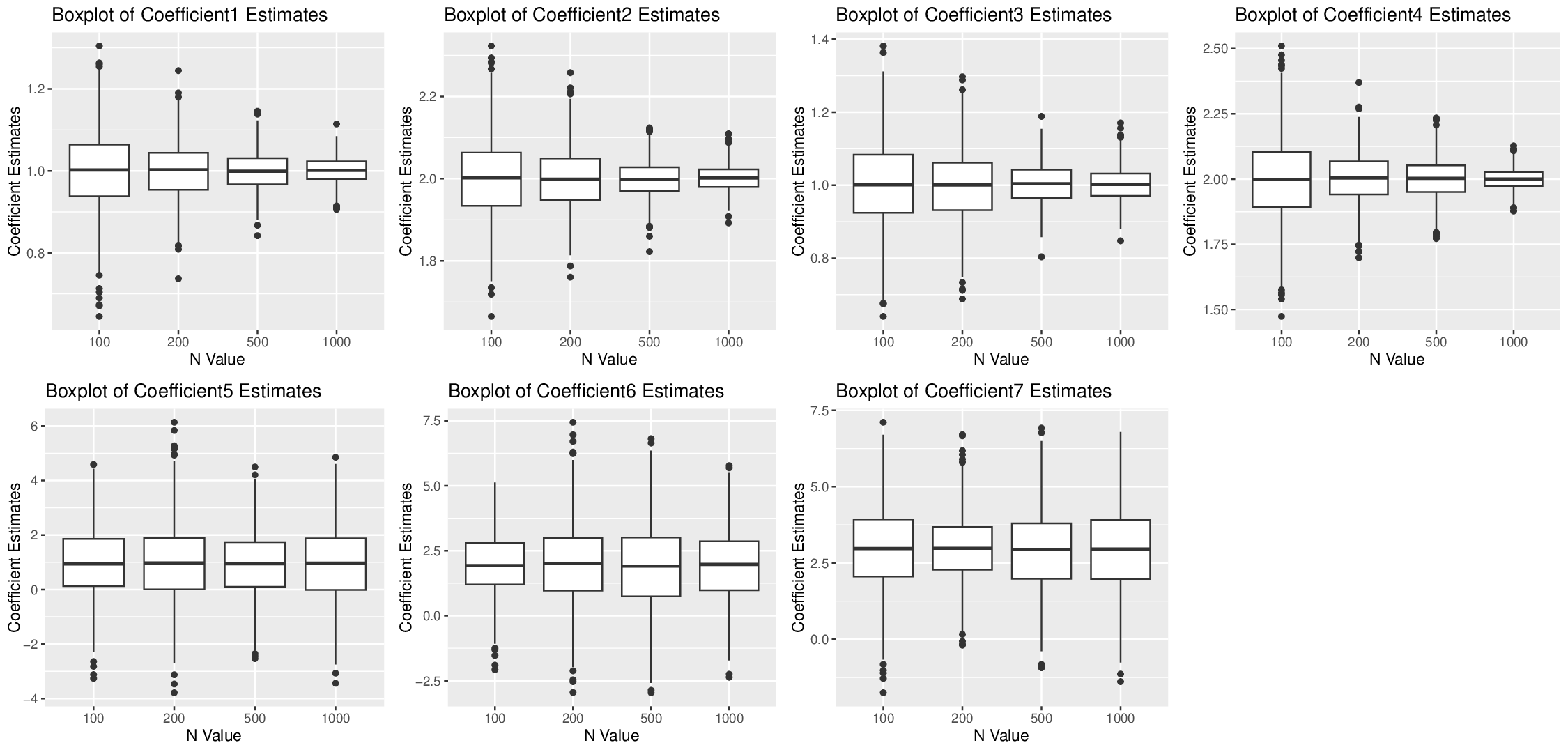}
		\caption{Boxplot of $n = 1000$ Estimates of 7 Coefficients $\beta = (\beta_{X_1}, \beta_{X_2}, \beta_{C_1}, \beta_{C_2}, \beta_{S_1}$, $\beta_{S_2}$, $\beta_{S_3})$ in \textbf{RCFE} when $a_N = N$ without measurement error. Even though the order of $a_N$ is already large, $\hat{\beta}_S$  still lacks consistency. }
		\label{Figure: CFE}
	\end{figure}
	
	\begin{figure}
		\centering
		\includegraphics[scale=0.4]{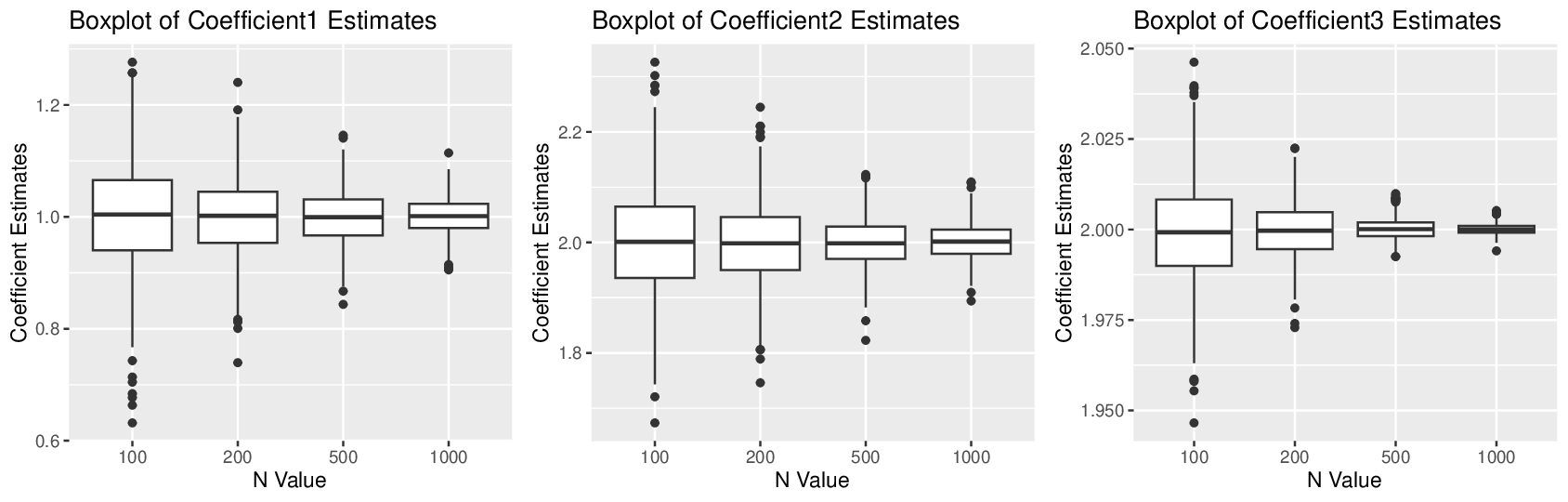}
		\caption{Boxplot of $n = 1000$ Estimates of 3 Coefficients $\beta = (\beta_{X_1}, \beta_{X_2}, \beta_Z)$ in \textbf{CC-MNetR} when $a_N = N$ without measurement error. Compared to $\hat{\beta}_S$, $\hat{\beta}_Z$ shows better consistency. }
		\label{Figure: CCMNetR}
	\end{figure}	
	Specifically, as is shown in Figure \ref{Figure: CFE}, the standard deviation of the estimators does not decrease with increasing $N$. In contrast, as shown in Figure \ref{Figure: CCMNetR}, simulation results of $\hat{\beta}_Z^{(ols)}$ show consistency properties since our CC-MNetR method addresses the inconsistency of $\hat{\beta}_C^{(ols)}$ by incorporating restricted community structure into the centrality measure. 
\end{itemize}
\newpage
\section{QQ-plots and MSE plots of coefficients of C-MNetR/CC-MNetR in Simulation part:}
\label{sec: mse}
In this section, we present the QQ-plots of $\hat{\beta}^{(ols)}$, $\tilde{\beta}^{(ols)}$ and mean-squared error plots of $\hat{\beta}^{(ols)}$, $\tilde{\beta}^{(ols)}$, $\hat{\beta}$, and $\tilde{\beta}$.
\subsection{QQ-plots}
\label{sec: qqplots}
Figure~\ref{fig:qqplot_thm1} and Figure~\ref{fig:qqplot_thm2} shows the QQ-plots of corresponding estimators.
\begin{figure}[h]
	\centering
	\includegraphics[width = \textwidth]{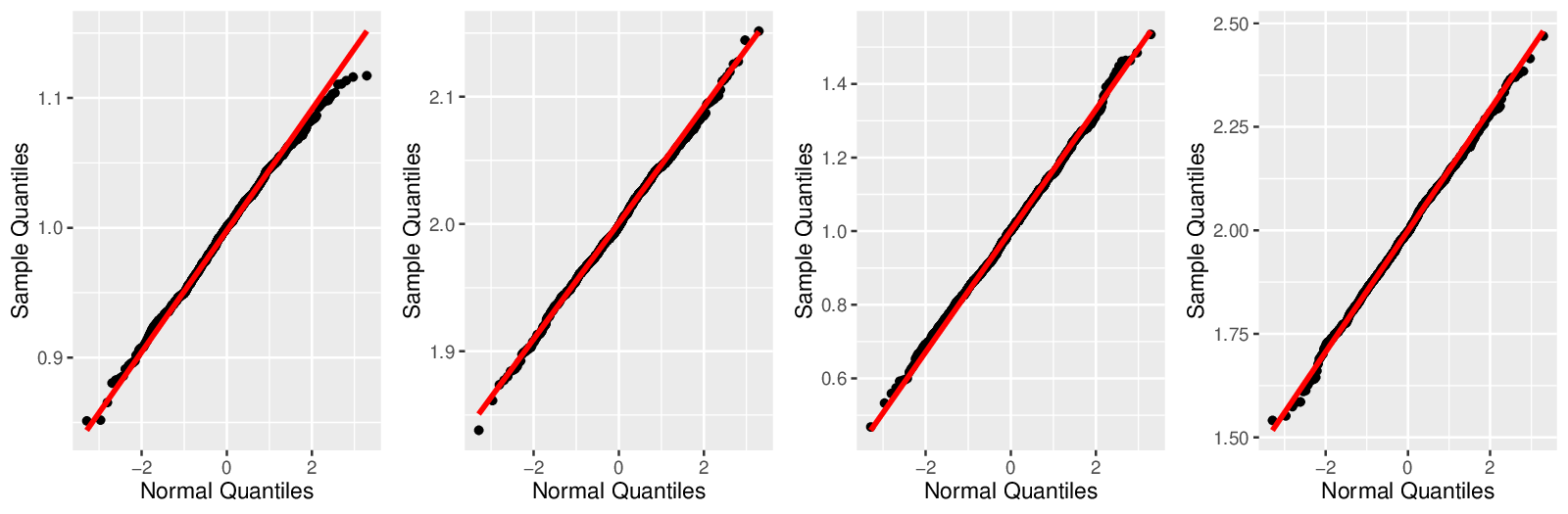}
	\caption{QQ-plots of $\hat{\beta}^{(ols)} = (\hat{\beta}_{X_1}^{(ols)},\hat{\beta}_{X_2}^{(ols)},\hat{\beta}_{C_1}^{(ols)},\hat{\beta}_{C_2}^{(ols)})$ (left to right) for $N = 500$, $a_N = N^{0.8}$, and $\beta = (1,2,1,2)^\top$.}
	\label{fig:qqplot_thm1}
\end{figure}

\begin{figure}[h]
	\centering
	\includegraphics[width = \textwidth]{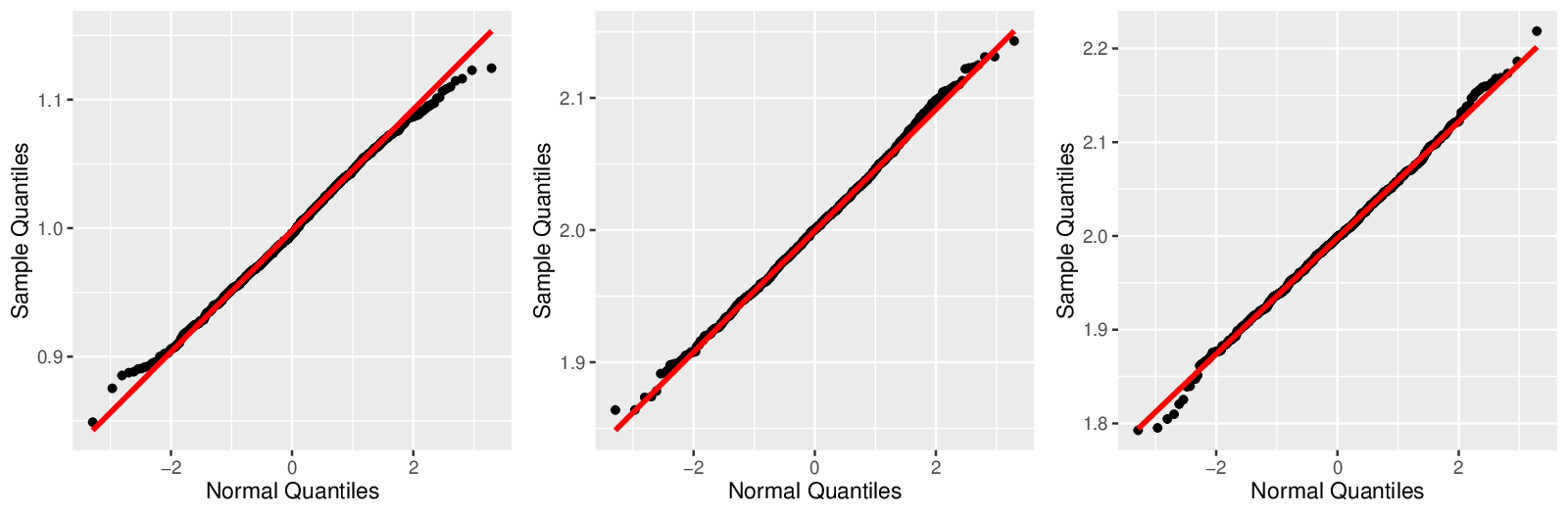}
	\caption{QQ-plots of $\tilde{\beta}^{(ols)} = (\tilde{\beta}_{X_1}^{(ols)},\tilde{\beta}_{X_2}^{(ols)},\tilde{\beta}_{Z}^{(ols)})$ (left to right), for $N = 500$, $a_N = \sqrt{N}$, and $\beta = (1,2,2)^\top$.}
	\label{fig:qqplot_thm2}
\end{figure}

\subsection{The order of smallest singular values of $(I- P_X)V$}
\label{sec: order}

Figure~\ref{Figure: order} inspects the smallest singular values for different configurations of the connection matrices in a SBM. The light blue line represents the scenario where the connection probability matrices of the two layers are identical as $P = 
\left[\begin{array}{ccc}
	0.8 & 0.1 & 0.1  \\
	0.1 & 0.8 & 0.1  \\
	0.1 & 0.1 & 0.8  
\end{array}\right]$. The blue line corresponds to the scenario where the connection probability matrices of the two layers are different as $P_1 = 
\left[\begin{array}{ccc}
	0.8 & 0.1 & 0.1  \\
	0.1 & 0.8 & 0.1  \\
	0.1 & 0.1 & 0.8  
\end{array}\right]$ for Layer 1 and $P_2 = 
\left[\begin{array}{ccc}
	0.5 & 0.25 & 0.25  \\
	0.25 & 0.5 & 0.25  \\
	0.25 & 0.25 & 0.5  
\end{array}\right]$ for Layer 2. The red line indicates line $y = N^{-1/2}$.
From Figure~\ref{Figure: order}, when $a_N = \sqrt{N}$, the multiplex network configuration we used results in $V$ satisfying $l_N = O(N^{-1/2})$ and $a_N l_N = O(1)$. Thus Condition~(3.20) is not fulfilled and the consistency of $\hat{\beta}_C^{(ols)}$ cannot be guaranteed.
\begin{figure}[h]
	\centering
	\includegraphics[width = \textwidth]{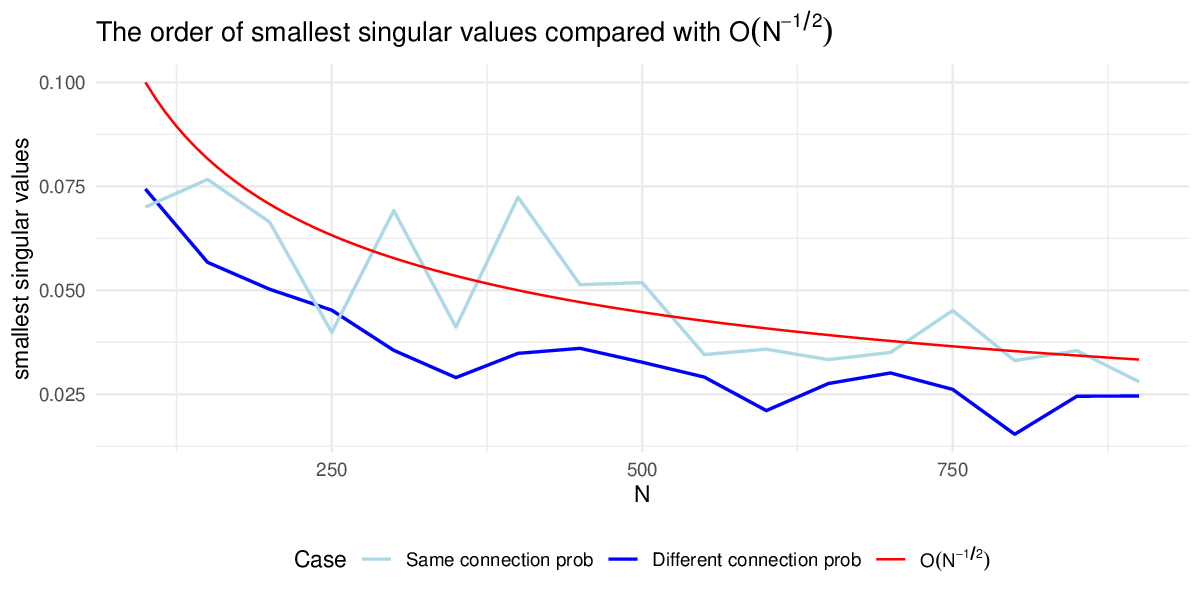}
	\caption{Smallest Singular Values for Connection Probability Matrices in Stochastic Block Model (SBM) Settings.}
	\label{Figure: order}
\end{figure}
\subsection{MSE plots}
\label{sec: sub_MSE}
This subsection presents the Mean Squared Error (MSE) of coefficient estimates for the C-MNetR and CC-MNetR models under different scenarios. The results are divided into two cases: (1) in the absence of measurement error and (2) with measurement error. For each case, the impact of varying the parameter \(a_N\) on the consistency and bias of the estimators is analyzed.  

In the absence of measurement error, Figure~\ref{Figure: CMNetR_no_ME} shows the MSE of the C-MNetR model for different values of \(a_N\). When \(a_N = \sqrt{N}\), the MSE of \(\hat{\beta}_C^{(ols)}\) does not decrease as \(N\) increases, but consistency is achieved with higher orders of \(a_N\). In contrast, Figure~\ref{Figure: CCMNetR_no_ME} demonstrates that the CC-MNetR model achieves consistency regardless of the order of \(a_N\).  

When measurement error is introduced, Figure~\ref{Figure: CMNetR_ME} reveals that the C-MNetR estimator \(\hat{\beta}_C\) remains biased for all values of \(a_N\). However, as shown in Figure~\ref{Figure: CCMNetR_ME}, the CC-MNetR estimator \(\tilde{\beta}\) maintains consistency even in the presence of measurement error, regardless of the order of \(a_N\).  

These results highlight the robustness of the CC-MNetR model compared to the C-MNetR model, particularly in scenarios involving measurement error.

\begin{figure}[h]
	\centering
	\begin{subfigure}[b]{0.48\textwidth}
		\includegraphics[width=\textwidth]{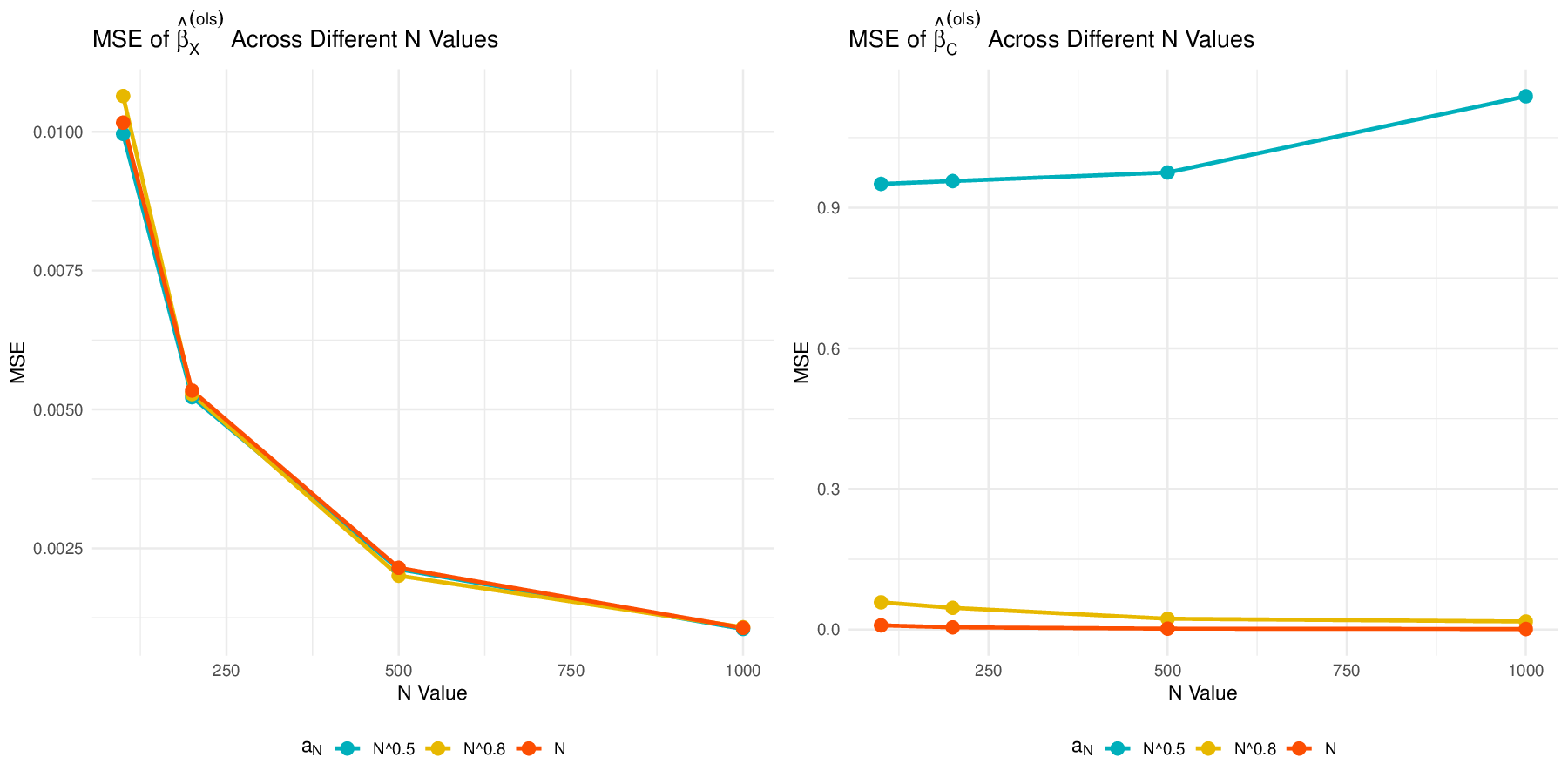}
		\caption{MSE of 1000 estimates of coefficients of C-MNetR in the absence of measurement error under different values of $a_N$. When $a_N = \sqrt{N}$, the MSE of $\hat{\beta}_C^{(ols)}$ does not decrease as $N$ increases, but consistency is achieved with higher orders of $a_N$.}
		\label{Figure: CMNetR_no_ME}
	\end{subfigure}
	\hfill
	\begin{subfigure}[b]{0.48\textwidth}
		\includegraphics[width=\textwidth]{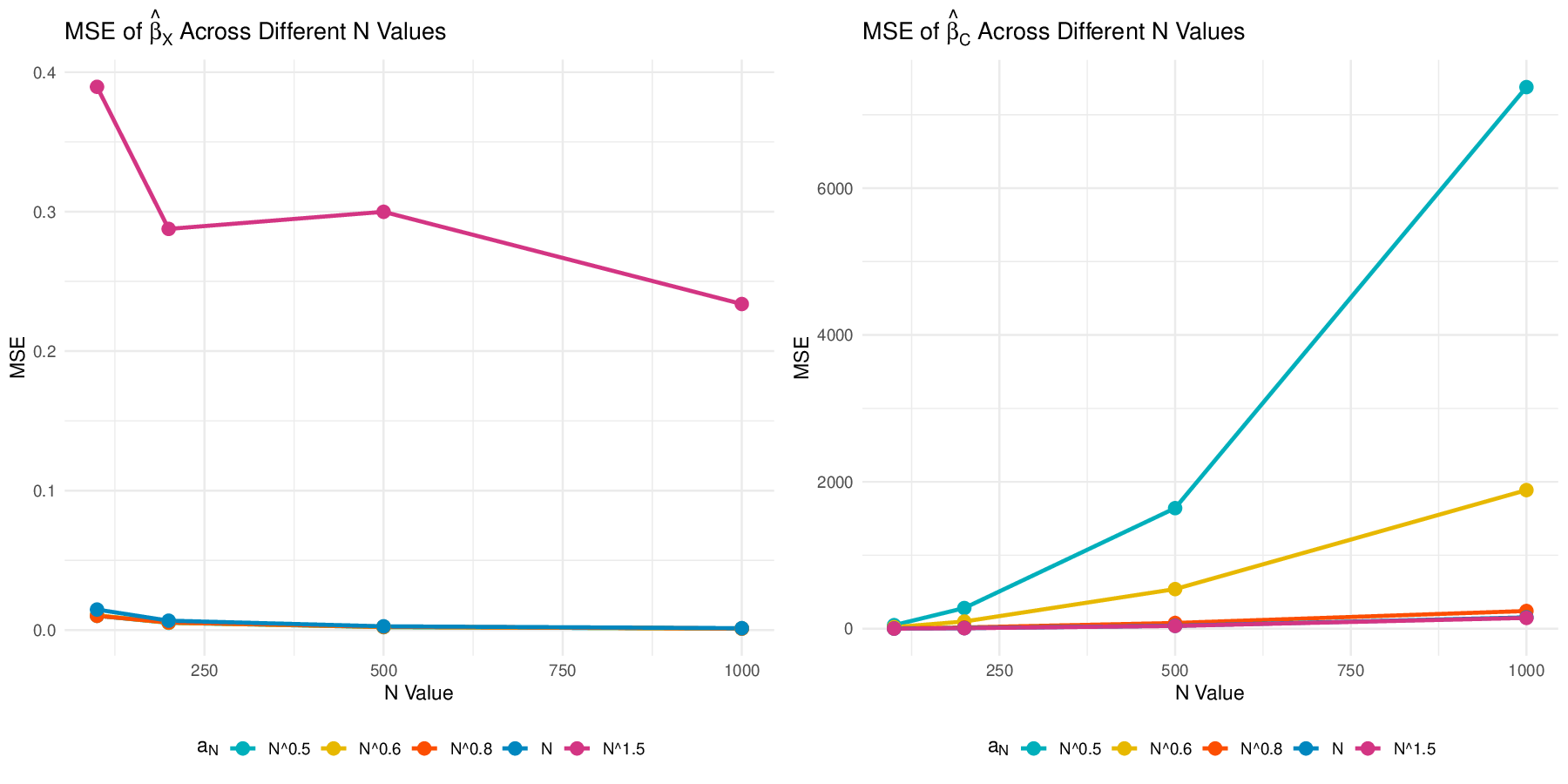}
		\caption{MSE of 1000 estimates of coefficients of C-MNetR with measurement error under different values of $a_N$. Regardless of $a_N$'s order, $\hat{\beta}_C$ remains biased.}
		\label{Figure: CMNetR_ME}
	\end{subfigure}
	\vskip\baselineskip
	\begin{subfigure}[b]{0.48\textwidth}
		\includegraphics[width=\textwidth]{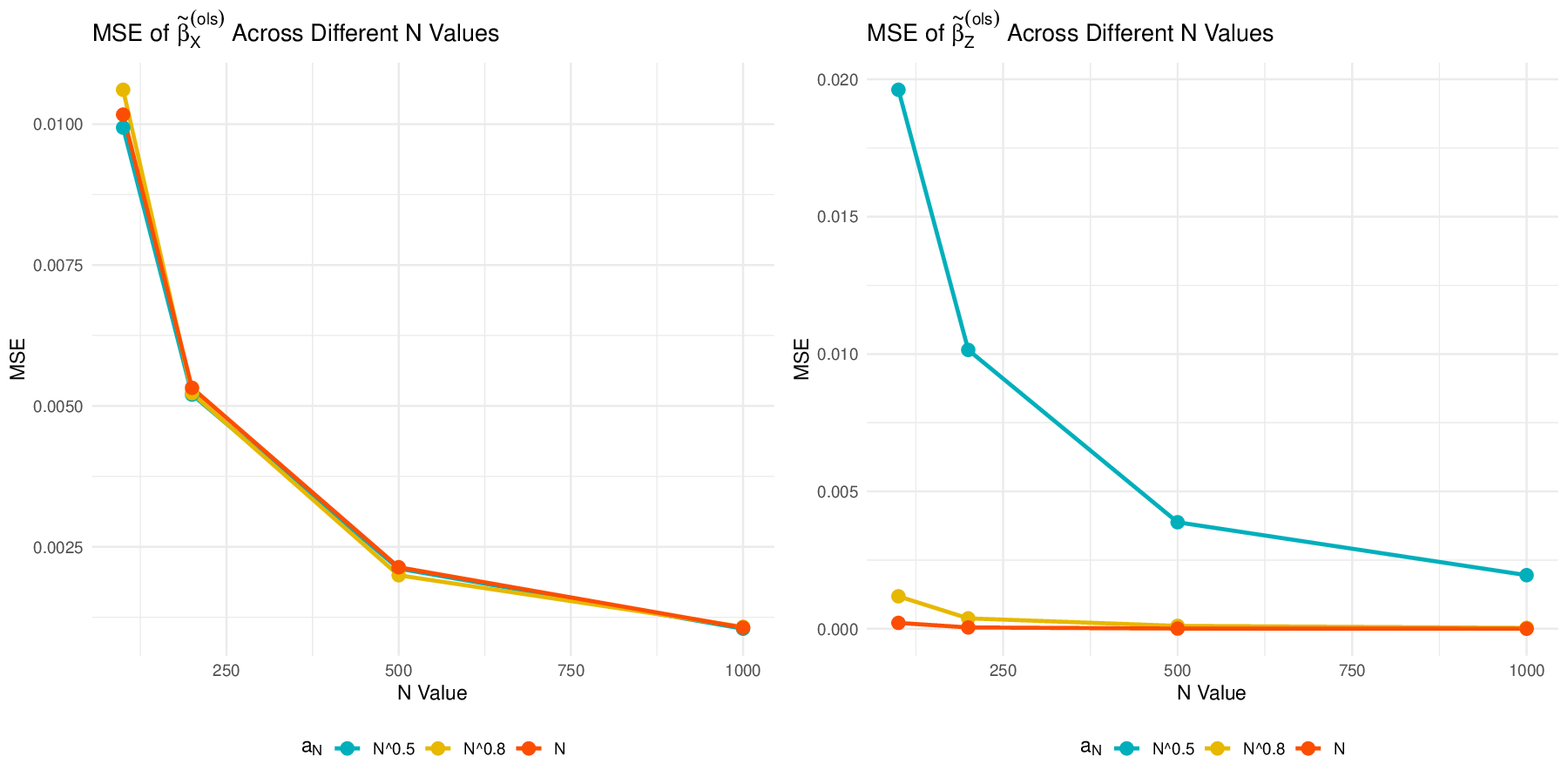}
		\caption{MSE of 1000 estimates of coefficients of CC-MNetR in the absence of measurement error under different values of $a_N$. $\tilde{\beta}^{(ols)}$ shows consistency regardless of $a_N$'s order.}
		\label{Figure: CCMNetR_no_ME}
	\end{subfigure}
	\hfill
	\begin{subfigure}[b]{0.48\textwidth}
		\includegraphics[width=\textwidth]{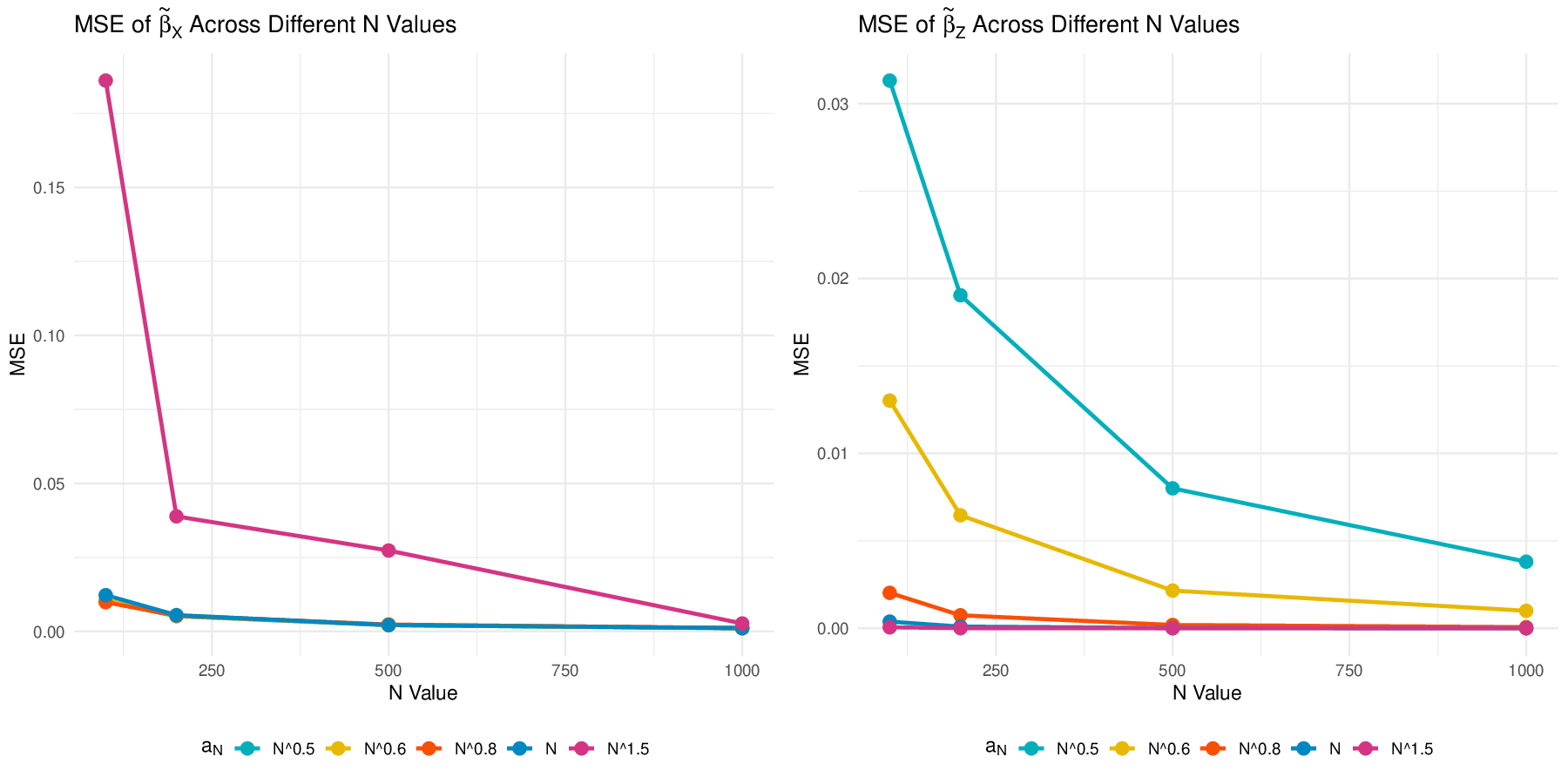}
		\caption{MSE of 1000 estimates of coefficients of CC-MNetR with measurement error under different values of $a_N$. $\tilde{\beta}$ shows consistency regardless of $a_N$'s order.}
		\label{Figure: CCMNetR_ME}
	\end{subfigure}
	\caption{Comparison of MSE for C-MNetR and CC-MNetR models with and without measurement error under different values of $a_N$.}
	\label{Figure: MSE_Comparison}
\end{figure}

\section{WIOD data description and analysis result}
\subsection{Centrality comparison}
Figure~\ref{fig: ctr_07_14} shows the difference of community-based centrality $Z$ in 2007 and 2014.
\begin{figure}[H]
	\centering
	\includegraphics[width = \textwidth]{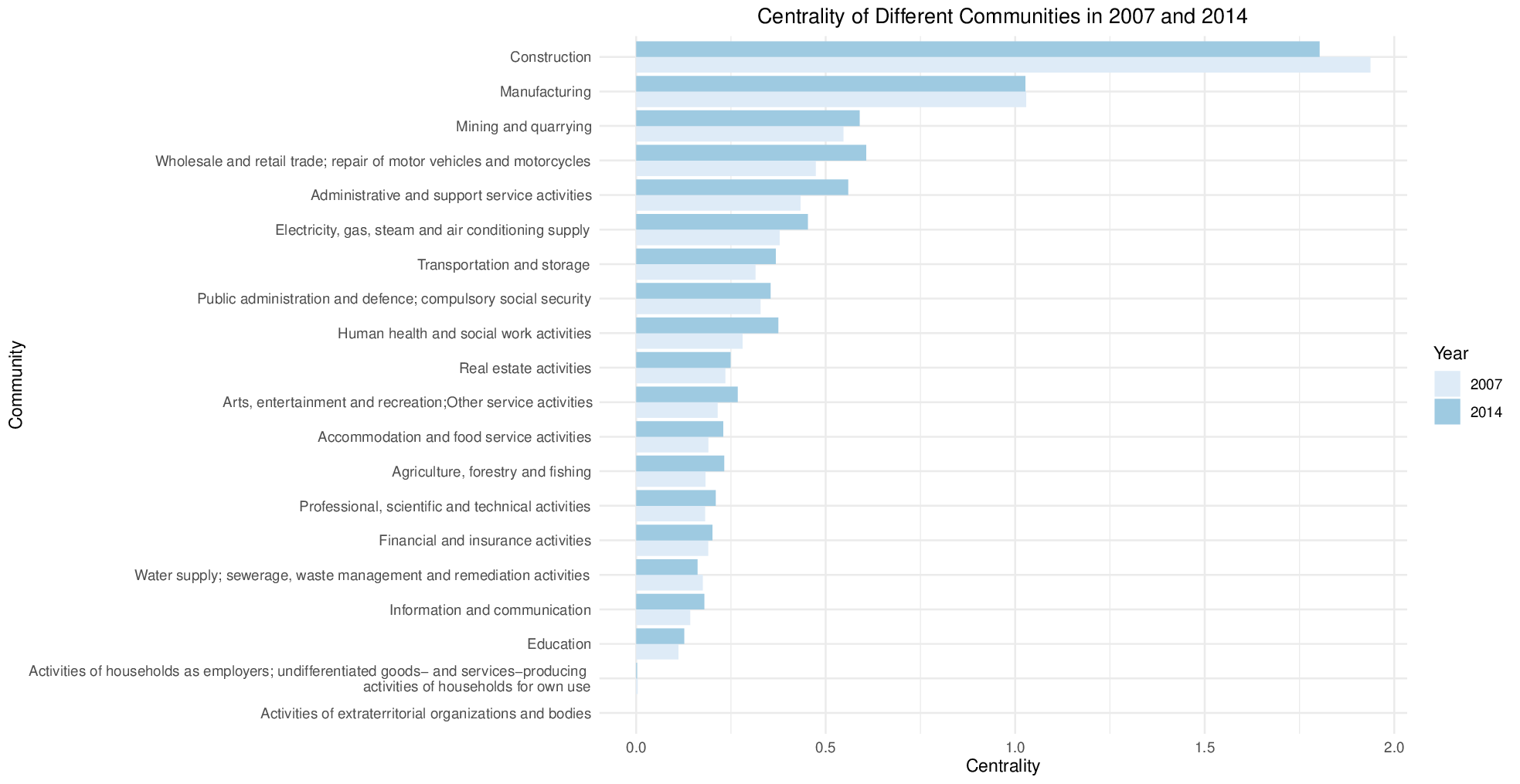}
	\caption{Changes in the centrality of different communities corresponding to industries from 2007 to 2014.}
	\label{fig: ctr_07_14}
\end{figure}
\subsection{Details of variables in SEA dataset:}
The variable details of the SEA dataset are summarized in Table~\ref{tab: sea}, and their relationships are visualized in the scatterplot (Figure~\ref{fig: scatterplot}).
\label{sec: var}
\begin{table}[h!]
	\centering
	\begin{tabular}{l|l}
		\textbf{Values} &	\textbf{Description}\\
		\hline GO&Gross output by industry at current basic prices (in millions of national currency)\\
		II&Intermediate inputs at current purchasers' prices (in millions of national currency)\\
		VA&Gross value added at current basic prices (in millions of national currency)\\
		EMP&Number of persons engaged (thousands)\\
		EMPE&Number of employees (thousands)\\
		H\_EMPE&Total hours worked by employees (millions)\\
		COMP&Compensation of employees (in millions of national currency)\\
		LAB&Labour compensation (in millions of national currency)\\
		CAP&Capital compensation (in millions of national currency)\\
		K&Nominal capital stock (in millions of national currency)\\
	\end{tabular}
	\caption{Descriptions of 10 variables contained in SEA.}
	\label{tab: sea}
\end{table}

\begin{figure}[h]
	\centering
	\includegraphics[width = \textwidth]{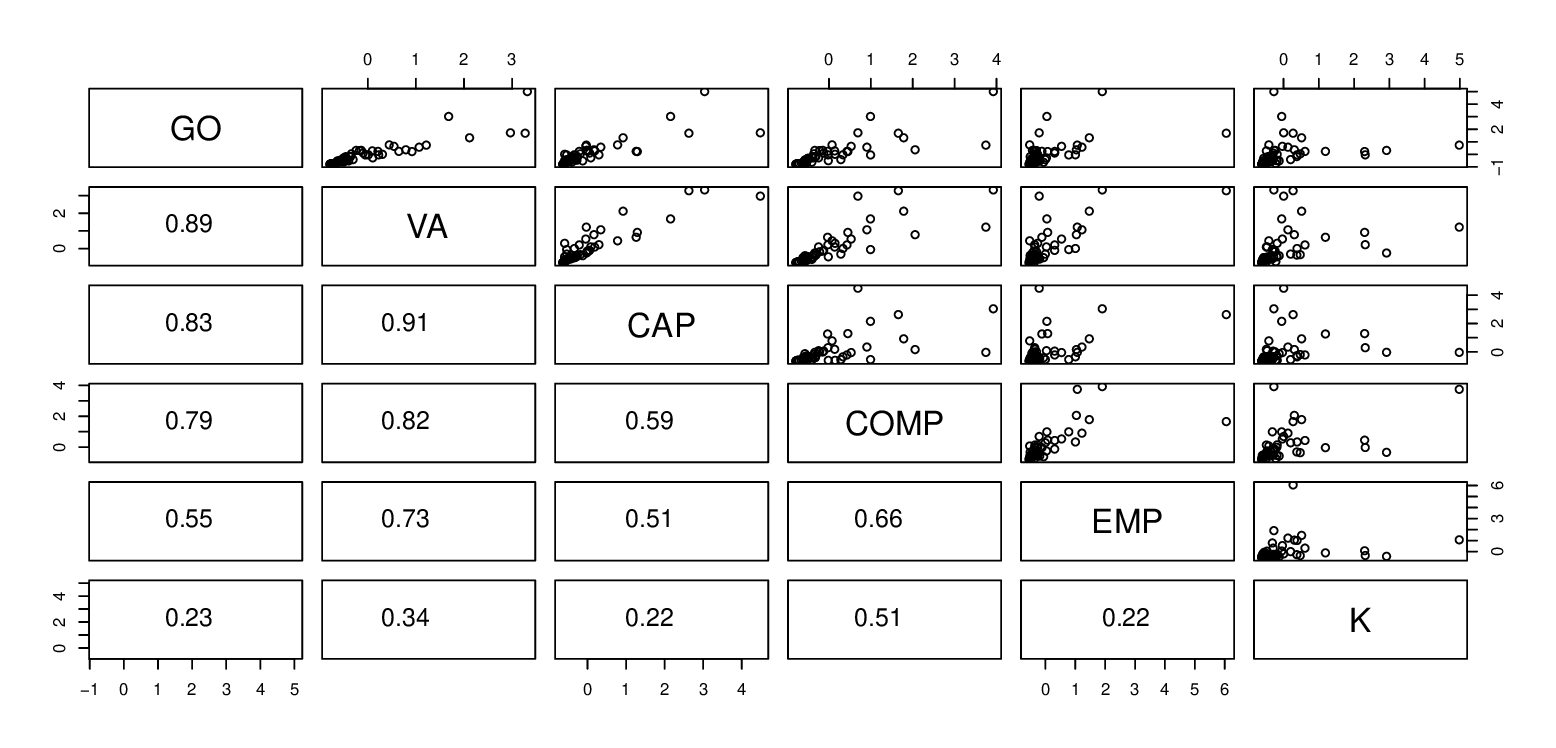}
	\caption{Scatterplot matrix of 6 variables and the lower panel of this matrix denotes the correlation coefficient between variables.}
	\label{fig: scatterplot}
\end{figure}

\subsection{Details of sectors:}

According to ISIC Rev.4, industries in WIOD release 2016 are shown in Table~\ref{tab:sector_wiod2016}.
\begin{longtable}{p{0.03\textwidth}p{0.11\textwidth}p{0.4\textwidth}p{0.4\textwidth}}
		\textbf{No.} & \textbf{Industry} & \textbf{Description}& \textbf{Community}\\
		\hline 
		\endhead                                                                                                                                          
		1 & A01 & Crop and animal production, hunting and related service activities & Agriculture, forestry and fishing \\
		2 & A02 & Forestry and logging & Agriculture, forestry and fishing\\
		3 & A03 & Fishing and aquaculture & Agriculture, forestry and fishing \\
		4 & B & Mining and quarrying & Mining and quarrying\\
		5 & C10-C12 & Manufacture of food products, beverages and tobacco products & Manufacturing\\
		6 & C13-C15 & Manufacture of textiles, wearing apparel and leather products & Manufacturing\\
		7 & C16 & Manufacture of wood and of products of wood and cork, except furniture; etc. & Manufacturing\\
		8 & C17 & Manufacture of paper and paper products & Manufacturing\\
		9 & C18 & Printing and reproduction of recorded media & Manufacturing\\
		10 & C19 & Manufacture of coke and refined petroleum products & Manufacturing\\
		11 & C20 & Manufacture of chemicals and chemical products & Manufacturing\\
		12 & C21 & Manufacture of basic pharmaceutical products and pharmaceutical preparations & Manufacturing\\
		13 & C22 & Manufacture of rubber and plastic products & Manufacturing\\
		14 & C23 & Manufacture of other non-metallic mineral products & Manufacturing\\
		15 & C24 & Manufacture of basic metals & Manufacturing\\
		16 & C25 & Manufacture of fabricated metal products, except machinery and equipment & Manufacturing\\
		17 & C26 & Manufacture of computer, electronic and optical products & Manufacturing\\
		18 & C27 & Manufacture of electrical equipment & Manufacturing\\
		19 & C28 & Manufacture of machinery and equipment n.e.c. & Manufacturing\\
		20 & C29 & Manufacture of motor vehicles, trailers and semi-trailers & Manufacturing\\
		21 & C30 & Manufacture of other transport equipment & Manufacturing\\
		22 & C31-C32 & Manufacture of furniture; other manufacturing & Manufacturing\\
		23 & C33 & Repair and installation of machinery and equipment & Manufacturing\\
		24 & D & Electricity, gas, steam and air conditioning supply & Electricity, gas, steam and air conditioning supply\\
		25 & E36 & Water collection, treatment and supply & Water supply; sewerage, waste management and remediation activities\\
		26 & E37-E39 & Sewerage; waste collection, treatment and disposal activities; materials recovery; etc. & Water supply; sewerage, waste management and remediation activities\\
		27 & F & Construction & Construction\\
		28 & G45 & Wholesale and retail trade and repair of motor vehicles and motorcycles & Wholesale and retail trade; repair of motor vehicles and motorcycles\\
		29 & G46 & Wholesale trade, except of motor vehicles and motorcycles & Wholesale and retail trade; repair of motor vehicles and motorcycles \\
		30 & G47 & Retail trade, except of motor vehicles and motorcycles & Wholesale and retail trade; repair of motor vehicles and motorcycles\\
		31 & H49 & Land transport and transport via pipelines & Transportation and storage\\
		32 & H50 & Water transport & Transportation and storage\\
		33 & H51 & Air transport & Transportation and storage\\
		34 & H52 & Warehousing and support activities for transportation & Transportation and storage\\
		35 & H53 & Postal and courier activities & Transportation and storage\\
		36 & I & Accommodation and food service activities & Accommodation and food service activities\\
		37 & J58 & Publishing activities & Information and communication\\
		38 & J59-J60 & Motion picture, video and television program production, sound recording and music publishing activities; etc. & Information and communication\\
		39 & J61 & Telecommunications & Information and communication\\
		40 & J62-J63 & Computer programming, consultancy and related activities; information service activities & Information and communication\\
		41 & K64 & Financial service activities, except insurance and pension funding & Financial and insurance activities\\
		42 & K65 & Insurance, reinsurance and pension funding, except compulsory social security & Financial and insurance activities\\
		43 & K66 & Activities auxiliary to financial services and insurance activities & Financial and insurance activities\\
		44 & L & Real estate activities & Real estate activities\\
		45 & M69-M70 & Legal and accounting activities; activities of head offices; management consultancy activities & Professional, scientific and technical activities\\
		46 & M71 & Architectural and engineering activities; technical testing and analysis & Professional, scientific and technical activities\\
		47 & M72 & Scientific research and development & Professional, scientific and technical activities\\
		48 & M73 & Advertising and market research & Professional, scientific and technical activities\\
		49 & M74-M75 & Other professional, scientific and technical activities; veterinary activities & Professional, scientific and technical activities\\
		50 & N & Rental and leasing activities, Employment activities, Travel services, security and services to buildings & Administrative and support service activities\\
		51 & O & Public administration and defence; compulsory social security & Public administration and defence; compulsory social security\\
		52 & P & Education & Education\\
		53 & Q & Human health and social work activities & Human health and social work activities\\
		54 & R-S & Creative, Arts, Sports, Recreation and entertainment activities and all other personal service activities & Arts, entertainment and recreation; Other service activities\\
		55 & T & Activities of households as employers; undifferentiated goods- and services-producing activities of households for own use & Activities of households as employers; undifferentiated goods- and services-producing activities of households for own use\\
		56 & U & Activities of extra-territorial organizations and bodies & Activities of extraterritorial organizations and bodies\\
		\hline
		\caption{56 sectors and their corresponding communities in WIOD release 2016}
		\label{tab:sector_wiod2016}
	\end{longtable}

\subsection{Regression results}
Table~\ref{tab:tab5.3} shows the estimated results of regression models for 2007 and 2014 WIOD tables.
\begin{table}[h]
	\centering
	\begin{subtable}[t]{\textwidth}
		\centering
		\begin{tabular}{lccccc}
			\hline
			Variable & Estimate & Std Error & F value & p-value & Significance  \\
			\hline
			Z & 0.7018 & 0.1136 & 84.68 & 2.03e-12 & \textbf{***} \\
			VA & 0.8905 & 0.0690 & 334.04 & $<$2.2e-16 & \textbf{***} \\
			EMP & -0.0648 & 0.0625 & 1.14 & 0.2915 &  \\
			K & 0.0057 & 0.0528 & 0.01 & 0.9141 &  \\
			Intercept & -0.3778 & 0.0763 & & & \\
			\hline
		\end{tabular}
		\caption{Results for 2007}
		\label{tab:5.3a}
	\end{subtable}
        \bigskip
	\begin{subtable}[t]{\textwidth}
		\centering
		\begin{tabular}{lccccc}
			\hline
			Variable & Estimate & Std Error & F value & p-value & Significance  \\
			\hline
			Z & 0.6112 & 0.1319 & 68.08 & 5.90e-11 & \textbf{**} \\
			VA & 0.9517 & 0.0811 & 269.99 & $<$2.2e-16 & \textbf{***} \\
			EMP & -0.1385 & 0.0759 & 3.26 & 0.0769 & \textbf{$\cdot$} \\
			K & -0.0155 & 0.0555 & 0.08 & 0.7809 &  \\
			Intercept & -0.3446 & 0.0896 & & & \\
			\hline
		\end{tabular}
		\caption{Results for 2014}
		\label{tab:5.3b}
	\end{subtable}
	\caption{Estimated coefficients, standard errors, and ANOVA results for 2007 and 2014. Significance levels: \textbf{***} (p $<$ 0.001), \textbf{**} (p $<$ 0.01), \textbf{*} (p $<$ 0.05), \textbf{$\cdot$} (p $<$ 0.1).}
	\label{tab:tab5.3}
\end{table}

  \bibliographystyle{chicago}      
\bibliography{Netreg}   